\documentclass[final,12pt]{clear2025} 

\title[Median-based Splitting Rules for Causal Trees and Forests]{Median-based Splitting Rules for Causal Trees and Forests}
\usepackage{times}
\usepackage{algorithm} 
\usepackage{multirow}
\usepackage{dsfont} 
\usepackage{mathtools} 
\usepackage{physics} 
\usepackage{booktabs} 
\usepackage{bm} 
\input{ee.sty} 
\usepackage{float}
\usepackage{comment}
\usepackage{thmtools}
\usepackage{enumitem}
\usepackage{array}
\usepackage{rotating}

\newenvironment{notes}
  {\begin{minipage}{\linewidth}\smallskip\footnotesize\emph{Notes:}}
  {\end{minipage}}

\clearauthor{%
 \Name{Lennard Maßmann}
 \Email{lennard.massmann@uni-due.de}\\
 \addr Chair of Econometrics, Faculty of Business Administration and Economics, University of Duisburg-Essen, Universitätsstraße 12, 45117 Essen, Germany \\
 \addr Ruhr Graduate School in Economics (RGS Econ), Research Academy Ruhr, Universitätsstr. 150, 44801 Bochum, Germany
 \AND
 \Name{Karolina Gliszczy\'nska-Schroeder} \Email{karolina.gliszczynska@vwl.uni-due.de}\\
 \addr Chair of Econometrics, Faculty of Business Administration and Economics, University of Duisburg-Essen, Universitätsstraße 12, 45117 Essen, Germany%
}

\begin{document}


\newtheorem{assump}{Assumption}


\maketitle


\begin{abstract}%

Heavy-tailed and skewed outcomes are common in the randomized experiments and observational studies used to estimate heterogeneous treatment effects, yet the mean-squared-error criterion that guides splitting in honest causal trees is sensitive to the extreme values they generate.
Building on the causal forest framework \citep{athey_recursive_2016, wager_estimation_2018}, we introduce the Median Squared Deviation (MSD) criterion, which replaces the leafwise difference in means in the honest splitting objective with the Hodges--Lehmann location estimator while leaving honest leaf estimation and forest inference unchanged. Two further median-based rules, the Median Absolute Deviation (MAD) and the Least Median of Squares (LMS), serve as robust baselines. We evaluate the criteria in a simulation study covering precision, bias, and confidence interval coverage. MSD restricts its robustness to split selection and lowers the error of conditional average treatment effect estimates under heavy-tailed and skewed outcomes. Further, we re-visit two empirical applications: the first analyzes the electoral effects of a Mexican conditional cash transfer program on precinct-level observations, while the second application studies antiretroviral treatments in HIV-positive adults. 
\end{abstract}

\begin{keywords}%
    Heterogeneous treatment effects, causal forests, Hodges--Lehmann estimator, robust estimation, heavy-tailed outcomes
\end{keywords}

\begin{jelclass}
    C14, C21, C52
\end{jelclass}


\section{Introduction}
\label{sec:Introduction}

Extreme and heavy-tailed outcomes pose a recurring challenge for estimating treatment effects in randomized trials and observational studies in economic applications \citep{li_robust_2021, li_robust_2023, Athey2023, ghosh_robustness_2026}. A natural and established remedy is to base inference on the ranks of the outcomes rather than their magnitudes \citep{lehmann_nonparametrics_1975} within the potential outcomes framework \citep{imbens_causal_2015}. Building on this idea, \cite{Rosenbaum1993} combines the location estimator of \cite{hodges_estimates_1963} with rank statistics to estimate a constant additive treatment effect. \cite{ghosh_robustness_2026} establish an asymptotic theory for this rank-based estimator.
Although the constant treatment effect assumption serves as a convenient and widely used starting point for causal inference \citep{rosenbaum_covariance_2002, Athey2023}, many applications in economics, social sciences, and public health exhibit treatment effect heterogeneity \citep{kunzel_meta-learners_2019, yao_survey_2021, kennedy_towards_2023, cinelli_challenges_2025}. Motivated by the goal of combining nonparametric heterogeneous treatment effect (HTE) estimation with robustness to heavy-tailed or skewed outcomes, this paper extends the causal forest framework \citep{athey_recursive_2016, wager_estimation_2018} by using the Hodges--Lehmann estimator as the basis for median-based splitting rules in the underlying causal trees. The construction retains the honest sample splitting of \citet{athey_recursive_2016} and \citet{wager_estimation_2018} and we use pointwise confidence intervals based on the bootstrap of little bags (BLB) \citep{sexton_standard_2009, athey_generalized_2019} to quantify uncertainty around the conditional average treatment effect (CATE) estimates.

Nonparametric estimation of HTEs has become a central theme in modern causal inference, motivated by the need to understand how treatment effects vary across subpopulations without imposing restrictive parametric assumptions on the outcome model. In the potential outcomes framework \citep{imbens_causal_2015}, where each unit has unobserved counterfactual responses under treatment and control, a broad class of flexible methods exists, from regression trees and forests to meta-learners and Bayesian nonparametric models \citep{wager_estimation_2018, kunzel_meta-learners_2019, hahn_bayesian_2020}, that adaptively model treatment effects as functions of covariates without specifying a fixed functional form. 
These approaches are useful in settings where the true relationship between covariates, outcomes, and treatment effects is complex and high-dimensional. While much of the HTE literature focuses on smooth or well-behaved outcome distributions, heavy-tailed or irregular potential outcomes are prevalent in real-world applications. 
\cite{li_robust_2021, li_robust_2023} analyze electronic health record data to estimate and compare HTEs between two alternative therapies for hypertension with a right-skewed outcome variable, flexibly adjusting for a set of demographic and clinical covariates. 
\cite{Athey2023} study two applications using skewed and heavy-tailed house price data and medical expenditures for pneumonia patients as outcome variables. In this paper, we revisit the conditional cash transfer experiment of \cite{de_la_o_conditional_2013}, also studied by \cite{ghosh_robustness_2026}. This randomized experiment evaluates the electoral impact of Mexico's conditional cash transfer program on skewed political support outcomes. Additionally, we analyze the ACTG 175 antiretroviral trial of \cite{hammer_trial_1996}, which was also used as an empirical application by \cite{leqi_median_2022}. For both applications, we examine potential treatment effect heterogeneity while accounting for their skewed outcomes.

Causal trees and their aggregation into causal forests provide a flexible, nonparametric framework for estimating the HTE by adapting recursive partitioning to target CATE rather than prediction error. \cite{athey_recursive_2016} and \cite{wager_estimation_2018} formalize an honest estimation approach by separating tree construction from leaf-level effect estimation to reduce adaptive bias and enable asymptotically valid inference, while forest aggregation and subsampling stabilize estimates in high-dimensional settings. 
Subsequent work connects causal forests to semiparametric efficiency theory through orthogonalized and doubly robust score constructions \citep{chernozhukov_doubledebiased_2018, kennedy_towards_2023}, thereby mitigating bias from nuisance function estimation and improving robustness to model misspecification. Nonetheless, practical performance can deteriorate under heavy-tailed or skewed outcomes, where classical mean-squared-error splitting criteria become sensitive to extreme values \citep{galimberti_notes_2011}.

Within random forest algorithms based on the framework of \citet{Breiman2001}, splitting rules of the underlying trees typically rely on MSE minimization for regression tasks or impurity measures such as the Gini index for classification. For predictive tasks, several robust alternatives have been proposed. 
\citet{Hastie2009} review absolute and Huber loss functions as splitting criteria to reduce sensitivity to extreme outcomes, and \citet{roy_robustness_2012} propose median-based strategies encompassing tree aggregation, splitting criterion choice, and outcome transformation to robustify random forest predictions. \citet{ishwaran2015effect} study weighted splitting criteria of which MSE and Gini index splitting emerge as special cases, while \citet{L1_splitting} compare log-rank tests and integrated absolute differences as splitting criteria in survival forests. 
\citet{LiMartin2017} provide a unifying framework for forest-type regression that supports a wide range of robust loss functions, nesting classical random forests and quantile regression forests as special cases. Translating these robustness ideas to the causal forest setting is non-trivial, however, because the splitting criteria must be adapted to the honest sample-splitting framework of \citet{athey_recursive_2016} and \citet{wager_estimation_2018}. In this framework, the data used for split selection are separated from the data used for leafwise treatment-effect estimation. Beyond the standard MSE-based splitting rules for causal trees, \citet{athey_recursive_2016} discuss several alternative partitioning estimators, including transformed-outcome trees, fit-based trees, which choose splits by improvements in outcome fit, and t-statistic trees, which prioritize treatment-effect heterogeneity across candidate leaves. Further extensions include \citet{Lechner_Modified_CF}, who augment variance-based splitting by a propensity-score penalty, and \citet{chernozhukov_riesznet_2022}, who generalize the splitting mechanism via their ForestRiesz estimator. 

Our approach differs in construction from the weighted least-absolute-deviation criterion of \citet{li_robust_2021}. We retain the honest, mean-based splitting objective of \citet{athey_recursive_2016} and \citet{wager_estimation_2018} and replace its difference-in-means anchor with the Hodges--Lehmann location estimator. To our knowledge, anchoring the honest splitting objective on the Hodges--Lehmann estimator has not previously been studied. Our primary contribution is the resulting Median Squared Deviation (MSD) criterion, anchored on this estimator within the causal forest framework. We implement a computationally convenient version of this criterion and evaluate its performance in simulations and applications.
We further examine two related heuristic variants, the Median Absolute Deviation (MAD) and the Least Median of Squares (LMS), which we treat as practical alternatives without the same formal guarantees. 

The paper is structured as follows. Section~\ref{sec:PO_framework} reviews the potential outcome framework. Section~\ref{sec:causal_trees} introduces causal trees and the honest sample-splitting construction. Section~\ref{sec:MedianSplitting} introduces the Hodges--Lehmann estimator as a consistent estimator of the leafwise location shift and implements it in new splitting rules, primarily the Median Squared Deviation (MSD) criterion and, as heuristic variants, the Median Absolute Deviation (MAD) and the Least Median of Squares (LMS). Section~\ref{sec:simulations} presents a simulation study comparing the different causal tree splitting rules, embedded in the resulting causal forest estimators, in terms of ATE and CATE precision and confidence interval coverage. Section~\ref{sec:empappl} illustrates the methods on the Progresa conditional cash transfer experiment of \cite{de_la_o_conditional_2013} and the ACTG 175 antiretroviral trial of \cite{hammer_trial_1996}. Section~\ref{sec:conclusion} concludes.

\section{Potential Outcomes Framework}
\label{sec:PO_framework}
To formalize the HTEs motivated in Section~\ref{sec:Introduction}, we work within the potential outcomes framework, a standard approach for causal inference in economics and the social sciences, particularly for observational data and randomized experiments \citep{dominici_controlled_2021, hernan_miguel_a_what_nodate}. 
The general idea was introduced by Jerzy Neyman in his 1923 work on randomized agricultural experiments \citep{splawa-neyman_application_1990}. It was developed further in a series of papers \citep{rubin_estimating_1974, rubin_inference_1976, rubin_bayesian_1978, rubin_randomization_1980} and addresses what \cite{holland_statistics_1986} later termed the fundamental problem of causal inference, that only one of the two potential outcomes is observed for each unit.
The resulting Rubin Causal Model \citep{rubin_rubin_2011, hernan_miguel_a_what_nodate} rests mainly on the three assumption of stable unit treatment values, unconfoundedness, and common support, as reviewed below. 
We consider independent and identically distributed sampled data $(X_i, Y_i, D_i)$ for $i = 1, ..., N$ with the pair of potential outcomes $(Y_i(0), Y_i(1))$ for each unit $i$ based on a binary treatment indicator $D_i \in \{0,1\}$ where \(D_i=1\) denotes treatment and \(D_i=0\) denotes control, and a covariate vector $X_i$ of dimension $K$, where $K$ is the number of covariates. The individual treatment effect (ITE), 
\begin{equation}
\label{eq:ITE}
    \tau_i \coloneqq Y_i(1)-Y_i(0),
\end{equation}
is based on the existence of a potential outcome pair for each observation $i$. The notation of potential outcomes relies on the Stable Unit Treatment Value Assumption (SUTVA), which rules out interference between units and multiple or hidden versions of treatment. Assumption~\ref{assump:SUTVA} formalizes the link between the potential outcomes and the observed outcome.

\begin{assump}{Stable unit treatment value assumption (SUTVA).}
\label{assump:SUTVA}
\begin{align*}
\text{If } D_i = d \text{, then } Y_i(d) = Y_i^{obs} \text{ , } \forall d \in \{0, 1 \} \text{ , } \forall i \in \{1, ..., N \}. 
\end{align*}
\end{assump}
Under Assumption~\ref{assump:SUTVA}, the observed outcome can be written as
\[Y_i^{obs}=Y_i(1)D_i+Y_i(0)(1-D_i).\]
Together with the no-interference and no-hidden-versions components of SUTVA, this implies that each unit has two well-defined potential outcomes and that unit $i$'s potential outcomes are unaffected by the treatment assignment of other units. Moreover, Assumption \ref{assump:SUTVA} removes different treatment manifestations from consideration, which is also known as treatment variation irrelevance. While Assumption \ref{assump:SUTVA} is a common but relatively strong assumption, modifications such as allowing for multiple versions of treatment \citep{vanderweele_ignorability_2008, vanderweele_causal_2013, laffers_identification_2020} or allowing for interference between individuals within certain groups of the population \citep{hudgens_toward_2008} have been proposed but are not the focus of this paper. 

We further assume an unconfounded treatment assignment mechanism independent of potential outcomes in Assumption \ref{assump:unconfoundedness}. Given the covariates, the presence of unmeasured confounders is ruled out and leads to conditional independence between potential outcomes and treatment status.

\begin{assump}{Unconfoundedness.}
\label{assump:unconfoundedness}
\begin{align*}
D_i \perp\!\!\!\perp \left(Y_i(1), Y_i(0)\right) \mid X_i, \\
\text{or equivalently,} \quad 
Pr(D_i \mid Y_i(1), Y_i(0), X_i) = Pr(D_i \mid X_i).
\end{align*}
\end{assump}

In perfectly randomized experiments, Assumption \ref{assump:unconfoundedness} is fulfilled by design as the assignment mechanisms are under the researcher's control. In the analysis of observational studies, this assumption and its sensitivity can, at best, be verified indirectly. Often, it remains untestable and has to be justified in the context of the research question at hand \citep{imbens_nonparametric_2004}.  

Thirdly, we assume common support so that treatment effects are identifiable throughout the support of the covariates. Let the conditional treatment probability or propensity score be defined as  \begin{equation}
\label{eq:propensity}
p(X_i=x) \coloneqq Pr(D_i=1|X_i=x).
\end{equation} 

The following overlap condition in Assumption \ref{assump:overlap} requires at every covariate value in the support of \(X_i\), units have a nonzero probability of receiving either treatment state.

\begin{assump}{Overlap.}
\label{assump:overlap} 
\begin{align*}
\exists\, \epsilon > 0 \text{ such that } \epsilon < p(X_i=x) < 1-\epsilon, \quad \forall x \text{ in the support of } X_i, \text{ with probability } 1.
\end{align*}
\end{assump}

Strong overlap improves the stability of treatment-effect estimates and is also
important for the construction of valid confidence intervals. Moreover, convergence rates of semi-parametric estimators depend heavily on the degree of common support. Two common solutions to ensure sufficient overlap are trimming based on the propensity score and the use of more robust inference procedures that estimate data adaptively based on the amount of covariate overlap \citep{damour_overlap_2021, busso_new_2014, rothe_robust_2017}. Assumptions \ref{assump:unconfoundedness} and \ref{assump:overlap} imply that treatment assignment is ignorable for the identification of treatment effects. Using Assumptions \ref{assump:SUTVA} to \ref{assump:overlap}, one can identify the conditional average treatment effect (CATE) as
\begin{align}
\begin{split}
\label{eq:CATE}
\tau(X_i = x) &\coloneqq \mathbb{E}[Y_i(1) - Y_i(0) \mid X_i = x] \\
&= 
\underbrace{\mathbb{E}[Y_i \mid D_i = 1, X_i = x]}_{\coloneqq\mu_1(x)} 
- 
\underbrace{\mathbb{E}[Y_i \mid D_i = 0, X_i = x]}_{\coloneqq\mu_0(x)}. 
\end{split}
\end{align}
The CATE in Equation~\eqref{eq:CATE} is the target estimand of this paper, and the splitting rules of Section~\ref{sec:MedianSplitting} are designed to create partitions of the covariate space that allow a robust estimation of the sample analog of \eqref{eq:CATE}. Note that the sample analogs of the conditional mean functions $\mu_1(x)$ and $\mu_0(x)$ can be estimated separately from observational data and that the ITE in Equation~\eqref{eq:ITE} should not be confused with the CATE in \eqref{eq:CATE}, since in general $\tau(X_i) \neq \tau_i$ \citep{post2024flexible, PostvandenHeuvel+2025}.
Averaging the CATE over the covariate distribution yields the Average Treatment Effect (ATE),
\begin{equation} \label{eq:ATE}
\tau = \mathbb{E}[Y_i(1) - Y_i(0)] = \mathbb{E}\left[\tau(X_i)\right],
\end{equation}
a scalar summary of the heterogeneity in \eqref{eq:CATE} that we also report in Sections~\ref{sec:simulations} and~\ref{sec:empappl}. 
Because both the CATE and the ATE are functionals of the same conditional mean functions $\mu_1(x)$ and $\mu_0(x)$, a rich representation of the CATE in Equation~\eqref{eq:CATE} requires a flexible nonparametric estimator of these conditional means. Causal trees and forests provide such an estimator while supporting valid inference through honest sample splitting, as summarized in Section~\ref{sec:causal_trees}.

\section{Causal Trees and Causal Forests}
\label{sec:causal_trees}

By partitioning the feature space into rectangles, the general tree-based method builds upon the idea of fitting trivial separate models to each rectangle, aiming at estimating the expectation of the outcome variable conditional on the regressors \citep{Hastie2009}. These methods can cope very well with large-scale datasets while the structure of a single tree retains its interpretability. Throughout this paper we focus on regression trees based on the CART procedure proposed by \cite{cart84} and follow the notation of \cite{athey_recursive_2016}.\footnote{Related approaches we do not consider are, for instance, the Iterative Dichotomiser 3 (ID3) or multivariate adaptive regression splines (MARS) \citep{quinlan_induction_1986, friedman_multivariate_1991}.} A CART regression tree $\Pi$, also referred to as a partition, divides the feature space into separate segments in a recursive manner until reaching a set of leaves $\Pi = \{l_1, ..., l_L\}$ with a predefined minimum leaf size \citep{wager_estimation_2018}. Let $\mathcal{S}$ denote a sample with sample size $N$ and let $l(x; \Pi)$ be a specific leaf $l \in \Pi$ with $x \in l$. Then the leaf-level conditional mean function $\mu_d(x;\Pi)$ is the population average outcome within the leaf containing $x$,
\begin{align}
\label{eq:conditional_mean_function}
    \mu_d(x; \Pi) = \mathbb{E}[Y_i|D_i=d, X_i \in l(x;\Pi) ]. 
\end{align}
Here $\mu_d(x; \Pi)$ is the leaf-averaged, partition-dependent counterpart of the pointwise population conditional mean $\mu_d(x)$ defined in \eqref{eq:CATE}, averaging $\mu_d(\cdot)$ over the leaf $l(x; \Pi)$ containing $x$ rather than conditioning on the single point $X_i = x$.
For a partition $\Pi$ chosen independently of a sample $\mathcal{S}$, the leaf sample mean is an unbiased estimate of this conditional mean function and is given by
\begin{equation}
\label{estim_cond_mean_func}
\hat\mu_d(x; \mathcal{S}, \Pi) \coloneqq \frac{1}{\# (i \in \mathcal{S}: D_i= d,  X_i \in l(x; \Pi))} \sum_{i\in \mathcal{S}:D_i=d, X_i \in l(x; \Pi)} Y_i.
\end{equation}

 Given a tree $\Pi$ and a test point $x$, the treatment effect within a specific leaf $l \in \Pi$ can thus be denoted as the difference of average outcomes between the treatment and control group in leaf $l$
\begin{equation}
\label{leaf_treatment_effect}
\hat\tau(x; \mathcal{S}, \Pi) \coloneqq \hat\mu_1(x; \mathcal{S}, \Pi) - \hat\mu_0(x; \mathcal{S}, \Pi).
\end{equation}
For simplicity and without loss of clarity, we abstract from explicitly including $\Pi$ or $\mathcal{S}$ in our further notation if the partition itself is not relevant for the main argument (i.e. $\hat\tau(x; \mathcal{S})$ instead of $\hat\tau(x; \mathcal{S}, \Pi)$). 
We further use two interchangeable forms for any leaf-constant quantity: a pointwise form $\hat{\tau}(x; \mathcal{S}, \Pi)$ when predicting for individuals or summing over $i \in \mathcal{S}$, and a leaf-direct form $\hat{\tau}(l; \mathcal{S}, \Pi)$ when summing over $l \in \Pi$. Both refer to the same object via $l = l(x; \Pi)$, and the sample argument $\mathcal{S}$ is likewise suppressed when clear from context.

In the usual CART algorithm, recursive partitioning and leaf estimation use the same training data set. This adaptive procedure uses the same observations to choose the partition and to estimate the within-leaf quantities. An adaptive reuse of the data is prone to overfitting, because splits are chosen precisely to improve in-sample fit. 
Tree complexity is typically controlled through cost-complexity pruning. After growing a large tree, a penalty on the number of terminal leaves is used to select a smaller subtree. The benefit of the conventional, adaptive CART procedure is its efficient data usage, as all training observations are considered for partitioning and estimation. 
However, the CART algorithm tends to produce more extreme leaf means as spurious extreme values are grouped with high probability \citep{cart84, Hastie2009}.

\cite{athey_recursive_2016} address this problem by introducing honesty.
A tree is called honest if the data used to estimate within-leaf responses is disjoint from, and thereby independent of, the data used to determine the splits of the partition $\Pi$. Honesty is a property of the
estimation protocol, not of any specific algorithm: the honest causal tree of \cite{athey_recursive_2016}, the Double-Sample Tree of \cite{wager_estimation_2018}, and the Propensity Tree of \cite{wager_estimation_2018} all achieve it by slightly different constructions. 
Honest leaf estimation mitigates overfitting bias that would otherwise arise from using the same data both to select $\Pi$ and to estimate $\hat\tau(l)$, and is what delivers per-leaf conditional unbiasedness for $\tau(l;\Pi)$.

Following \citet{athey_recursive_2016}, we distinguish between three disjoint samples by considering a training sample $\mathcal{S}^{tr}$ to build a tree, an estimation sample $\mathcal{S}^{est}$ to estimate effects for a given tree partition $\Pi$, and a test sample $\mathcal{S}^{te}$ to evaluate out-of-sample performance. 
We denote subsample sizes by $N^{\star}_{d,l}$, where the superscript $\star \in \{\mathrm{tr}, \mathrm{est}, \mathrm{te}\}$ indicates the sample $\mathcal{S}^{\star}$, the subscript $d \in \{0,1\}$ indicates treatment status, and $l \in \Pi$ indicates the leaf. Subscripts are dropped when the corresponding restriction is not imposed, so that $N^{\mathrm{tr}}$ is the full training-sample size, $N^{\mathrm{tr}}_1$ the number of treated training observations, $N^{\mathrm{tr}}_l$ the number of training observations in leaf $l$, and $N^{\mathrm{tr}}_{1,l}$ the number of treated training observations in leaf $l$. The sample superscript is suppressed when clear from context.

Throughout the paper, we refer to a standard honest causal tree
\citep{athey_recursive_2016} that uses the following tree construction and estimation
procedure in four steps.
\begin{enumerate}[label=\textup{(}\roman*\textup{)}, ref=\roman*]
    \item \label{itm:first} split the data into a training sample
      $\mathcal{S}^{tr}$, an estimation sample $\mathcal{S}^{est}$, and a test
      sample $\mathcal{S}^{te}$;
    \item \label{itm:second} grow a tree $\Pi$ on $\mathcal{S}^{tr}$ by
      maximizing the honest splitting criterion
      $\hat Q^{\,\text{MSE, H}}_\tau$ in \eqref{eq:MSE_CT_H}, which trades the
      in-sample heterogeneity reward
      $\sum_l (N^{tr}_l/N^{tr})\,\hat\tau(l)^2$ against a leaf-variance
      penalty weighted by $(1/N^{tr}_l + 1/N^{est}_l)$, and determine tree
      complexity by $k$-fold cross-validation within $\mathcal{S}^{tr}$ via
      cost-complexity pruning. We use $k$ here for the number of cross-validation folds, a standard convention that is unrelated to the observation indices elsewhere in the paper;
    \item \label{itm:third} with $\Pi$ fixed, estimate leaf-level CATEs as the
      difference-in-means on the held-out estimation sample,
      $\hat{\tau}(l;\mathcal{S}^{est}) = \bar{Y}^{\,est}_{1,l} - \bar{Y}^{\,est}_{0,l}$. Because $\mathcal{S}^{est}$ and $\mathcal{S}^{tr}$ are independent, this estimator is conditionally unbiased for $\tau(l;\Pi)$, and
    \item \label{itm:fourth} predict for a new $x$ by routing $x$ through
      $\Pi$ to its leaf $l(x;\Pi)$ and reading
      $\hat\tau(x;\Pi) = \hat\tau\bigl(l;\mathcal{S}^{est}\bigr)$.
\end{enumerate}

Step~(\ref{itm:second}) is the core component in the causal tree construction process in \citet{athey_recursive_2016} and is
where the four splitting rules in this paper differ. We introduce the
mean-based criterion of \citet{athey_recursive_2016} here in detail and use it
as baseline for constructing the median-based and heuristic rules in
Sections~\ref{sec:MSD} and \ref{sec:MAD_LMS}: each rule is presented through
an infeasible MSE on $(\mathcal{S}^{te},\mathcal{S}^{est})$ that defines
its population target, followed by an in-sample criterion on
$\mathcal{S}^{tr}$ that is the sample analog used for split selection.

Since the individual treatment effect $\tau_i = Y_i(1) - Y_i(0)$ is never
observed, the natural prediction risk
$\frac{1}{N^{te}}\sum_{i\in\mathcal{S}^{te}}(\tau_i - \hat\tau(X_i))^2$ is
infeasible. \citet{athey_recursive_2016} work with the modified MSE that
removes the unobservable component while preserving the ranking of candidate
trees, since $\tau_i^2$ depends on neither $\Pi$ nor $\mathcal{S}^{est}$:
\begin{equation}\label{eq:MSE_treatment_effect}
  \text{MSE}_\tau(\mathcal{S}^{te}, \mathcal{S}^{est}, \Pi)
  = \frac{1}{N^{te}} \sum_{i \in \mathcal{S}^{te}}
    \Bigl[\bigl(\tau_i - \hat{\tau}(X_i; \mathcal{S}^{est}, \Pi)\bigr)^2 - \tau_i^2\Bigr].
\end{equation}
Equation~\eqref{eq:MSE_treatment_effect} is still not feasible in step~(\ref{itm:second}):
the leaf estimates $\hat\tau(X_i;\mathcal{S}^{est},\Pi)$ depend on $\mathcal{S}^{est}$, which has not been used during split selection, and the individual $\tau_i$ are unobserved on $\mathcal{S}^{te}$. To make the criterion feasible, \citet{athey_recursive_2016} take expectations of \eqref{eq:MSE_treatment_effect} over \((\mathcal S^{te},\mathcal S^{est})\) to obtain the expected mean squared error \(\mathrm{EMSE}_\tau(\Pi)\). They then construct a training-sample estimator of \(-\mathrm{EMSE}_\tau(\Pi)\) by replacing the unknown leaf CATE \(\tau(l;\Pi)\) with the training-sample difference-in-means and by subtracting the corresponding variance correction. This yields the honest splitting criterion
\begin{equation}\label{eq:MSE_CT_H}
  \hat{Q}_\tau^{\,\text{MSE, H}}(\mathcal{S}^{tr}, \Pi)
  = \sum_{l \in \Pi} \frac{N_l^{tr}}{N^{tr}}
    \left[\hat{\tau}(l;\mathcal{S}^{tr})^2
    - \left(\frac{1}{N_l^{tr}} + \frac{1}{N_l^{est}}\right)
      \!\left(\frac{\hat{S}^2_{1}(l;\mathcal{S}^{tr})}{p}
      + \frac{\hat{S}^2_{0}(l;\mathcal{S}^{tr})}{1-p}\right)\right],
\end{equation}
which is maximized over $\Pi$. Here, \(\hat\tau(l;\mathcal S^{tr})=\bar Y^{tr}_{1,l}-\bar Y^{tr}_{0,l}\) is the
training-sample CATE estimate, and
\(\hat S_d^2(l;\mathcal S^{tr})\) is the within-leaf outcome variance for treatment status \(d\). The factor \(N_l^{tr}/N^{tr}\) estimates the leaf probability, while
\(N_l^{est}\) denotes the estimation-sample leaf size associated with leaf \(l\) under \(\Pi\).
The treatment share $p$ equals the known propensity score in the randomized designs considered here and should not be confused with the leaf probability.
The first term in \eqref{eq:MSE_CT_H}, $\hat\tau(l;\mathcal{S}^{tr})^2$
weighted by $N_l^{tr}/N^{tr}$, is the in-sample heterogeneity reward: splits
that produce leaves with larger differences in CATEs increase this term. The second term is the variance penalty. It corrects the upward bias of the plug-in $\hat\tau(l;\mathcal{S}^{tr})^2$ and discounting heterogeneity that will be hard to detect on
$\mathcal{S}^{est}$. The main contribution of this paper
is the modification of \eqref{eq:MSE_treatment_effect} and \eqref{eq:MSE_CT_H} by replacing its mean-based anchor with median-based alternatives that are less sensitive to heavy-tailed or skewed outcome distributions.

Although every quantity in \eqref{eq:MSE_CT_H} is computed on $\mathcal{S}^{tr}$, the penalty depends on the estimation-sample leaf size \(N_l^{est}\), which determines how strongly estimation variance is penalized. The observations in $\mathcal{S}^{est}$ are used only later for honest leafwise estimation in step~(\ref{itm:third}). The relation between \eqref{eq:MSE_CT_H} and \eqref{eq:MSE_treatment_effect} is derived in detail in Appendix~\ref{appendix:Honest_Splitting_CT}. We further note that the superscript ``H'' on $\hat Q^{\,\text{MSE, H}}_\tau$ marks the splitting criterion as honest due to the variance penalty derived under the EMSE objective. An adaptive criterion
\(\hat Q^{\,\mathrm{MSE,A}}_\tau\) omits this term. This is conceptually different from honest leaf estimation, which uses an independent $\mathcal{S}^{est}$ in step~(\ref{itm:third}) to compute $\hat\tau(l;\mathcal{S}^{est})$. Honest leaf estimation delivers per-leaf conditional unbiasedness for $\tau(l;\Pi)$, given the selected partition.

So far, we have focused on split selection and leafwise estimation for a single honest causal tree. The criterion in \eqref{eq:MSE_CT_H} is defined at the tree level. In practice, however, single trees can be unstable because small changes
in the data may lead to different partitions and high-variance leaf estimates.  For this reason, the estimator we
evaluate in Sections~\ref{sec:simulations} and~\ref{sec:empappl} is the
corresponding causal forest of \citet{wager_estimation_2018}: an ensemble of
$B$ honest causal trees grown on random subsamples of the data and aggregated
by averaging,
\begin{equation}\label{eq:causal_forest}
    \hat{\tau}(X_i) \;=\; \frac{1}{B}\sum_{b=1}^{B} \hat{\tau}_b(X_i),
\end{equation}
where $\hat\tau_b(X_i)$ is the prediction of the $b$-th tree, built and
re-estimated according to steps~(\ref{itm:first})--(\ref{itm:fourth}) on its
assigned subsample. The splitting rules considered in this paper (the mean-based
$\hat{Q}^{\,\text{MSE, H}}_\tau$ in \eqref{eq:MSE_CT_H}, the median-based MSD
in Section~\ref{sec:MSD}, and the heuristic MAD and LMS in
Section~\ref{sec:MAD_LMS}) differ only in step~(\ref{itm:second}) of each tree. Honest leaf estimation, subsampling, and forest aggregation are kept fixed across methods.

Aggregating the CATE estimates in \eqref{eq:causal_forest} over the full sample gives a simple plug-in estimator of the ATE in \eqref{eq:ATE},
\begin{equation} \label{eq:diff_estimator_ATE}
\hat{\tau}_{\text{plug-in}} 
= \frac{1}{N} \sum_{i \in \mathcal{S}} \hat{\tau}(X_i),
\end{equation}
which is consistent but not necessarily efficient. For the ATE comparisons in Sections~\ref{sec:simulations} and~\ref{sec:empappl} we instead use the augmented inverse probability weighting (AIPW) estimator \citep{athey_policy_2021, robins_estimation_1994, chernozhukov_doubledebiased_2018}, detailed in Appendix~\ref{appendix:AIPW_ATE_CF}.

The pointwise consistency and asymptotic normality of the causal forest CATE estimates in \eqref{eq:causal_forest} are established in Theorem~4.1 of \citet{wager_estimation_2018},
which carries the generic forest central limit theorem of their Theorem~3.1 into the potential outcomes setting under unconfoundedness and overlap. In particular the trees must be honest, $\alpha$-regular, random-split,
and symmetric in the sense of \citet{wager_estimation_2018}, and the subsample size must scale appropriately.
We note that there is recent literature that scrutinizes these conditions of \citet{wager_estimation_2018}. \citet{cattaneo_honest_2025} show that adaptive causal trees generate unbalanced leaves with non-vanishing probability and argue that fixed-fraction $\alpha$-regularity may be incompatible with standard recursive partitioning, so that honesty alone may not guarantee uniform convergence. By contrast,
\citet{bladt_consistency_2026} show that honest trees remain consistent when leaves localize in feature space while retaining enough observations for stable within-leaf estimation. We treat
honesty and $\alpha$-regularity as maintained assumptions, enforced at the level of the implementation in Appendix~\ref{appendix:implementation}, and note that
the operative practical condition is the leaf-size regime rather than the splitting rule.

Replacing the mean-based criterion with the Median Squared Deviation (MSD), Median Absolute Deviation (MAD), or Least Median of Squares (LMS) rule introduced in the next section changes only the split-selection criterion in
step~(\ref{itm:second}). The honest sample split, subsampling scheme, and difference-in-means leaf estimator in step~(\ref{itm:third}) remain unchanged.\footnote{The balance constraint is Definition~4 of \citet{wager_estimation_2018}, which keeps $\alpha$-regularity intact under criteria whose unconstrained optimum can be highly unbalanced, most notably LMS.} 
Therefore, their Theorem~4.1 continues to apply. The result relies on the finite-moment conditions of \citet{wager_estimation_2018}, which the heavy-tailed but finite-variance designs of Section~\ref{sec:simulations} satisfy, so the robustness we target lies in the stability of split selection under heavy-tailed or skewed outcome distributions. Appendix~\ref{appendix:implementation} documents that the implemented forest satisfies the conditions of Theorems~3.1 and~4.1 of \citet{wager_estimation_2018} and identifies the mechanism enforcing it, so the transfer rests on properties of the protocol that we briefly discuss there.

\section{Median-based Splitting Rules}
\label{sec:MedianSplitting}

We introduce median-based splitting rules as robust alternatives to the MSE-based criteria in Equations~\eqref{eq:MSE_treatment_effect} to \eqref{eq:MSE_CT_H}. \cite{roy_robustness_2012} proposed median-based splits to robustify regression forests for outcome prediction with outliers, while \cite{athey_recursive_2016} proposed the causal tree framework in Section \ref{sec:causal_trees} to allow for estimation and inference of CATE. We combine both approaches. This idea is also related to the robust M-estimation perspective of \citet{Huber1964},
where the target estimand can remain unchanged while the loss function used for estimation is modified to improve finite-sample stability (see also \citet{Hastie2009}). We adapt this logic to the causal-tree setting, the target estimand remains the CATE in \eqref{eq:CATE}, while the split-selection objective is modified by replacing its mean-based rule with median-based quantities, to quantify treatment effect heterogeneity under irregular outcome distributions. Because changing the criterion affects only which partition is chosen, and not how effects are estimated within leaves, the leaf-level unbiasedness and the forest inference of Section~\ref{sec:causal_trees} are unaffected by the choice of the splitting rule. This is a property of the honest leaf estimator, not of the splitting criterion in particular, and it holds for all considered rules. Whether the MSD criterion itself reproduces the honest variance-bias decomposition is a separate question, discussed in Section~\ref{sec:MSD}.

Before introducing the splitting rules, let us define a median-based analogue of the average treatment effect, the median treatment effect (MTE), as
\begin{align} \label{eq:MTE}
        \tau_{\operatorname{med}} &\coloneqq \operatorname{med} \tau_i = \operatorname{med} (Y_i(1) - Y_i(0)), 
\end{align} 
Unlike the estimation of the ATE in \eqref{eq:ATE}, the estimation of the MTE in Definition \eqref{eq:MTE} is particularly challenging, because it is generally not identified from the marginal distributions of \(Y_i(1)\) and \(Y_i(0)\) alone. In particular, it depends on the joint distribution of \((Y_i(1),Y_i(0))\), while only one potential outcome is observed for each unit and crucially, two different joint distributions can share the exact same marginal distributions yet produce very different $\tau_{\operatorname{med}}$. Thus, additional assumptions on the dependence structure of the potential outcomes are required \citep{Addanki2024LimitsOA}.

The same issue remains for the conditional median treatment effect (CMTE) adapted to our leaf-wise potential outcomes setting and defined as 
\begin{align} \label{eq:CMTE}
        \tau_{\operatorname{med}}(X_i) &\coloneqq \operatorname{med}\bigl(Y_i(1) - Y_i(0)\mid X_i \bigr).  
\end{align}
Since $\tau_i = Y_i(1) - Y_i(0)$ is not observable, we cannot directly compute the median of these unobserved quantities.  The key distinction to keep in mind is that
\begin{align}\label{eq:median_of_differences}
  \tau_{\operatorname{med}}(X_i) = \operatorname{med} (Y_i(1) - Y_i(0) \mid X_i )\neq \operatorname{med}\bigl(Y_i(1) \mid X_i \bigr) - \operatorname{med}\bigl(Y_i(0) \mid X_i  \bigr),
\end{align}
which is the difference of conditional medians. The difference of conditional medians is identifiable from observed data \citep{leqi_median_2022}, but as \citet{Addanki2024LimitsOA} argue, it can diverge substantially from the median of individual differences because the coupling between $Y_i(1)$ and $Y_i(0)$ conditional on $X_i$ is unobservable.

\subsection{The Hodges–Lehmann estimator as splitting anchor}

Instead of using the right-hand side of \eqref{eq:median_of_differences} directly as a splitting target, we follow the literature on the Hodges--Lehmann (HL) estimator and the Wilcoxon rank-sum statistic \citep{hodges_estimates_1963, hoyland_HL, hollander_nonparametric_2014}. The HL estimator is most naturally interpreted as a rank-based estimator of a location-shift parameter. We therefore first establish its leafwise shift properties.
For a fixed \(\Pi\), define the leafwise median treatment effect in
the leaf containing \(x\) as
\[\tau_{\operatorname{med}}(x;\Pi):=\operatorname{med}\bigl(Y_i(1)-Y_i(0)\mid X_i\in l(x;\Pi)\bigr).\]

The connection between the shift parameter, $\tau_{\operatorname{med}}(x;\Pi)$ or the CATE, as well as its role in the subsequent splitting criteria, requires additional assumptions, which are discussed in the following subsections.
Let $\mathcal{S}_{1}$ and $\mathcal{S}_{0}$ denote the subsamples for the treatment and control group of sample $\mathcal{S}$, with corresponding sample sizes $N_{1}$ and $N_{0}$. Given a tree $\Pi$ and test point $x$, we define the HL estimator adapted to the leaf-wise potential outcomes setting as
\begin{equation}\label{eq:HL-leafwise}
  \hat\tau_{\mathrm{HL}}(x;\Pi)
=
\operatorname{med}\Bigl\{
Y_j-Y_m
:\,
j\in\mathcal S_1,\;
m\in\mathcal S_0,\;
X_j\in l(x;\Pi),\;
X_m\in l(x;\Pi)
\Bigr\}.
\end{equation}
To study the large-sample behavior of $\hat\tau_{\mathrm{HL}}(x;\Pi)$ in the leafwise setting, we
first introduce the following assumptions.

\begin{assump}\label{ass:HL_shift}
For a fixed $\Pi$ and a target point $x$, let $l(x;\Pi)$ be the leaf containing $x$.
Define the leafwise population conditional distributions
\[
F_{1,l}(z) := Pr\bigl(Y_i(1)\le z \mid X_i\in l(x;\Pi)\bigr),
\qquad
F_{0,l}(z) := Pr\bigl(Y_i(0)\le z \mid X_i\in l(x;\Pi)\bigr).
\]
Assume that these distributions are absolutely continuous, with corresponding densities $f_{1,l}$ and
$f_{0,l}$. In addition, suppose that:
\begin{itemize}
    \item[$i)$] the leafwise potential-outcome distributions differ by a location shift, that is,
    there exists a leaf-specific parameter $\Delta(x;\Pi)\in\mathbb{R}$ such that
    \[
    F_{1,l}(z)=F_{0,l}(z-\Delta(x;\Pi))
    \qquad \text{for all } z\in\mathbb{R};
    \]
     \item[$ii)$] under condition~(i), the common shifted density \(f_l\),
    defined by
    \[
    f_{0,l}(z)=f_l(z),
    \qquad
    f_{1,l}(z)=f_l(z-\Delta(x;\Pi)),
    \]
    is continuous and satisfies
    \[
    0<\int f_l^2(u)\,du<\infty;
    \]
    \item[$iii)$] as the total sample size tends to infinity, the leaf containing $x$ contains
    diverging numbers of treated and control observations, with
    \[
    N_{0,l(x;\Pi)} + N_{1,l(x;\Pi)} \to \infty
    \qquad\text{and}\qquad
    \frac{N_{0,l(x;\Pi)}}{N_{0,l(x;\Pi)} + N_{1,l(x;\Pi)}} \to \lambda \in (0,1).
    \]
\end{itemize}
\end{assump}
Condition (i) in Assumption \ref{ass:HL_shift} is the main restriction. This condition requires the treatment to shift the leafwise outcome distribution without changing its shape. Under this restriction, the treatment-control comparison is summarized by the scalar shift parameter \(\Delta(x;\Pi)\), making the Hodges--Lehmann estimator a natural rank-based anchor. The shift model holds automatically under a within-leaf constant effect (Lemma~\ref{lemma:Delta_equals_tau_leaf}), and we return in Section~\ref{sec:MSD} to the weaker reading of the criterion when it does not.
The following results, proved in the appendix, establish that $  \hat\tau_{\mathrm{HL}}(x;\Pi)$ is a well-behaved anchor for our splitting criteria.

\begin{theorem}\label{thm:theorem_consistency}
Under Assumption~\ref{ass:HL_shift}, the Hodges--Lehmann estimator
$\hat\tau_{\mathrm{HL}}(x;\Pi)$ associated with the Wilcoxon rank-sum statistic in the
Mann--Whitney form is asymptotically normal and centered around
$\Delta(x;\Pi)$. In particular, $\hat\tau_{\mathrm{HL}}(x;\Pi)$ is a consistent estimator of
$\Delta(x;\Pi)$.
\end{theorem}

\begin{proof}
See Appendix \ref{appendix:proofs}.
\end{proof}

In practice, recursive partitioning produces progressively smaller leaves, so finite-sample behavior also matters. A useful property is median unbiasedness, meaning that the
estimator is equally likely to overestimate or underestimate its target. The following lemma gives a general
bound. 

\begin{lemma} \label{lem:median_unbiased}
Let $\mathcal W_l=t(Y^0_l,Y^1_l)$ be the leafwise Wilcoxon rank-sum statistic. If
\[
Pr\bigl(t(Y^0_l,Y^1_l)=\xi_l \mid \Delta(x;\Pi)=0\bigr)=\delta,
\]
where $\xi_l$ is a symmetry point, then
\[
\frac12-\frac{\delta}{2}
\le
Pr\!\left(\hat\tau_{\mathrm{HL}}(x;\Pi)\le \Delta(x;\Pi)\right)
\le
\frac12+\frac{\delta}{2}.
\]
\end{lemma}

\begin{proof}
See Appendix \ref{appendix:proofs}.
\end{proof}

Lemma~\ref{lem:median_unbiased} shows that exact median unbiasedness obtains whenever the symmetry point is unattainable, so whenever $\delta=0$. When the symmetry point is attainable, the same lemma quantifies the possible deviation from exact median unbiasedness by the size of $\delta$. Theorem~\ref{thm:theorem_median_unbiased} applies this observation to show that $\hat\tau_{\mathrm{HL}}(x;\Pi)$ is median unbiased for $\Delta(x;\Pi)$. 

\begin{theorem}\label{thm:theorem_median_unbiased}
For a fixed $\Pi$ and a target point $x$, let $l(x;\Pi)$ be the leaf containing $x$.
Suppose that, under the null shift $\Delta(x;\Pi)=0$, the distribution of the leafwise Wilcoxon rank-sum statistic $\mathcal W_l$ is symmetric about
\[
\xi_l=\frac{N_{1,l}N_{0,l}}{2}.
\]
Then the Hodges--Lehmann estimator $\hat\tau_{\mathrm{HL}}(x;\Pi)$ is median unbiased for
$\Delta(x;\Pi)$ when $N_{1,l}N_{0,l}$ is odd, in the sense that
\[
Pr\!\left(\hat\tau_{\mathrm{HL}}(x;\Pi)\le \Delta(x;\Pi)\right)=\frac12.
\]
If \(N_{1,l}N_{0,l}\) is even, the estimator is approximately median unbiased, in the sense that its deviation from exact median unbiasedness is bounded by \(\delta/2\), with \(\delta\) defined in Lemma~\ref{lem:median_unbiased}.
\end{theorem}

\begin{proof}
See Appendix \ref{appendix:proofs}.
\end{proof}

\subsection{Motivating Example: Splitting under Heavy Tails}
\label{sec:motivation}

Consider a leaf $l$ of a tree $\Pi$ grown on $\mathcal{S}^{tr}$ with $N_{1,l} = N_{0,l} = 10$
observations per treatment group, and let the data follow a constant additive treatment effect,
\begin{align}\label{eq:sharp-null}
\begin{split}
  Y_i(0) &= 3 + \varepsilon_i, \\
  Y_i(1) &= Y_i(0) + \tau, \qquad \tau = 2.  
\end{split}
\end{align}
Under \eqref{eq:sharp-null} the conditional average and conditional median
treatment effects in $l$ both equal $\tau$, and the leafwise shift model of Assumption \ref{ass:HL_shift}(i) holds with $\Delta(x; \Pi) = \tau$ by Lemma \ref{lemma:Delta_equals_tau_leaf}. Thus, the example isolates a simple setting in which the HL anchor is well defined. 

What is at stake here is not how leaf effects are estimated, but which partition is selected.
Honest leaf estimation in step (\ref{itm:third}) reports the difference in means on the held-out $\mathcal{S}^{est}$, and the median-based rules of this section leave that step untouched. The quantity a heavy-tailed draw can corrupt is the training-sample anchor $\hat\tau(l ; \mathcal{S}^{tr})$ entering the splitting criterion \eqref{eq:MSE_CT_H}.
We therefore compare two candidate anchors for that criterion, the training-sample difference-in-means estimator and the HL estimator,
\begin{align}\label{eq:two-estimators}
\begin{split}
  \hat\tau(l; \Pi) &= \bar Y_{1,l} - \bar Y_{0,l}, \\
  \hat\tau_{\mathrm{HL}}(l; \Pi)
    &= \operatorname{med} \{ Y_{j,l} - Y_{m,l}: 1 \le j \le N^{tr}_{1,l},\,
      1 \le m \le N^{tr}_{0,l} \},
\end{split}
\end{align}
where $\hat\tau_{\mathrm{HL}}(l; \Pi)$ is the median of all $N^{tr}_{1,l} \times N^{tr}_{0,l} = 100$ pairwise treated--control differences within $l$ and $\bar Y_{1,l}, \bar Y_{0,l}$ the corresponding training-sample outcome averages.  
The difference-in-means estimator targets $\tau$ because the treatment and control means differ by $\tau$ in population.
The HL estimator targets the same quantity because each pairwise treated--control difference equals $Y_{j,l}-Y_{m,l}=\tau+(\varepsilon_j-\varepsilon_m)$, and $\varepsilon_j-\varepsilon_m$ is symmetric about zero when $\varepsilon_j$ and $\varepsilon_m$ are independent draws from the same baseline distribution.
Hence, the population median of the pairwise differences is $\tau$, even if the baseline distribution itself is skewed. Both anchors in \eqref{eq:two-estimators} are therefore consistent for $\tau$ but they might differ in finite samples.

To see this, let the baseline be heavy-tailed with $\varepsilon_i$ in \eqref{eq:sharp-null} as
\begin{align}
\varepsilon_i \stackrel{\mathrm{iid}}{\sim} t_\nu
\end{align}
and $\nu = 3$, in particular $\Var(\varepsilon) = 3$.
A single observation that is $c$ units above the typical treated outcome increases the treatment-group mean by $c/N^{tr}_{1,l}$. For $N^{tr}_{1,l}=10$, this increase equals $c/10$ and choosing $c=10$ gives an increase of $1$, so $\hat\tau(l;\mathcal{S}^{tr})$ moves from $\tau=2$ to approximately $3$.
One such draw suffices, because the sample mean has breakdown zero. 
The same draw enters only $N^{tr}_{0,l}=10$ of the $100$ pairwise differences, so the median defining $\hat\tau_{\mathrm{HL}}(l;\Pi)$ is still taken over $90$ uncontaminated values and is left essentially unchanged \citep{hollander_nonparametric_2014}. We note that honesty does not fully close this gap. Step (\ref{itm:third}) removes the training-sample draw from the reported effect by re-estimating on $\mathcal{S}^{est}$, but it cannot remove it from $\Pi$, because the split has already been made and the partition is not revised. The mean-based criterion in \eqref{eq:MSE_CT_H} rewards apparent leaf-heterogeneity through $\hat\tau(l;\mathcal{S}^{tr})^2$ while the variance penalty does not offset this increase. Appendix \ref{app:motivation} provides further details for this behavior based on our setup in \eqref{eq:sharp-null}. Because \eqref{eq:MSE_CT_H} is maximized over a large number of candidate splits, and the maximum of a collection of heavy-tailed statistics is governed by its tails, the split that isolates the extreme draw is preferentially selected, yielding trees that track tail noise rather than genuine heterogeneity. Anchoring the split  criterion at $\hat\tau_{\mathrm{HL}}(l;\Pi)$ suppresses this behavior while still targeting $\tau$ under \eqref{eq:sharp-null}.

\subsection{Median Squared Deviation as a Robust Splitting Criterion}
\label{sec:MSD}

We now construct the Median Squared Deviation (MSD) criterion, the main median-based alternative to the honest MSE criterion in \eqref{eq:MSE_CT_H}. The MSD rule modifies only the split-selection objective in step~(\ref{itm:second}) by replacing the mean-based leaf anchor in the sample analog of \eqref{eq:MSE_treatment_effect} with the leafwise Hodges--Lehmann estimator from \eqref{eq:HL-leafwise}. Honest leaf estimation remains difference-in-means estimation on \(\mathcal S^{est}\), so the final target remains the leafwise CATE \(\tau(l;\Pi)\). The idea is that, when the conditional outcome distributions are
heavy-tailed, $\hat\tau(l;\mathcal S^{tr},\Pi)$ is volatile and the within-leaf variances that enter the honest penalty are inflated, both of which distort split selection. 

The HL estimator is consistent for the population location shift $\Delta(x;\Pi)$, and under the location-shift conditions in Appendix~\ref{appendix:HL_estimator} we have $\Delta(x;\Pi) = \tau(l;\Pi)$. 
When those conditions fail, the Hodges--Lehmann anchor estimates the location shift $\Delta(x;\Pi)$, which need not equal $\tau(l;\Pi)$, so the criterion no longer targets the CATE objective exactly and may rank splits suboptimally. This affects only partition selection. 
Because leaf effects are estimated by difference in means on $\mathcal{S}^{est}$ under every splitting rule that we consider in this paper, the resulting CATE estimates remain unbiased for $\tau(l;\Pi)$ on whatever partition is chosen, so a misaligned anchor costs estimation efficiency through a coarser or less informative partition, not validity of the estimate. Scenarios~S3 and~S4 in Section~\ref{sec:simulations} illustrate this violation, and the MSD forest retains the lowest bias and root-mean-square error there, which bounds the practical size of this efficiency cost on the designs we study.

Starting from the modified, infeasible MSE in Equation~\eqref{eq:MSE_treatment_effect}, we
construct a median squared deviation splitting rule by replacing the leafwise treatment effect
estimator $\hat\tau(x;\mathcal S^{est},\Pi)$ with the Hodges--Lehmann estimator
$\hat\tau_{\mathrm{HL}}(x;\mathcal S^{est},\Pi)$. The previous subsection established that
$\hat\tau_{\mathrm{HL}}(x;\Pi)$ estimates the leafwise location shift $\Delta(x;\Pi)$. This gives rise to the infeasible median squared deviation (MSD) splitting criterion based on the MSE in Equation~\eqref{eq:MSE_treatment_effect}: 
\begin{equation}\label{eq:mse_tau_med}
    \mathrm{MSE}_{\tau,\mathrm{med}}(\mathcal{S}^{te},\mathcal{S}^{est},\Pi)
    \coloneqq \frac{1}{N^{te}} \sum_{i \in \mathcal{S}^{te}}
      \Bigl[\bigl(\tau_i - \hat{\tau}_{\mathrm{HL}}(X_i; \mathcal{S}^{est}, \Pi)\bigr)^2 - \tau_i^2\Bigr],
\end{equation}
where $\hat{\tau}_{\mathrm{HL}}(X_i;\mathcal{S}^{est},\Pi)$ is the leaf-constant HL estimate assigned to unit~$i$. The subtraction of $\tau_i^2$ again preserves the relative ranking of candidate trees since $\tau_i^2$ does not depend on $\Pi$ or $\mathcal{S}^{est}$, as in \eqref{eq:MSE_treatment_effect}. Mirroring the honest derivation in Appendix~\ref{appendix:Honest_Splitting_CT} and using
Theorems~\ref{thm:theorem_consistency} and~\ref{thm:theorem_median_unbiased}, we obtain the honest objective
\begin{align}
\label{eq:honest_median}
\begin{split}
\hat{Q}^{\,\text{MSD, H}}_\tau(\mathcal{S}^{tr}, \Pi)
    &= \sum_{l \in \Pi} \frac{N_l^{tr}}{N^{tr}}
    \left[\hat{\tau}_{\mathrm{HL}}(l;\mathcal{S}^{tr})^2
    -  \left(\frac{1}{N_l^{tr}} + \frac{1}{N_l^{est}}\right) \widehat{V}_l \right] 
\end{split}   
\end{align}
where $\hat\tau_{\mathrm{HL}}(l;\mathcal{S}^{tr})$ is the within-leaf HL estimate on the
training sample and $\widehat V_l$ estimates its asymptotic variance. The $(1/N_l^{tr}+1/N_l^{est})$ penalty weight is identical to that of \eqref{eq:MSE_CT_H}. Further details on the construction are provided in Appendix~\ref{appendix:Honest_Splitting_MSD}.

Two features distinguish \eqref{eq:honest_median} from the mean-based honest criterion $\hat Q^{\,\text{MSE,H}}_\tau$ in \eqref{eq:MSE_CT_H}.
First, the heterogeneity reward is anchored on the training-sample estimate $\hat\tau_{\mathrm{HL}}(l;\mathcal{S}^{tr})^2$ rather than the difference in means.
Second, the variance penalty replaces the treatment and control group sample variances $\hat S^2_d / N^{est}_l$ with $\widehat V_l$, an estimated variance appropriate for a median-based estimator.
Because $\widehat V_l$ requires kernel density evaluations and is costly within a tree-growing loop, our implementation maximizes an auxiliary splitting criterion to preserve the qualitative effect of penalizing leaves with weak effective signal-to-noise while avoiding density estimation at every candidate split. 
\begin{align}\label{eq:Q_MSD_code_main}
\begin{split}
    \widehat Q^{\,\text{MSD}}_\tau\left(\mathcal{S}^{tr}, \Pi\right)
    &=\sum_{l\in\Pi}\frac{N_l^{tr}}{N^{tr}}
       \!\left[\hat\tau_{\mathrm{HL}}(l;\mathcal{S}^{tr})^2\,-\Bigl(
       \hat\tau_{\mathrm{HL}}\left(l;\mathcal{S}^{tr}\right)-\hat\tau\left(l;\mathcal{S}^{tr}\right)\Bigr)^2\right]
\end{split}
\end{align}
The honest MSD criterion in \eqref{eq:honest_median} reproduces the variance-bias decomposition of the mean-based criterion in \eqref{eq:MSE_CT_H}, with the Hodges--Lehmann variance in place of the outcome-variance penalty. The implemented criterion in \eqref{eq:Q_MSD_code_main} is the leading-order surrogate of \eqref{eq:honest_median}. This surrogate drops the explicit variance penalty in favor of the squared difference between Hodges--Lehmann estimates and standard mean-based leafwise estimates while reproducing the decomposition only up to a residual of order $O(1/N_l^{est})$, derived in Appendix~\ref{appendix:Honest_Splitting_MSD}. It is this surrogate that the simulations of Section~\ref{sec:simulations} evaluate.

\subsection{Two Baseline Robust Splitting Criteria}
\label{sec:MAD_LMS}

We complement the previous discussion with two additional splitting rules that are robust by construction but differ from MSD in their theoretical role. First, we consider the median absolute deviation (MAD) rule that uses the absolute difference between the mean-based leaf estimate and the Hodges--Lehmann leaf estimate, and second a least median squares (LMS) rule built on fit-based
partitioning \citep{ZeileisModelBasedRP, athey_recursive_2016, roy_robustness_2012}. We treat both as heuristics and report them in
Section~\ref{sec:simulations} as comparison rules for MSD, since they do not yield the same honest EMSE decomposition. Appendix~\ref{appendix:MAD_LMS} provides further details about MAD and LMS.

Following \citet{roy_robustness_2012}, the absolute-deviation analog of the
infeasible CATE-targeting MSE in \eqref{eq:MSE_treatment_effect} is
\begin{equation}\label{eq:infeasible_MAD}
    \mathrm{MAD}_\tau(\mathcal{S}^{tr})
    \;:=\;\frac{1}{N^{tr}}\sum_{i\in\mathcal{S}^{tr}}
    \bigl|\tau_i - \tau_{\mathrm{med}}(X_i)\bigr|,
\end{equation}
where $\tau_{\mathrm{med}}(X_i)$ is the CMTE from \eqref{eq:CMTE}.
\eqref{eq:infeasible_MAD} is doubly infeasible: $\tau_i$ is unobserved (as in
the MSE), and $\tau_{\mathrm{med}}(X_i)$ is not point-identified without
restrictions on the joint distribution of potential outcomes
\citep{Addanki2024LimitsOA}. We obtain a feasible rule by replacing $\tau_i$ by
the mean-based plug-in $\hat\tau(X_i;\mathcal{S}^{tr},\Pi)$ and
$\tau_{\mathrm{med}}(X_i)$ by the leafwise HL estimator
$\hat\tau_{\mathrm{HL}}(X_i;\mathcal{S}^{tr},\Pi)$:
\begin{equation}\label{eq:MAD_split}
    \widehat Q^{\,\mathrm{MAD}}_\tau(\mathcal{S}^{tr},\Pi)
    \;:=\;\frac{1}{N^{tr}}\sum_{i\in\mathcal{S}^{tr}}
    \bigl|\hat\tau(X_i;\mathcal{S}^{tr},\Pi)
    -\hat\tau_{\mathrm{HL}}(X_i;\mathcal{S}^{tr},\Pi)\bigr|.
\end{equation}
Splits are chosen by minimizing \eqref{eq:MAD_split}, so the criterion favors partitions for which the mean-based and Hodges--Lehmann leaf estimates are close.   
Large values of \eqref{eq:MAD_split} indicate leaves in which the mean-based treatment-effect estimate differs substantially from the rank-based location estimate, which can occur when tail observations influence the leaf mean.

The interpretation of \eqref{eq:MAD_split} differs from that of the infeasible criterion in \eqref{eq:infeasible_MAD}. The infeasible \eqref{eq:infeasible_MAD} is defined in terms of individual treatment effects $\tau_i$ and the 
CMTE $\tau_{\mathrm{med}}(X_i)$, whereas the feasible \eqref{eq:MAD_split}
compares two leaf-level estimators of $\tau(l;\Pi)$ (the mean-based
$\hat\tau(X_i;\mathcal{S}^{tr},\Pi)$ and the location-shift-based $\hat\tau_{\mathrm{HL}}(X_i;\mathcal{S}^{tr},\Pi)$) on
$\mathcal{S}^{tr}$. The two objects coincide only under the location-shift
conditions of Appendix~\ref{appendix:HL_estimator}, where
$\Delta(x;\Pi) = \tau_{\mathrm{med}}(x) = \tau(l;\Pi)$. Outside that regime,
\eqref{eq:MAD_split} is a robustness diagnostic and not a sample analog of
\eqref{eq:infeasible_MAD}.

Next, we motivate LMS as a robust analogue of fit-based splitting. The adaptive fit-based criterion of \citet{ZeileisModelBasedRP}, mirroring
\citet{athey_recursive_2016}, selects splits by within-leaf
goodness-of-fit of an outcome model $\hat\mu(D_i,X_i;\mathcal{S}^{tr},\Pi)$
consisting of an intercept and a treatment indicator:
\begin{equation}\label{eq:MSE_fit_A}
    \mathrm{MSE}^{\mathrm{fit\text{-}A}}_\mu(\mathcal{S}^{tr},\mathcal{S}^{tr},\Pi)
    \;:=\;\frac{1}{N^{tr}}\sum_{i\in\mathcal{S}^{tr}}
    \left[\left(Y_i-\hat\mu(D_i,X_i;\mathcal{S}^{tr},\Pi)\right)^2-Y_i^2\right].
\end{equation}
Replacing the mean squared residual within each leaf by the median squared residual \citep{rousseeuw_least_1984} yields the LMS criterion,
\begin{equation}\label{eq:LMS_split}
\widehat Q_\mu^{\mathrm{LMS}}
(\mathcal S^{tr},\Pi)
:=
\sum_{l\in\Pi}
\frac{N_l^{tr}}{N^{tr}}
\underset{i\in\mathcal S_l^{tr}}{\operatorname{med}}
\left[
\left(
Y_i-
\hat\mu(D_i,X_i;\mathcal S^{tr},\Pi)
\right)^2
\right],
\end{equation}
where the median squared residual is computed separately within each leaf. Candidate splits are chosen by minimizing \eqref{eq:LMS_split}, or equivalently by maximizing the reduction in the criterion relative to the parent node.
Like \eqref{eq:MSE_fit_A}, LMS targets fit of $Y$ given $(D,X)$ within leaves, not treatment effect heterogeneity and it inherits the limitations of fit-based partitioning discussed in
\citet[\S 3.1.3]{athey_recursive_2016}. 

Similar to the transformed-outcome trees described in
\citet{athey_recursive_2016}, MAD and LMS are best viewed as benchmark splitting rules with a more practical motivation. They do not mimic the honest EMSE derivation of the causal tree criterion, but they provide simple robust alternatives that can be implemented within the same causal forest algorithm. MAD focuses split selection on the absolute deviation between mean-based and rank-based leaf estimates, thereby highlighting partitions whose estimated effects are sensitive to tail observations. LMS robustifies fit-based partitioning by replacing the
mean of squared residuals with their median. 

\section{Simulation Study}
\label{sec:simulations}

The simulation study builds on the non-linear treatment effect setting of
\citet{wager_estimation_2018} and extends the data-generating process to analyze the conditions of
Assumption~\ref{ass:HL_shift}.  We generate the potential outcomes $Y_i(0), Y_i(1)$ as
\begin{align}\label{eq:PO_dgp}
\begin{split}
    Y_i(0) &= \varepsilon_i,\\
    Y_i(1) &= Y_i(0) + \tau(X_i) + U_i,
\end{split}
\end{align}
where $\varepsilon_i$ is the outcome noise, $\tau(X_i)$ is the CATE from Equation \eqref{eq:CATE}, and $U_i$ is a mean-zero individual treatment effect component independent of $(X_i,\varepsilon_i)$. The treatment indicator $D_i$ is randomized as $D_i\sim\operatorname{Bernoulli}(0.5)$ with $D_i\perp X_i$ and we consider a constant propensity score with $Pr(D_i=1|X_i=x)=0.5$.
The covariates are drawn independently as $X_i\sim\mathrm{Uniform}[0,1]^{K}$ as in \citet{wager_estimation_2018} and we use $K=10$ and a sample size of $N=1000$ as defaults throughout this section. In contrast to \citet{wager_estimation_2018}, we vary three components of the data-generating process across scenarios:
\begin{enumerate}
    \item[(a)] The noise distribution $\varepsilon_i$, either standard
    Gaussian or Student's $t_3$,
    \begin{equation*}
        \varepsilon_i\sim\mathcal{N}(0,1)
        \qquad\text{or}\qquad
        \varepsilon_i\sim t_3.
    \end{equation*}
    \item[(b)] The non-linear treatment effect function $\tau(x)$. In the
    non-sparse setting, $\tau(x)=\zeta(x_1)\,\zeta(x_2)$ with, in general, 
    \begin{equation*}
        \zeta(v) = 1 + \frac{1}{1+\exp\!\left(-20\bigl(v-\tfrac{1}{3}\bigr)\right)},
    \end{equation*}
    such that $\tau(x)$ depends multiplicatively on the first two covariates, $x_1$ and $x_2$, only. In the sparse setting, we add
    $10\cdot\mathds{1}(x_1^2+x_2^2>1.2^2)$ to $\tau(x)=\zeta(x_1)\,\zeta(x_2)$. The sparse region $\{x:x_1^2+x_2^2>1.2^2\}$ is the upper-right corner of the unit square in the first two coordinates, so the extreme responders are a minority subpopulation rather than a bulk effect.
    \item[(c)] The ITE component $U_i$, set either to zero (so that
    $Y_i(1)-Y_i(0)$ is deterministic given $X_i$) or drawn from a
    mean-centered log-normal,
    \begin{equation*}
        U_i = 2\left(\exp(Z_i) - \exp(\tfrac{1}{2})\right),\qquad
        Z_i\sim\mathcal{N}(0,1),
    \end{equation*}
    which is right-skewed with $\mathbb{E}[U_i]=0$ but
$\operatorname{med}(U_i)=2\!\left(1-\exp\!\left(\tfrac{1}{2}\right)\right)<0$. The non-zero specification places $U_i$ in the role of the unmeasured effect modifier studied by \citet{post2024flexible} and \citet{PostvandenHeuvel+2025}, in an additive form consistent with their setup. \citet{PostvandenHeuvel+2025} emphasize that flexible CATE estimators silently miss the resulting gap between ITE and CATE whenever $U_i$ is non-degenerate. The independence $U_i\perp(\varepsilon_i,X_i)$ imposed here corresponds to their identification setup, and
the log-normal shape of $U_i$ mirrors an illustrative ITE distribution in \citet{PostvandenHeuvel+2025}.
\end{enumerate}

We consider four scenarios, summarized in Table~\ref{tab:scenarios}.
Scenario \textbf{S1} serves as the benchmark design. It uses Gaussian
outcome noise, a smooth treatment effect with no sparse region, and no residual individual effect component, which places it closest to the idealized leafwise location shift of Assumption~\ref{ass:HL_shift}. The only departure from an exact location shift arises through the smooth variation of $\tau(X_i)$ inside each leaf, and this departure vanishes as the partition refines, so S1 is the setting in which the MSE criterion is expected to perform well.

Scenario \textbf{S2} differs from S1 only in replacing the Gaussian noise with Student's $t_3$ noise. Because $\Var(t_3)=3$, this both thickens the tails and triples the noise variance relative to the unit-variance Gaussian, so S2 stresses robustness on both counts while leaving the treatment effect structure unchanged. Condition~(i) of Assumption~\ref{ass:HL_shift} therefore continues to hold in the same asymptotic sense as in S1, since changing the noise distribution leaves the within-leaf shift structure unchanged.

Scenario \textbf{S3} retains Gaussian noise and $U_i = 0$ but introduces the sparse region of extreme responders in $\tau(x)$. Leaves that cross the responder boundary then contain a bimodal mixture of baseline and extreme effects, which produces sharp heterogeneity within the leaf and a clear departure from the location shift model. Because the boundary is fixed and the offending leaves shrink in measure as the forest refines, this violation is a property of the partition rather than of the data generating process. 

Scenario \textbf{S4} keeps the smooth treatment effect of S1 but adds a skewed component $U_i$ with mean zero. Since $U_i$ is independent of $X_i$ with $\mathbb{E}[U_i] = 0$, the pointwise CATE remains $\tau(x)$, yet the skew inflates and reshapes the treated potential outcome distribution relative to the control distribution, so the mean and median individual treatment effects no longer coincide. In contrast to S3, this violation is structural and persists at every leaf as the sample size grows.\footnote{Besides the violation of condition (i) in Assumption \ref{ass:HL_shift}, scenario~S4 also violates the symmetry condition~(ii) of
Proposition~\ref{prop:Delta_equals_tau_equals_tauMed}, since the skewed $U_i$
makes the within-leaf distribution of $Y_i(1)-Y_i(0)$ asymmetric. Table \ref{tab:scenarios}
records condition~(i) in Assumption \ref{ass:HL_shift} only, for comparability across scenarios.}

\begin{table}[htbp]
\centering
\caption{Simulation scenarios.}
\label{tab:scenarios}
\begin{tabular}{lcccc}
\toprule
Scenario & $\varepsilon_i$ & $\tau(X_i)$ & $U_i$
& Violation of condition (i) in Assumption~\ref{ass:HL_shift} \\
\midrule
\textbf{S1} & $\mathcal{N}(0,1)$ & non-sparse & $0$
& none \\
\textbf{S2} & $t_3$ & non-sparse & $0$
& none \\
\textbf{S3} & $\mathcal{N}(0,1)$ & sparse & $0$
& within-leaf effect heterogeneity at boundary \\
\textbf{S4} & $\mathcal{N}(0,1)$ & non-sparse & skewed
& non-degenerate ITE inflates and skews $F_{1,l}$ \\
\bottomrule
\end{tabular}
\begin{notes}
The error term $\varepsilon_i$, the treatment effect function $\tau(X_i)$, and the ITE component $U_i$ determine the potential outcomes in \eqref{eq:PO_dgp}. The leafwise location-shift condition~(i) of Assumption~\ref{ass:HL_shift} holds only in S1 and S2. The regularity
conditions~(ii) and~(iii) hold by construction in every scenario.
\end{notes}
\end{table}

We illustrate the CATE and ATE results in Figures~\ref{fig:CATE_precision_s1s4} and~\ref{fig:ATE_precision}. Appendix~\ref{app:add_sim_res} reports sensitivity to sample size (Tables~\ref{tab:cate_samplesize} and~\ref{tab:ate_samplesize}) and to covariate dimension (Tables~\ref{tab:cate_covariatesize} and~\ref{tab:ate_covariatesize}). Further details of the design are presented in Appendix~\ref{appendix:MC_Sim} while Table \ref{tab:runtime} compares the algorithm runtime of causal forests with different splitting criteria. Confidence intervals around the causal forest CATE estimates use an analog of the bootstrap of little bags (BLB) \citep{sexton_standard_2009, athey_generalized_2019} with scaling. Appendix~\ref{appendix:BLB} presents the implementation and compares BLB with and without scaling in Table~\ref{tab:blb_scaling_ci}.

The MSD results throughout this section use the implemented surrogate criterion \eqref{eq:Q_MSD_code_main}, which approximates the honest criterion \eqref{eq:honest_median} to leading order. Therefore, the simulations evaluate the practical implementation of MSD rather than the exact formal honest criterion. Figure~\ref{fig:CATE_precision_s1s4} shows that this surrogate MSD splitting criterion attains the lowest absolute bias in all four scenarios. In S1 and S2, which are closest to the location-shift setting in Assumption~\ref{ass:HL_shift} and where the MSE splitting criterion is asymptotically optimal for $\tau(x)$, MSD nevertheless lowers root-mean-square error (RMSE) by roughly 13 to 15 percent relative to the MSE splitting criterion in Table~\ref{tab:cate_samplesize}. The MAD and LMS rules perform worse than both MSE and MSD. MAD penalizes differences between the mean-based and Hodges--Lehmann leaf estimates, but does not reward treatment-effect heterogeneity. LMS chooses splits based on outcome fit rather than treatment-effect heterogeneity. As a result, both rules are less directly aligned with variation in $\tau(x)$ than MSD.

Empirical coverage in Figure~\ref{fig:CATE_precision_s1s4} shows the sharpest contrast between MSD and the alternatives. MSD stays closest to the nominal $0.95$ level in the designs where MSE undercovers (S1, S3), while the MAD and LMS rules undercover already in the rather well-behaved scenario S1 around $0.70$. The undercoverage of MAD and LMS confirms that their corresponding BLB variance estimators are poorly calibrated, consistent with the absence of a clean identification result for those criteria. In S1, MSD overcovers at around $0.99$, and in S4 all four methods overcover, reflecting the additional variance contributed by the skewed $U_i$, which inflates absolute interval widths uniformly across methods. The price of MSD's conservative coverage is visible in the bottom row, where MSD intervals are systematically about $10$ to $15$ percent wider than MSE intervals. Across all four scenarios, MSD therefore trades a modest width premium for improved coverage, while simultaneously delivering the best point estimate precision.
\begin{figure}[htbp]
  \centering
    \caption{Precision and confidence interval coverage of CATE estimates from causal forests with different splitting criteria.}
  \includegraphics[width=\textwidth]{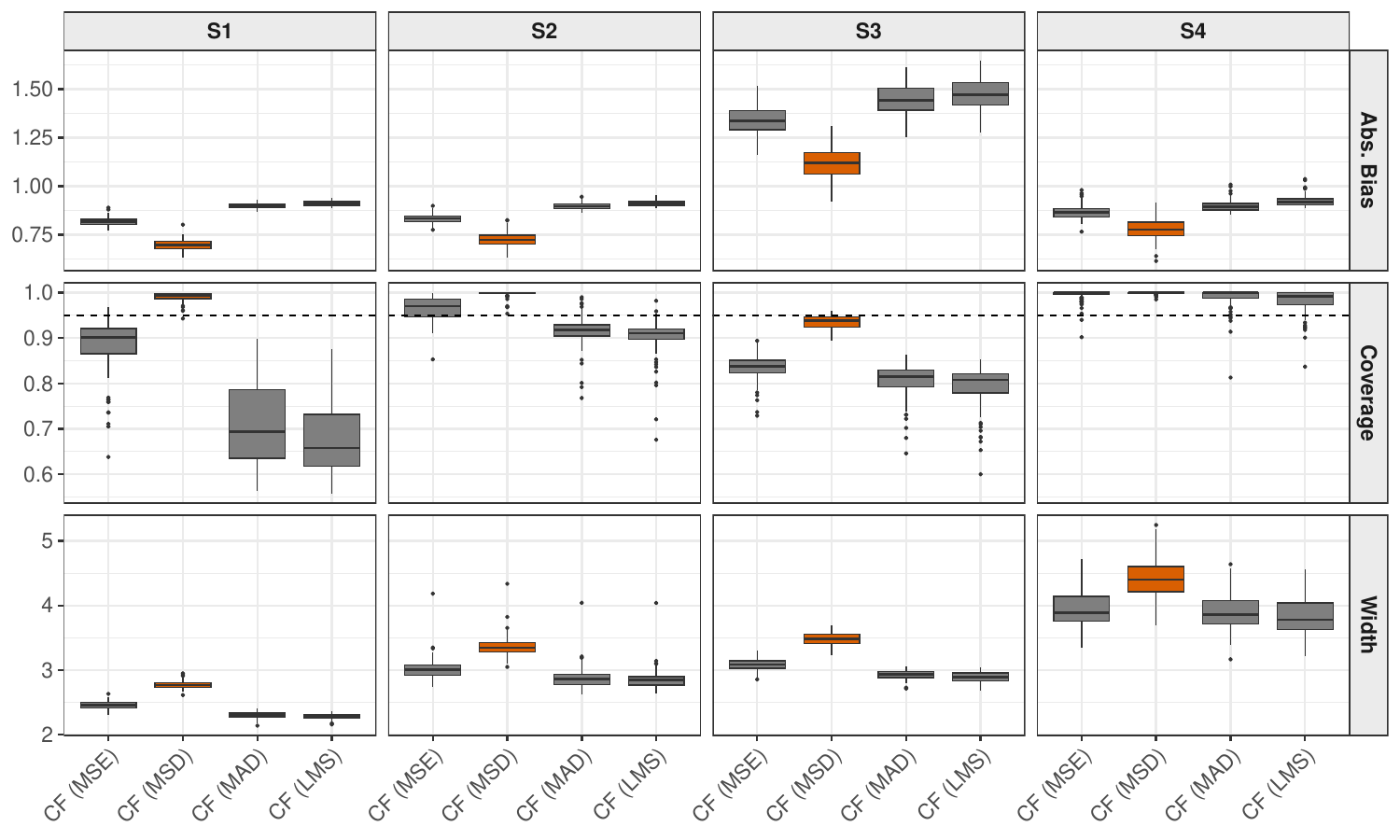}
  \begin{notes}
    Per-replication absolute bias (top row) of CATE estimates as well as empirical coverage (middle row) and width (bottom row) of 95\% confidence intervals for the CATE estimates from causal forests with different splitting criteria across $MC = 100$ Monte Carlo replications ($N = 1000$, $K = 10$). Each panel corresponds to one of the four simulation scenarios defined in Table~\ref{tab:scenarios}. Splitting criteria within the causal forest are: MSE in \eqref{eq:MSE_CT_H}, MSD in \eqref{eq:Q_MSD_code_main} and highlighted in orange, MAD in \eqref{eq:MAD_split}, and LMS in \eqref{eq:LMS_split}. Boxes span the interquartile range (IQR), whiskers extend to $1.5 \times \text{IQR}$ and dots mark outlying replications. Appendix \ref{app:add_sim_res} reports full results in tabular format.
  \end{notes}
  \label{fig:CATE_precision_s1s4}
\end{figure}

Part of MSD's advantage is not directly attributable to robustness. In S1, with Gaussian noise and no outliers to resist, one would expect the Hodges--Lehmann anchor to be marginally less efficient than the mean, since robust location estimators trade efficiency at the Gaussian model for resistance to contamination \citep{huber_robust_2009}. For the Hodges--Lehmann estimator this cost is small, with an asymptotic relative efficiency under normality that is close to that of the mean \citep{hodges_estimates_1963}. The superior performance of MSD in S1 therefore suggests that its conservative split selection, which avoids aggressive splitting in response to small-sample fluctuations, provides a regularization benefit that is distinct from tail robustness. The analysis does not disentangle the effects of regularization and robustness with a separate baseline. The observed advantage in S1 should therefore be interpreted as consistent with regularization rather than as direct evidence for the Hodges--Lehmann anchor. The contribution of robustness is instead indicated by changes in the margin: it is roughly constant across S1 and S2 and widens in S3, precisely the scenario in which asymmetric within-leaf contamination pulls mean-based splits toward spurious heterogeneity. A regularization effect would not concentrate in the contaminated design, so this increment is attributable to the robust anchor rather than to conservative splitting. In S4, the skewed but mean-zero $U_i$ leaves the pointwise CATE equal to $\tau(x)$, so the difference-in-means leaf estimates remain unbiased for $\tau(x)$ under every splitting rule. The skew inflates and reshapes $F_{1,l}$, and the HL splitting anchor may be biased for $\tau(l;\Pi)$ under the failed location shift. However, this affects only which splits are chosen, not the difference-in-means leaf estimates, so the MSD point estimates carry no median-shift bias and MSD is still superior, because its conservative splitting reduces estimation variance.
\begin{figure}[htbp]
  \centering
  \caption{Absolute bias of ATE estimates from causal forests and competing estimators.}
    \label{fig:ATE_precision}
  \includegraphics[width=\textwidth]{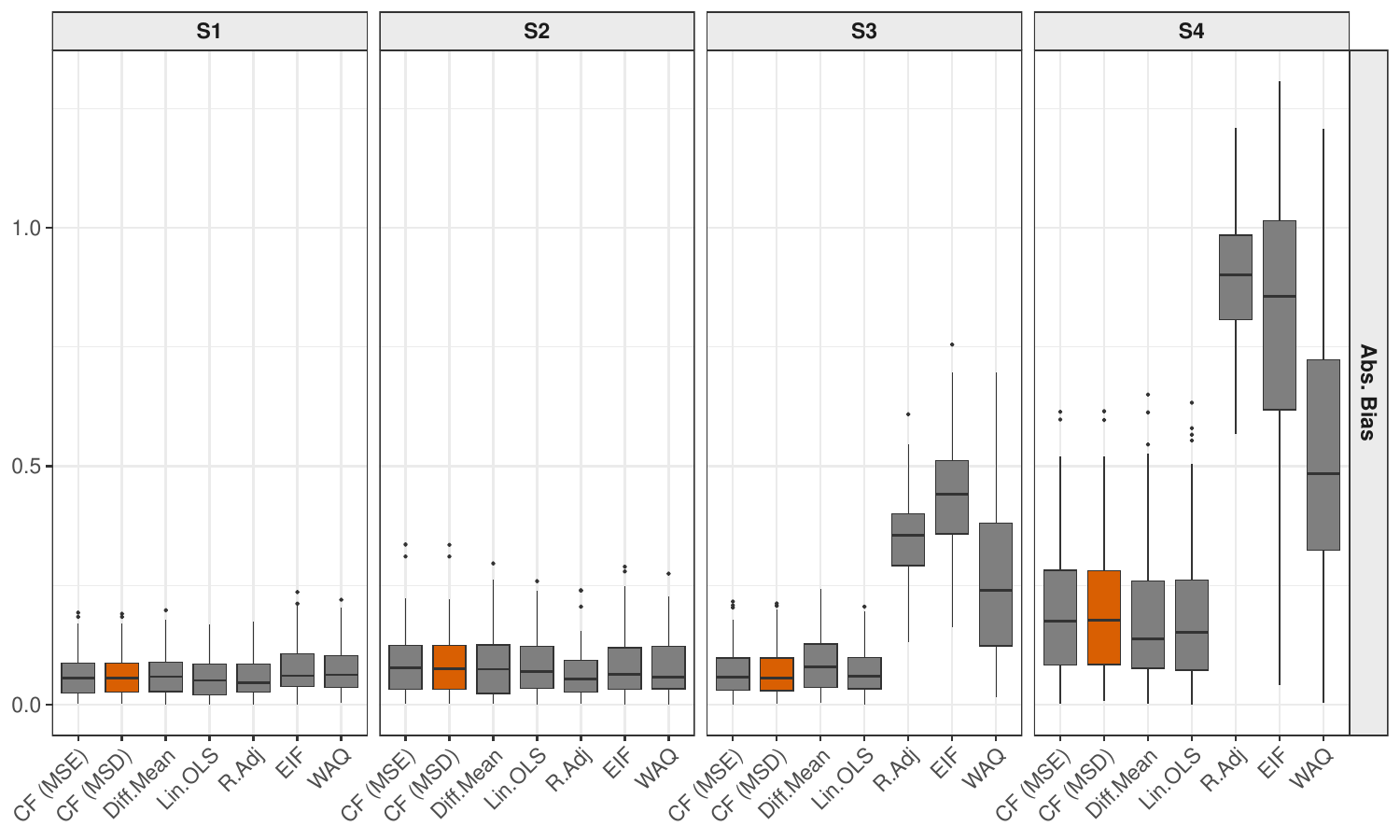}
  \begin{notes}
      Per-replication absolute bias of ATE estimates from causal forests across $MC = 100$ Monte Carlo replications ($N = 1000$, $K = 10$). For CF (MSE) and CF (MSD), we use cross-fitted AIPW \citep{robins_estimation_1994, chernozhukov_doubledebiased_2018} with causal forest CATE estimates as explained in Appendix \ref{appendix:AIPW_ATE_CF}. Further reported ATE estimators are: the difference-in-means estimator (Diff.Mean), OLS with interaction (Lin.OLS) proposed in \cite{lin_agnostic_2013}, the rank-based regression-adjustment estimator (R.Adj) of \cite{ghosh_robustness_2026}, and the efficient-influence-function (EIF) and weighted-average-of-quantiles (WAQ) estimators used in \cite{Athey2023}. We provide further information about these competing estimators in Appendix \ref{appendix:ATE_estimators}. Scenarios S1--S4 are as defined in Table~\ref{tab:scenarios} and the splitting criteria within the causal forests are MSE in \eqref{eq:MSE_CT_H} and MSD (coloured in orange) in \eqref{eq:Q_MSD_code_main}. Boxes span the interquartile range (IQR) and the whiskers extend to $1.5 \times \text{IQR}$ while dots mark outlying replications.
  \end{notes}
\end{figure}

For the ATE, the splitting rule within the causal forest is immaterial for precision and coverage results. CF (MSD) and CF (MSE) give near-identical ATE estimates and confidence interval widths across all four scenarios (Tables~\ref{tab:ate_samplesize} and~\ref{tab:ate_covariatesize} in Appendix \ref{app:add_sim_res}), since the splitting choice washes out once CATE estimates are aggregated by AIPW and is further discussed in Appendix \ref{appendix:AIPW_ATE_CF}. This holds true even for the noisier CATE estimates of CF (MAD) and (LMS), which are not reported here for brevity due to their similar ATE estimates, comparable to CF (MSD) and CF (MSE).
The informative contrast that we present here is with the competing estimators from the simulation study in \citet{ghosh_robustness_2026} and described in Appendix \ref{appendix:ATE_estimators}. The difference-in-means estimator and OLS with interactions \citep{lin_agnostic_2013} stay
centered on the mean ATE in every scenario, holding nominal coverage in S1 to
S3 and over-covering in S4 along with the causal forests, as the heavy $U_i$ inflates all variances. The rank- and quantile-based estimators (Rosenbaum's adjustment (R.Adj) \citep{ghosh_robustness_2026}, efficient-influence-function (EIF) and weighted-average-of-quantiles (WAQ) estimators \citep{Athey2023}) instead lose coverage in S3 and S4. These undercoverage issues are an estimand mismatch, not an efficiency loss. Each of the rank- and quantile-based estimators targets a median, rank, or quantile contrast that equals the mean ATE only under symmetry, and the sparse responders of S3 and the skewed individual effects of S4 break that equality, so the intervals concentrate around the wrong target. Both causal forests stay centered on the mean ATE because MSD confines its robustness to split selection and estimates leaves by the difference in means. Robust partitioning therefore buys outlier resistance without the estimand drift that pure median- or quantile-based estimators incur under skew.

\section{Empirical Applications}
\label{sec:empappl}

We provide two empirical applications that instantiate the skewness contrast of the simulation study. The Progresa vote share \citep{de_la_o_conditional_2013} in the first empirical study in Subsection \ref{sec:progresa} is a strongly right-skewed outcome variable. This case aligns with scenario~S4 of Table \ref{tab:scenarios} in our simulation study and predicts the choice of the splitting rule to reshape the CATE distribution, while the AIPW-aggregated ATE is unaffected. However, rank- and quantile-based estimators drift from the mean ATE due to estimand mismatch. The second application in Subsection \ref{sec:ACTG_175} re-visits a study of antiretroviral treatments in HIV-positive adults \citep{hammer_trial_1996, leqi_median_2022}. The outcome variable, CD4
count, is closer to symmetric and matches the scenarios in cases S1 and S2 of Table \ref{tab:scenarios}, in which all methods are expected to agree regarding ATE estimation results. Appendix \ref{append:emp_appl} provides supplementary material for both empirical applications.

\subsection{Progresa}
\label{sec:progresa}

We examine data from the randomized rollout of Mexico’s conditional cash transfer program, Progresa, which was designed as a social policy intervention and provides an opportunity to study its political consequences, as analyzed in \cite{de_la_o_conditional_2013}. In the original experiment, eligible villages were randomly assigned to begin receiving program benefits either well before (early treatment) or shortly before (delayed control) the 2000 presidential election, thereby generating plausibly exogenous variation in program exposure at the local level. Following the empirical strategy in \cite{de_la_o_conditional_2013}, we aggregate the data to the precinct level (417 observations) and focus on electoral support for the incumbent party in the 2000 election, measured as its vote share among eligible voters as the outcome variable. Figure \ref{fig:outcome_distribution_progresa} in Appendix \ref{append:emp_appl} illustrates the right-skewed nature of the target variable. 

We note that the Progresa outcome is a vote share bounded in $[0,100]$, so the leafwise location-shift condition~(i) of Assumption~\ref{ass:HL_shift} cannot hold exactly: a pure shift would move probability mass beyond the support boundary.
By the leaf-estimation argument of Section~\ref{sec:MSD}, this does not bias the MSD point estimates, which are differences in means on the estimation sample, and degrades only the interpretation of the splitting anchor to a robust proxy for $\tau(l;\Pi)$.
The analysis adjusts for pre-treatment socioeconomic and political characteristics, including measures of poverty, historical population size, prior voter turnout, and past party vote shares, along with village fixed effects, to improve precision and account for baseline differences. This design leverages the randomized timing of program implementation to identify the treatment effect of exposure to Progresa on subsequent electoral behavior.
\cite{ghosh_robustness_2026} reanalyze the study in \cite{de_la_o_conditional_2013} to apply its proposed regression-adjusted estimator based on \cite{Rosenbaum1993}. In our analysis, reported in Table~\ref{tab:progresa_estimates}, the difference-in-means estimate of 3.62 is borderline, with a 95\% interval of $[-0.05, 7.30]$ that only just includes zero, and OLS with interaction \citep{lin_agnostic_2013} is the single estimator returning a significant positive effect at 4.21 with interval $[0.24, 8.19]$.
Every estimator that adjusts for the skewed outcome attenuates the effect to between $1.3$ and $2.2$ and finds it statistically insignificant, including Rosenbaum's regression adjustment, the EIF and WAQ estimators, and both causal forests. Rosenbaum's adjustment gives the narrowest interval of these. The two forests agree closely, at $1.47$ for MSE and $1.48$ for MSD, which matches the simulation finding in Figure \ref{fig:ATE_precision} that the splitting rule has little effect on the aggregated ATE estimate. 
That a statistically significant positive treatment effect appears only under an estimator that is exposed to the skewed tail and vanishes under every robust adjustment matches the broader reassessment of this experiment. \citet{imai_nonpartisan_2020} reanalyze the Progresa rollout of \citet{de_la_o_conditional_2013} under a bias-corrected setup and find no electoral effect on incumbent support, so our robust estimates should be read as further evidence that median-based splitting does not manufacture an outlier-driven treatment effect.
\begin{table}[htbp]
\centering
\caption{ATE estimates of the effect of early Progresa on PRI support rates with corresponding standard errors, 95\% confidence intervals, and interval widths.}
\label{tab:progresa_estimates}
\begin{tabular}{lcccc}
\hline\hline
Estimator & Estimate & Std.\ Error & 95\% CI & CI Width \\
\hline
Diff.Mean          & 3.622 & 1.875 & $[-0.052,\; 7.297]$ & 7.348 \\
Lin.OLS \citep{lin_agnostic_2013}    & 4.214 & 2.027 & $[\phantom{-}0.240,\; 8.187]$ & 7.947 \\
R.Adj \citep{ghosh_robustness_2026}  & 2.185 & 1.338 & $[-0.439,\; 4.808]$ & 5.246 \\
EIF \citep{Athey2023}     & 1.953 & 1.648 & $[-1.276,\; 5.182]$ & 6.459 \\
WAQ \citep{Athey2023}      & 1.306 & 1.632 & $[-1.892,\; 4.504]$ & 6.396 \\
$\text{CF (MSE)}$ \citep{wager_estimation_2018}    & 1.474 & 1.485 & $[-1.436,\; 4.385]$ & 5.822 \\
$\text{CF (MSD)}$         & 1.482 & 1.487 & $[-1.432,\; 4.396]$ & 5.828 \\
\hline\hline
\end{tabular}
\begin{notes}
ATE estimates for the Progresa application in Section~\ref{sec:progresa}. Estimator definitions are given in Appendix~\ref{appendix:ATE_estimators}.
\end{notes}
\end{table}

Figure~\ref{fig:histogram_cate_progresa} shows where the two causal forests differ in terms of CATE estimates. Both place the bulk of precinct-level CATE estimate distribution near the small positive ATE estimate, but the shapes differ. MSE concentrates its estimates on two sharp modes, while MSD spreads them more evenly and carries more mass into both tails. The discreteness of the MSE estimates reflects mean-based splitting settling on a small number of leaf values, whereas MSD's smoother spread follows from its more conservative partitioning. Neither distribution gives evidence of strong structured heterogeneity, and the agreement in location matches the near-identical ATEs. Only the shape and the tails separate the two, which is where the skewed outcome makes the splitting rule matter. 
We further note that the pointwise confidence intervals behind CATE estimates in Figure \ref{fig:histogram_cate_progresa}, reported in Figure~\ref{fig:plot_cate_progresa} of Appendix \ref{append:emp_appl}, are wide and almost all contain the ATE estimate, so the causal forest CATE estimates for the Progresa data provide no evidence of substantial heterogeneity.
\begin{figure}[htbp]
  \centering
  \caption{Distribution of precinct-level CATE estimates for the Progresa data.}  
  \includegraphics[width=\textwidth]{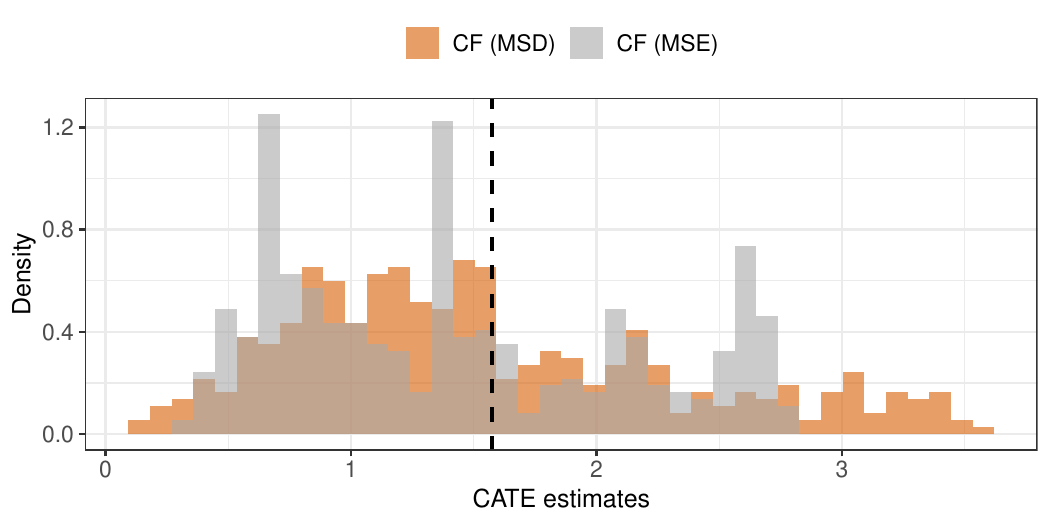}
\label{fig:histogram_cate_progresa}
\begin{notes}
    Causal forests with MSD (orange) and MSE (grey) splitting. The dashed, vertical line in black marks the similar ATE estimate of both causal forests.
\end{notes}
\end{figure}

\subsection{ACTG 175}
\label{sec:ACTG_175}

The second empirical study applies our methods to data from a randomized trial studying antiretroviral treatments in HIV-positive adults \citep{hammer_trial_1996, leqi_median_2022, juraska_speff2trial_2022}. Participants were assigned either to zidovudine monotherapy ($D_i=0$) or to alternative antiretroviral therapies ($D_i=1$), and their CD4 (also known as T helper cells) count was recorded approximately 96 weeks post-randomization. As in the Progresa application, condition~(i) of Assumption~\ref{ass:HL_shift}
is not directly verifiable, though the CD4 count is unbounded above and less obviously in conflict with a location-shift model than the bounded vote share.
We adjust for a rich set of baseline characteristics, including immune markers (baseline CD4 and CD8 counts), demographics (age, weight, race, gender), clinical indicators (Karnofsky score, symptomatic status, hemophilia), behavioral variables (homosexual activity, drug use history), and prior treatment exposure (zidovudine and antiretroviral use). Table \ref{tab:desc_stats_actg175} reports descriptive statistics of the considered variables. The randomized design yields a known, constant propensity score of $p(x) = 0.75$ across all subjects. Dropping individuals with missing 96-week measurements leaves an analysis sample of $N=1{,}342$. Dropout in ACTG 175 is substantial and has been linked to baseline characteristics, so this complete-case analysis rests on the outcome being missing at random given the covariates we adjust for, which we take as a maintained assumption rather than a tested one.

The ACTG 175 results for the ATE estimates in Table \ref{tab:ate_antiretroviral} are more uniform compared to the results for the Progresa data in Subsection \ref{sec:progresa}. Every ATE estimator finds a large and significant increase in CD4 count under the alternative antiretroviral therapies, with point estimates ranging from about $46$ for WAQ to $67$ for the causal forests. The two causal forests again agree almost exactly, at $66.60$ for MSE and $66.62$ for MSD. The outcome here is less skewed than the Progresa vote share, and the robust and mean-based estimators move together rather than splitting apart.
Figure~\ref{fig:histogram_cate_actg175} confirms the rather uniform findings for causal forest ATE estimates also at the CATE level. The MSD and MSE distributions are both unimodal, centered near the ATE, and overlap almost entirely, with MSD only slightly more dispersed toward lower values. Individual-level CATEs spread from roughly $50$ to $82$, but the two splitting rules are nearly indistinguishable. Where Progresa shows the rules diverging in shape, ACTG 175 shows them coinciding. Similar to Subsection \ref{sec:progresa}, we note that the pointwise confidence intervals of the CATE estimates, reported in Figure~\ref{fig:plot_cate_actg175} of Appendix \ref{append:emp_appl}, are wide and almost all contain the ATE estimates. At the resolution of the conservative scaled-BLB intervals, we cannot reject a constant treatment effect. 
We interpret this as a statement about the precision of our procedure on the outcome variable measuring CD4 at 96 weeks, not as a claim that treatment effects are homogeneous. This result is consistent with \citet{leqi_median_2022}, who also study the ACTG 175 data and likewise find only limited evidence of effect heterogeneity on this outcome.

\begin{table}[htbp]
\centering
\caption{ATE estimates of the effect of antiretroviral treatments in HIV-positive adults with corresponding standard errors, 95\% confidence intervals, and interval widths.}
\label{tab:ate_antiretroviral}
\begin{tabular}{lcccc}
\hline\hline
Estimator & Estimate & Std.\ Error & 95\% CI & CI Width \\
\hline
Diff.Mean          & 53.830 & 10.901 & $[32.465,\; 75.195]$ & 42.730 \\
Lin.OLS \citep{lin_agnostic_2013} & 63.485 &  8.798 & $[46.240,\; 80.730]$ & 34.489 \\
R.Adj \citep{ghosh_robustness_2026} & 62.188 &  8.709 & $[45.118,\; 79.258]$ & 34.140 \\
EIF \citep{Athey2023}     & 49.146 & 10.749 & $[28.078,\; 70.214]$ & 42.135 \\
WAQ \citep{Athey2023}     & 46.052 & 10.996 & $[24.500,\; 67.605]$ & 43.105 \\
$\text{CF (MSE)}$ \citep{wager_estimation_2018}  & 66.603 & 10.265 & $[46.484,\; 86.722]$ & 40.238 \\
$\text{CF (MSD)}$         & 66.621 & 10.271 & $[46.489,\; 86.753]$ & 40.264 \\
\hline\hline
\end{tabular}
\begin{notes}
ATE estimates for the ACTG 175 application in Section~\ref{sec:ACTG_175}. Estimator definitions are given in Appendix~\ref{appendix:ATE_estimators}.
\end{notes}
\end{table}

\begin{figure}[htbp]
  \centering
  \caption{Distribution of CATE estimates for the ACTG 175 data.}
  \includegraphics[width=\textwidth]{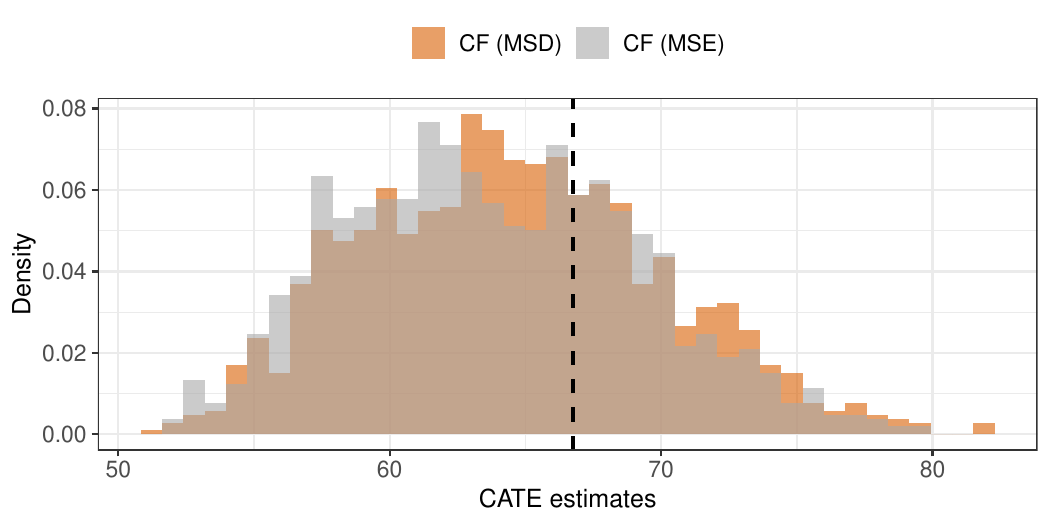}
    \label{fig:histogram_cate_actg175}
    \begin{notes}
        Causal forests with MSD (orange) and MSE (grey) splitting. The dashed, vertical line in black marks the similar ATE estimate of both causal forests.
    \end{notes}
\end{figure}

\section{Conclusion}
\label{sec:conclusion}

Heavy-tailed and skewed outcomes are common in economic and biomedical
applications and standard CATE estimators might be imprecise under these circumstances \citep{li_robust_2021, ghosh_robustness_2026, Athey2023}. 
Instead of relying on the mean-squared-error splitting rule \citep{athey_recursive_2016, wager_estimation_2018}, we have introduced a median-based splitting rule within the causal forest estimator that targets the CATE while resisting heavy-tailed and contaminated outcomes. 
Under the leafwise location-shift model, the Hodges--Lehmann estimator \citep{hodges_estimates_1963} is consistent for the leaf effect and median unbiased, and the resulting Median Squared Deviation (MSD) criterion in Section \ref{sec:MSD} preserves the variance-bias structure of the mean-based rule. An implementable cross-product form of the criterion avoids density estimation during tree growth. Two further rules, MAD and LMS, accompany it as robust baselines without an analogous honest MSE-based derivation.

The simulation study in Section \ref{sec:simulations} compares causal forests under the four splitting criteria across the four scenarios of Table \ref{tab:scenarios}.
For the CATE, MSD attains the lowest RMSE and absolute bias across all four designs. Its coverage is closest to nominal in the designs where the mean-based rule undercovers, most clearly under the sparse extreme responders of scenario S3, while it overcovers conservatively in the more well-behaved designs. 
For the ATE the splitting choice washes out under AIPW, so MSD and MSE coincide and both remain centered on the true ATE. The rank- and quantile-based competitors instead drift to a different estimand once skew breaks the symmetry under which their targets agree with the mean. Robust partitioning buys resistance to outliers in split selection without the estimand drift that pure median- or quantile-based estimators incur under skew. The two applications in Section \ref{sec:empappl}, chosen for their skewed outcomes, reproduce this pattern. MSD and MSE agree on the ATE estimate, and the wide CATE confidence intervals provide no evidence of heterogeneity distinguishable from a constant effect.

Several limitations qualify our presented results. The implemented MSD criterion is the cross-product surrogate. It estimates the leafwise second moment by the squared difference in means and is exact only to leading order, trading a small bias for the avoidance of exact density estimation. Identifying the Hodges--Lehmann anchor with the leaf CATE rests on the within-leaf location-shift condition, which fails under the bimodal and skewed leaf distributions of S3 and S4 in Table \ref{tab:scenarios}. 
Where it fails, the criterion still selects stable splits, but its target degrades to a robust proxy for the leaf effect. Robustness also carries a cost in conservatism, as MSD intervals are systematically wider than the mean-based intervals and overcover in the well-behaved designs.
We further note that our scaled BLB confidence intervals around the causal forest CATE estimates (see Appendix \ref{appendix:BLB}) are rather conservative. An exact calibration against alternative variance estimators remains for future work.
Finally, the Hodges--Lehmann anchor is rebuilt from all within-leaf pairwise differences at each candidate split. Table \ref{tab:runtime} in Appendix \ref{app:add_sim_res} shows that this computational effort makes MSD (and MAD) several times slower than the mean-based MSE, an overhead that a selection algorithm avoiding the pairwise set might remove without altering any result \citep{monahan_algorithm_1984}. Evaluating the density-penalty form and broadening the analysis to observational designs with estimated propensity scores are the natural next steps for further extensions.

\acks{We would like to thank Christoph Hanck for valuable feedback on earlier drafts of this manuscript and we are grateful to seminar participants at the RuhrMetrics Research Seminar and the EuroCIM 2023 for helpful comments and discussions. The authors acknowledge partial financial support from TRR 391 Spatio-temporal Statistics for the Transition of Energy and Transport (520388526) by the Deutsche Forschungsgemeinschaft (DFG, German Research Foundation) and from the Rhine-Ruhr Center for Scientific Data Literacy (DKZ.2R) by the German Federal Ministry of Education and Research (BMBF). During the preparation of this manuscript, we used Anthropic's and OpenAI's large language models (up to and including Claude Opus 4.8 and GPT-5.5) to assist with proofreading and editing. We reviewed and verified all resulting content and are solely responsible for any remaining errors.
}

\bibliography{paper_bibtex_file}

\appendix

\section{Additional Theory}
\subsection{MSE-based Honest Splitting Criterion in Causal Trees}
\label{appendix:Honest_Splitting_CT}

We follow \cite{athey_recursive_2016} and describe the honest splitting criterion in the causal tree by following the setup in Section \ref{sec:PO_framework}.

Let us restate the MSE in \eqref{eq:MSE_treatment_effect} and note that the individual treatment effect $\tau_i = Y_i(1) - Y_i(0)$ is never observed. To handle this, the MSE for treatment effects, evaluated on the test sample with estimates based on a distinct estimation sample, is
\begin{equation}\label{eq:mse_tau}
    \text{MSE}_\tau(\mathcal{S}^{te}, \mathcal{S}^{est}, \Pi) = \frac{1}{N^{te}} \sum_{i \in \mathcal{S}^{te}} \Bigl[\bigl(\tau_i - \hat{\tau}(X_i; \mathcal{S}^{est}, \Pi)\bigr)^2 - \tau_i^2\Bigr].
\end{equation}
The subtraction of $\tau_i^2$ is essential: since $\tau_i$ is never observed, the raw squared error $(\tau_i - \hat{\tau})^2$ cannot be computed. However, $\tau_i^2$ does not depend on $\Pi$ or on $\mathcal{S}^{est}$, so subtracting it preserves the relative ranking of candidate trees and minimizing \eqref{eq:mse_tau} over $\Pi$ yields the same optimum as minimizing the raw squared error without $\tau_i^2$. The remaining terms are estimable. Let $\text{EMSE}_\tau(\Pi)$ denote the expectation of \eqref{eq:mse_tau} over both $\mathcal{S}^{te}$ and $\mathcal{S}^{est}$.

To derive a tractable expression for $\text{EMSE}_\tau(\Pi)$, we first expand the square inside \eqref{eq:mse_tau}. Writing $\hat{\tau}_l \equiv \hat{\tau}(X_i; \mathcal{S}^{est}, \Pi)$ for the leaf-constant estimate assigned to unit $i$:
\begin{equation}\label{eq:expand}
    (\tau_i - \hat{\tau}_l)^2 - \tau_i^2 
    = \tau_i^2 - 2\tau_i\hat{\tau}_l + \hat{\tau}_l^2 - \tau_i^2 
    = -2\tau_i\hat{\tau}_l + \hat{\tau}_l^2.
\end{equation}
Since $\hat{\tau}_l$ is constant within each leaf, the MSE can be grouped by leaves $l \in \Pi$:
\begin{equation}\label{eq:grouped}
    \text{MSE}_\tau(\mathcal{S}^{te}, \mathcal{S}^{est}, \Pi) = \sum_{l \in \Pi} \frac{N_l^{te}}{N^{te}} \left[\hat{\tau}_l^2 - 2\hat{\tau}_l \cdot \frac{1}{N_l^{te}} \sum_{\substack{i \in \mathcal{S}^{te} \\ X_i \in l}} \tau_i\right].
\end{equation}

Now take the expectation over the independent samples $\mathcal{S}^{te}$ and $\mathcal{S}^{est}$. Note that $N_l^{te}/N^{te} \to p_l$, the population probability of landing in leaf $l$ and distinct from the treated share $p$. The
cross-term $\hat{\tau}_l \cdot \frac{1}{N_l^{te}} \sum_{i} \tau_i$ involves one factor from
$\mathcal{S}^{est}$ and one from $\mathcal{S}^{te}$. Since these samples are independent,
their expectations factorize:
\begin{equation}\label{eq:factorise}
    \mathbb{E}_{\mathcal{S}^{est}, \mathcal{S}^{te}}\!\left[\hat{\tau}_l \cdot \frac{1}{N_l^{te}}
    \sum_{\substack{i \in \mathcal{S}^{te} \\ X_i \in l}} \tau_i\right]
    = \mathbb{E}_{\mathcal{S}^{est}}[\hat{\tau}_l] \;\cdot\;
      \mathbb{E}_{\mathcal{S}^{te}}\!\left[\frac{1}{N_l^{te}}
      \sum_{\substack{i \in \mathcal{S}^{te} \\ X_i \in l}} \tau_i\right].
\end{equation}
The two factors on the right-hand side both equal $\tau(l;\Pi)$, but for distinct reasons. The
second factor equals $\tau(l;\Pi)$ by random sampling: test units in leaf $l$ are a random draw
from the subpopulation with $X_i \in l$, so their average treatment effect equals the population
leaf CATE. The first factor equals $\tau(l;\Pi)$ by the honesty construction: because
$\mathcal{S}^{est}$ is independent of the splits that defined the leaves, the estimator
$\hat{\tau}_l$ is unbiased as $\hat{\tau}$ is computed entirely on $\mathcal{S}^{est}$, while the tree $\Pi$ was built entirely on $\mathcal{S}^{tr}$. Because $\mathcal{S}^{est}$ is independent of the splits, honesty guarantees unbiasedness:
\begin{equation}\label{eq:unbiased}
    \mathbb{E}_{\mathcal{S}^{est}}\bigl[\hat{\tau}(x; \mathcal{S}^{est}, \Pi)\bigr] = \tau(x;\Pi).
\end{equation}

Combining these observations,
\eqref{eq:grouped} becomes
\begin{align}\label{eq:emse_intermediate}
    \text{EMSE}_\tau(\Pi) &= \sum_{l \in \Pi} p_l
    \Bigl[\mathbb{E}[\hat{\tau}_l^2] - 2\,\tau(l;\Pi)^2\Bigr] \\
    &= \sum_{l \in \Pi} p_l
       \Bigl[\Var(\hat{\tau}_l) + \tau(l;\Pi)^2 - 2\,\tau(l;\Pi)^2\Bigr] \nonumber \\[6pt]
    &= \sum_{l \in \Pi} p_l
       \Bigl[\Var\bigl(\hat{\tau}_l\bigr) - \tau(l;\Pi)^2\Bigr], 
       \label{eq:emse_final}
\end{align}
where we use the identity $\mathbb{E}[\hat{\tau}_l^2] = \Var(\hat{\tau}_l) +
(\mathbb{E}[\hat{\tau}_l])^2 = \Var(\hat{\tau}_l) + \tau(l;\Pi)^2$. 
This is the central decomposition. A good tree makes $\tau(l;\Pi)^2$ large across leaves (capturing genuine treatment effect heterogeneity) while keeping $\Var(\hat{\tau}_l)$ small (ensuring precise estimation within each leaf). Note that honesty is essential for this result: the unbiasedness property was used to collapse $-2\tau_i\hat{\tau}_l$ in \eqref{eq:expand} into $-2\,\tau(l)^2$ in \eqref{eq:emse_intermediate}. Without honesty, one constructs the tree and estimates its leaves based on the same data and we have $\mathbb{E}[\hat{\tau}_l] \neq \tau(l)$ due to adaptive bias. Therefore, the clean decomposition in Equation \eqref{eq:emse_final} would not hold.

The variance of the leaf-level CATE estimator depends on within-leaf outcome variances. Since
treated and control means are estimated independently within each leaf,
\begin{equation}\label{eq:var_tau}
    \Var(\hat{\tau}_l)
    = \Var\bigl(\hat{\mu}(1,l;\mathcal{S}^{est})\bigr)
    + \Var\bigl(\hat{\mu}(0,l;\mathcal{S}^{est})\bigr)
    = \frac{\sigma^2_{1}(l)}{N_{l,1}^{est}}
    + \frac{\sigma^2_{0}(l)}{N_{l,0}^{est}},
\end{equation}
where $\sigma^2_{1}(l) \equiv \Var(Y_i(1) \mid X_i \in l)$ and
$\sigma^2_{0}(l) \equiv \Var(Y_i(0) \mid X_i \in l)$ are the population
conditional variances of potential outcomes within leaf $l$. Note that $\Var(\cdot)$ appears
here in two roles: on the left-hand side of \eqref{eq:var_tau} it denotes the sampling variance of
the estimator $\hat{\tau}_l$ across repeated draws of $\mathcal{S}^{est}$, while on the right-hand
side it defines the population outcome variances $\sigma^2_{1}(l)$ and
$\sigma^2_{0}(l)$. These population variances are unknown and are estimated by the
within-leaf sample variances $\hat{S}^2_{1}(l;\mathcal{S}^{tr})$ and
$\hat{S}^2_{0}(l;\mathcal{S}^{tr})$, computed on the training sample.

Under random assignment with treatment probability $p$, the expected counts satisfy $N_{l,1}^{est} \approx p\,N_l^{est}$ and $N_{l,0}^{est} \approx (1-p)\,N_l^{est}$, so that
\begin{equation}\label{eq:var_tau_simplified}
    \Var(\hat{\tau}_l) \approx \frac{1}{N_l^{est}} \left[\frac{\sigma^2_{1}(l)}{p} + \frac{\sigma^2_{0}(l)}{1-p}\right].
\end{equation}

Substituting \eqref{eq:var_tau_simplified} into \eqref{eq:emse_final} gives the full criterion:
\begin{equation}\label{eq:emse_full}
    \text{EMSE}_\tau(\Pi) = \sum_{l \in \Pi} p_l \left[\frac{1}{N_l^{est}} \left(\frac{\sigma^2_{1}(l)}{p} + \frac{\sigma^2_{0}(l)}{1-p}\right) - \tau(l;\Pi)^2\right].
\end{equation}
Minimizing $\text{EMSE}_\tau(\Pi)$ over candidate trees is equivalent to maximizing
\begin{equation}\label{eq:max_criterion}
    \sum_{l \in \Pi} p_l\,\tau(l;\Pi)^2 - \sum_{l \in \Pi} p_l \cdot \frac{1}{N_l^{est}} \left(\frac{\sigma^2_{1}(l)}{p} + \frac{\sigma^2_{0}(l)}{1-p}\right).
\end{equation}
The first sum rewards heterogeneity (large squared CATEs across leaves) while the second sum penalizes imprecision (high variance from small or noisy leaves).

When building the tree, candidate splits must be evaluated using
$\mathcal{S}^{tr}$ alone, since $\mathcal{S}^{est}$ is reserved for estimation
after the tree is fixed. The population leaf probability $p_l$ is estimated by
$N_l^{tr}/N^{tr}$, and the conditional outcome variances $\sigma_d^2(l)$ by
$\hat S_d^2(l;\mathcal{S}^{tr})$. For the squared leaf CATE $\tau(l;\Pi)^2$
the natural plug-in $\hat\tau(l;\mathcal{S}^{tr})^2$ is upward biased,
because
\begin{align}\label{eq:plugin_bias_mean}
\begin{split}
    \mathbb{E}_{\mathcal{S}^{tr}}\!\left[\hat\tau(l;\mathcal{S}^{tr})^2 \,\middle|\, \Pi\right]
    &\;=\; \tau(l;\Pi)^2 \;+\; \Var_{\mathcal{S}^{tr}}\!\bigl(\hat\tau(l;\mathcal{S}^{tr})\bigr) \\
    &\;\approx\; \tau(l;\Pi)^2 \;+\; \frac{1}{N_l^{tr}}
        \!\left(\frac{\sigma_1^2(l)}{p} + \frac{\sigma_0^2(l)}{1-p}\right).
        \end{split}
\end{align}
Subtracting the estimated bias $\hat V_l/N_l^{tr}$ yields the unbiased plug-in
$\hat\tau(l;\mathcal{S}^{tr})^2 - \hat V_l/N_l^{tr}$ for $\tau(l;\Pi)^2$.
Combined with the estimation-sample variance penalty from \eqref{eq:var_tau_simplified} and based on \eqref{eq:max_criterion}, the
honest splitting criterion is
\begin{equation}\label{eq:honest_criterion}
  \hat Q^{\,\text{MSE, H}}_\tau(\Pi;\mathcal{S}^{tr})
  = \sum_{l \in \Pi} \frac{N_l^{tr}}{N^{tr}}
    \left[\hat\tau(l;\mathcal{S}^{tr})^2
    - \left(\frac{1}{N_l^{tr}} + \frac{1}{N_l^{est}}\right)
      \!\left(\frac{\hat S_1^2(l;\mathcal{S}^{tr})}{p}
      + \frac{\hat S_0^2(l;\mathcal{S}^{tr})}{1-p}\right)\right],
\end{equation}
to be maximized over $\Pi$.

\subsection{Foundations of the Hodges--Lehmann estimator}
\label{appendix:HL_estimator}

This section reviews the basic two-sample definition of the HL estimator \citep{hodges_estimates_1963} and its relation to the Wilcoxon rank-sum statistic \citep{hollander_nonparametric_2014}. The insights presented in this section are of a general nature, independent of the tree or leaf terminology, which is included in the main analysis.  
Let
$Y^1 = \{ Y_j \,|\, j \in \mathcal{S}_1 \} \text{ and } Y^0 = \{ Y_m \,|\, m \in \mathcal{S}_0 \}$  be the observed outcomes. The first step is to establish a unique median for all continuous distributions.
For completeness, we define the median of a distribution function $F$ as follows. 
\begin{definition}\label{median_distribution}
Let $\theta_1 = \inf \{\theta:F(\theta)=\tfrac{1}{2} \}$ and $\theta_2 = \sup \{\theta:F(\theta)=\tfrac{1}{2} \}$. Then, the median $\theta$ of a distribution $F$ is defined as,\footnote{While the definition ensures the theoretical uniqueness of a median,  we will use a computationally faster definition in our simulation studies in Section \ref{sec:simulations} by
$\theta \coloneqq \operatorname{min} \{ \theta_1, \theta_2 \}.$ Thereby, we resort to the Quickselect algorithm of \cite{hoare_algorithm_1961} as described in \cite{numerical_press_1992} to approximate the median in reasonable computing time.}
\begin{align} \label{eq:unique_median}
\theta \coloneqq \frac{(\theta_1 +  \theta_2)}{2}. 
\end{align}  
\end{definition}
In the classical two-sample setting, the HL estimator is associated with a location
shift model. Let $F_1$ and $F_0$ denote the outcome distributions of the treatment and control
group. The aim is to determine whether the two groups differ only by an unknown shift parameter
$\Delta$, so that
\[
F_1(z)=F_0(z-\Delta)
\qquad \text{for all } z\in\mathbb R.
\]
Given a test statistic
\[
\mathcal T=t(Y^0,Y^1)=t(Y^0_1,\dots,Y^0_{N_0};Y^1_1,\dots,Y^1_{N_1}),
\]
the null hypothesis is
\[
\mathbb H_0:\Delta=0
\]
against an appropriate one-sided or two-sided alternative.

\begin{definition}\label{def:HL_general}
Let
\[
\Delta^*=\sup\{\Delta:t(Y^0,Y^1-\Delta)>\xi\},
\quad \text{and} \quad
\Delta_*=\inf\{\Delta:t(Y^0,Y^1-\Delta)<\xi\},
\]
where the statistic $t(Y^0,Y^1)$ satisfies:
\begin{itemize}
    \item[$i)$] $t(Y^0,Y^1+a)$ is non-decreasing in $a$;
    \item[$ii)$] under $\Delta=0$, the distribution of $t(Y^0,Y^1)$ is symmetric about a fixed point $\xi$.
\end{itemize}
Then the Hodges--Lehmann estimator associated with $\mathcal{T}$ is defined as
\[
\hat\Delta_{\mathrm{HL}}=\frac{\Delta^*+\Delta_*}{2}.
\]
\end{definition}

However, the most common definition of the HL-estimator is based on the Wilcoxon rank-sum statistic \citep{hollander_nonparametric_2014}, because it can be computed and interpreted as the median of the Cartesian product of $\mathcal{S}_{1} \times \mathcal{S}_{0}$.  

\begin{definition}\label{defn: Mann-Whitney}
The Wilcoxon rank-sum statistic in the Mann--Whitney form is
\[
\mathcal W=\sum_{k=1}^{N_0}\sum_{j=1}^{N_1}\varphi(Y_m^0,Y_j^1),
\]
where
\[
\varphi(Y_m^0,Y_j^1)=
\begin{cases}
1, & \text{if } Y_m^0<Y_j^1,\\
0, & \text{otherwise}.
\end{cases}
\]
Then the Hodges--Lehmann estimator associated with the Wilcoxon rank-sum statistic for the shift parameter $\Delta$ is defined as
\[
\hat\Delta_{\mathrm{HL}}
=
\operatorname{med}\{Y_j^1-Y_m^0 : m\in\mathcal S_0,\ j\in\mathcal S_1\}.
\]
\end{definition}

The following lemma provides bounds for the distribution of 
\( \hat\Delta_{\mathrm{HL}}
\), which are necessary for our main discussion. 

\begin{lemma}\label{lem:lemma_bounds}
Let $\mathcal W=t(Y^0,Y^1)$ be the test statistic in Definition~\ref{defn: Mann-Whitney}. For any
$a\in\mathbb R$ and symmetry point $\xi$, $\hat\Delta_{\mathrm{HL}}$ satisfies
\[
Pr\bigl(t(Y^0,Y^1-a)<\xi\bigr)
\le
Pr\bigl(\hat\Delta_{\mathrm{HL}}<a\bigr)
\le
Pr\bigl(t(Y^0,Y^1-a)\le \xi\bigr).
\]
\end{lemma}

\begin{proof}
See Appendix~\ref{proof:proof_lemma_4}.
\end{proof}

So far, we have introduced the Hodges--Lehmann estimator in a general two-sample setting without making any additional assumptions about $F_{1}(\cdot)$ and $F_{0}(\cdot)$. The previous definitions do not, by themselves, identify which parameter the HL estimator targets in a causal sense.
Three related but distinct objects are relevant in this paper: (i) the location shift $\Delta$, (ii) the CATE, $\tau(x)$, and (iii)  the median of individual treatment effects, $\tau_{\operatorname{med}}$.  The HL estimator is most naturally interpreted as an estimator of the shift parameter $\Delta$, or, equivalently, as the median of all pairwise treated-minus-control differences.

We now record simple sufficient conditions under which the location shift parameter coincides with $\tau_{\operatorname{med}}(x;\Pi)$ and $\tau(l;\Pi)$. 
\begin{lemma}\label{lemma:Delta_equals_tau_leaf}
For a fixed partition $\Pi$ and a target point $x$, define the leafwise treatment effect by
\[
\tau(l;\Pi):=\mathbb E\bigl[Y_i(1)-Y_i(0)\mid X_i\in l(x;\Pi)\bigr].
\]
Suppose that within this leaf the treatment effect is constant, so that
\[
Y_i(1)=Y_i(0)+\tau(l;\Pi)
\qquad \text{almost surely for all } X_i\in l(x;\Pi).
\]
Then
\[
\Delta(x;\Pi)=\tau(l;\Pi)=\tau_{\operatorname{med}}(x;\Pi).
\]
\end{lemma}

\begin{proof}
See Appendix~\ref{proof:Delta_equals_tau_leaf}.
\end{proof}

Note that \(\tau(l;\Pi)\) and \(\tau_{\operatorname{med}}(x;\Pi)\) are leaf-level
targets induced by the tree approximation, but this does not by itself imply that the HL shift parameter \(\Delta(x;\Pi)\) equals either target. Lemma~\ref{lemma:Delta_equals_tau_leaf} establishes the equality under a constant within-leaf treatment effect. Proposition~\ref{prop:Delta_equals_tau_equals_tauMed}
gives an alternative sufficient condition based on a leafwise location-shift model and symmetry of the individual treatment-effect distribution.

\begin{proposition}\label{prop:Delta_equals_tau_equals_tauMed}
For a fixed partition $\Pi$ and a target point $x$, let $l=l(x;\Pi)$ be the leaf containing
$x$. Assume that:
\begin{itemize}
    \item[$i)$] the leafwise potential-outcome distributions satisfy a location-shift model,
    \[
    F_{1,l}(z)=F_{0,l}(z-\Delta(x;\Pi))
    \qquad \text{for all } z\in\mathbb R;
    \]
    \item[$ii)$] the conditional distribution of the individual treatment effect
    \[
    Y_i(1)-Y_i(0)\mid X_i\in l
    \]
    is symmetric.
\end{itemize}
Then
\[
\Delta(x;\Pi)=\tau(l;\Pi)=\tau_{\operatorname{med}}(x;\Pi). 
\]
\end{proposition}
\begin{proof}
    See Appendix \ref{proof:Delta_equals_tau_equals_tauMed}
\end{proof}

Overall, we provide sufficient conditions under which \(\Delta(x;\Pi)\) coincides with $\tau(l;\Pi)$ and $\tau_{\operatorname{med}}(x;\Pi)$. These are leafwise statements and their connection to the pointwise CATE or pointwise CMTE relies on the usual tree approximation. As the partition refines around \(x\), the leafwise targets approximate their pointwise counterparts under suitable smoothness conditions. When these conditions are violated, the HL estimator remains an estimator of \(\Delta(x;\Pi)\). Its use in the splitting criteria should then be understood as relying on a shift-based proxy rather than on an exact identification result.
A detailed analysis of the discrepancies that may arise under weaker conditions lies beyond the scope of this paper and is left for future work.

\subsection{Details for the motivating example}
\label{app:motivation}

Let us consider the setup in Section \ref{sec:motivation} with $N^{tr}_{1,l} = N^{tr}_{0,l} = 10$ and hence $N^{tr}_l = 20$. We further fix an estimation-sample leaf of the same size, $N^{est}_l = 20$, and a treatment share $p = 1/2$. Under \eqref{eq:sharp-null} we have $\tau = 2$, and the baseline satisfies $\mathbb{E}[\varepsilon_i] = 0$ and $\Var(\varepsilon) = 3$.
Index the treated observations in $l$ by $j = 1,\dots,N^{tr}_{1,l}$, so that $Y_j = 3 + \tau + \varepsilon_j$ under \eqref{eq:sharp-null}. In the following, we compare the mean-based splitting criterion in \eqref{eq:MSE_CT_H} under an uncontaminated and a contaminated leaf. The uncontaminated leaf contains these treated outcomes $Y_j$. The contaminated leaf, marked with a dagger, is identical except that
$\varepsilon_1$ is replaced by the constant $c=10$:
\begin{equation}\label{eq:app_config}
  Y^{\dagger}_1 = 3 + \tau + c,
  \qquad
  Y^{\dagger}_j = Y_j \quad (j = 2,\dots,N^{tr}_{1,l}),
\end{equation}
and the outcomes of the control group of size $N^{tr}_{0,l}$ are left unchanged in both the uncontaminated and contaminated leaves. 
Write $uc \coloneqq \varepsilon_1$ for the uncontaminated draw that $c$ displaces in the contaminated case. All expectations are taken over $uc$ and the remaining draws, and use no property of the baseline beyond $\mathbb{E}[\varepsilon_i]=0$ and $\mathbb{E}[\varepsilon_i^2]=\Var(\varepsilon)$.

Let us re-write the mean-based splitting criterion in \eqref{eq:MSE_CT_H} as
\begin{align} \label{MSE_CT_H_rewrite}
\hat Q^{\,\text{MSE, H}}_\tau(\mathcal{S}^{tr},\Pi) = \sum_{l\in\Pi} (N^{tr}_l/N^{tr})\,
\hat q(l),  
\end{align}
with
\begin{equation}\label{eq:app_bracket}
  \hat q(l) \coloneqq \hat\tau(l;\mathcal{S}^{tr})^2
  - \left(\frac{1}{N^{tr}_l} + \frac{1}{N^{est}_l}\right)
    \left(\frac{\hat S^2_1(l;\mathcal{S}^{tr})}{p}
        + \frac{\hat S^2_0(l;\mathcal{S}^{tr})}{1-p}\right)
\end{equation}
The weight $N^{tr}_l/N^{tr}$ in \eqref{MSE_CT_H_rewrite} is common to both the uncontaminated and contaminated cases compared below and therefore cancels throughout. We are interested in the expected difference between the contaminated and uncontaminated expressions in \eqref{eq:app_bracket}, by focusing on the difference in the expected change of the reward, $\mathbb{E}\left[\Delta R\right]$, and the change of the variance penalty term, $\mathbb{E}\left[\Delta\widehat{\text{pen}}\right]$, using 
\begin{align}
\begin{split}
  \Delta\hat q(l)
  &\coloneqq \hat q^{\dagger}(l) - \hat q(l)
    \notag\\
  &= \Bigl(\hat\tau^{\dagger}(l;\mathcal{S}^{tr})^2 - \widehat{\text{pen}}^{\dagger}\Bigr)
   - \Bigl(\hat\tau(l;\mathcal{S}^{tr})^2 - \widehat{\text{pen}}\Bigr)
    \notag\\
  &= \Bigl(\hat\tau^{\dagger}(l;\mathcal{S}^{tr})^2
         - \hat\tau(l;\mathcal{S}^{tr})^2\Bigr)
   - \Bigl(\widehat{\text{pen}}^{\dagger} - \widehat{\text{pen}}\Bigr)
    \notag\\
  &= \Delta R - \Delta\widehat{\text{pen}}.
\end{split}
\end{align}
Taking expectations gives
\begin{equation}\label{eq:app_step5a_expected}
    \mathbb{E}\left[ \Delta\hat q(l) \right] =  \mathbb{E}\left[ \Delta R \right] -  \mathbb{E}\left[ \Delta\widehat{\text{pen}} \right].
\end{equation}

Let us first inspect the reward term and re-write it solely in terms of uncontaminated quantities by
\begin{align}\label{eq:app_step2a}
  \mathbb{E}[\Delta R]
  &\coloneqq \mathbb{E}\bigl[ \hat\tau^{\dagger}(l;\mathcal{S}^{tr})^2 - \hat\tau(l;\mathcal{S}^{tr})^2\bigr]
    \notag\\
  &= \mathbb{E}\bigl[ \bigl(\hat\tau(l;\mathcal{S}^{tr}) + \Delta\hat\tau\bigr)^2
     - \hat\tau(l;\mathcal{S}^{tr})^2\bigr]
    \notag\\
  &=  \mathbb{E}\bigl[ \hat\tau(l;\mathcal{S}^{tr})^2
     + 2\,\hat\tau(l;\mathcal{S}^{tr})\,\Delta\hat\tau
     + (\Delta\hat\tau)^2
     - \hat\tau(l;\mathcal{S}^{tr})^2\bigr]
    \notag\\
  &= 2\,\mathbb{E}\bigl[\hat\tau(l;\mathcal{S}^{tr})\,\Delta\hat\tau\bigr] + \mathbb{E}\bigl[(\Delta\hat\tau)^2 \bigr],
\end{align}
For both expectations in \eqref{eq:app_step2a}, we use 
\begin{equation}\label{eq:app_step1}
  \Delta\hat\tau \coloneqq \hat\tau^{\dagger}(l;\mathcal{S}^{tr}) - \hat\tau(l;\mathcal{S}^{tr})  = \bigl(\bar Y^{\dagger}_{1,l} - \bar Y_{0,l}\bigr)
     - \bigl(\bar Y_{1,l} - \bar Y_{0,l}\bigr)
  = \frac{Y^{\dagger}_1 - Y_1}{N^{tr}_{1,l}}
  = \frac{c-uc}{N^{tr}_{1,l}}, 
\end{equation}
where $\bar Y_{1,l}$ and $\bar Y_{0,l}$ are averages of the treated and control outcomes in leaf $l$.
For the first term in \eqref{eq:app_step2a}, 
\begin{align}\label{eq:app_step2b}
  2\mathbb{E}\bigl[\hat\tau(l;\mathcal{S}^{tr})\,\Delta\hat\tau\bigr]
  &= 2\mathbb{E}\!\left[\bigl(\tau + \bar\varepsilon_{1,l} - \bar\varepsilon_{0,l}\bigr)
     \cdot \frac{c-uc}{N^{tr}_{1,l}}\right]
    \notag\\
  &= \frac{2}{N^{tr}_{1,l}}
     \Bigl\{\tau\,\mathbb{E}[c-uc]
     + \mathbb{E}\bigl[\bar\varepsilon_{1,l}(c-uc)\bigr]
     - \mathbb{E}\bigl[\bar\varepsilon_{0,l}(c-uc)\bigr]\Bigr\} \notag\\
  &= \frac{2}{N^{tr}_{1,l}}\left(\tau c - \frac{\Var(\varepsilon)}{N^{tr}_{1,l}}\right) = \frac{2}{10}\bigl(20 - 0.3\bigr) = 3.94.
\end{align}
For the second term in \eqref{eq:app_step2a}, we have
\begin{align}\label{eq:app_step2d}
  \mathbb{E}\bigl[(\Delta\hat\tau)^2\bigr]
  &= \frac{1}{(N^{tr}_{1,l})^2}\,\mathbb{E}\bigl[(c-uc)^2\bigr]
    \notag\\
  &= \frac{1}{(N^{tr}_{1,l})^2}
     \Bigl(c^2 - 2c\,\underbrace{\mathbb{E}[uc]}_{=\,0} + \mathbb{E}[uc^2]\Bigr)
    \notag\\
  &= \frac{c^2 + \Var(\varepsilon)}{(N^{tr}_{1,l})^2}
   = \frac{100 + 3}{100}
   = 1.03.
\end{align}
Inserting \eqref{eq:app_step2b} and \eqref{eq:app_step2d} into \eqref{eq:app_step2a} gives 
\begin{equation}\label{eq:app_step2}
  \mathbb{E}\bigl[\Delta R\bigr] = 2\mathbb{E}\bigl[\hat\tau(l;\mathcal{S}^{tr})\,\Delta\hat\tau\bigr] + \mathbb{E}\bigl[(\Delta\hat\tau)^2\bigr]
  = 3.94 + 1.03
  = 4.97.
\end{equation}

Let us now re-arrange the expected change in the penalty term, $\mathbb{E}\left[\Delta\widehat{\text{pen}}\right]$, in \eqref{eq:app_step5a_expected} by
\begin{align}\label{eq:app_step4a}
 \mathbb{E}\left[ \Delta\widehat{\text{pen}} \right] 
  &\coloneqq \mathbb{E}\left[\widehat{\text{pen}}^{\dagger}  - \widehat{\text{pen}}\right] 
    \notag\\
  &= \mathbb{E}\left[ \left(\frac{1}{N^{tr}_l} + \frac{1}{N^{est}_l}\right)
     \left(\frac{\hat S^{2\dagger}_1 - \hat S^{2}_1}{p}
         + \frac{\hat S^{2\dagger}_0 - \hat S^{2}_0}{1-p}\right) \right] 
    \notag\\
  &= \mathbb{E}\left[  \left(\frac{1}{N^{tr}_l} + \frac{1}{N^{est}_l}\right)
     \left(\frac{\Delta\hat S^2_1}{p}
         + \frac{\overbrace{\Delta\hat S^2_0}^{=\,0}}{1-p}\right) \right] 
    \notag\\
  &= \left(\frac{1}{N^{tr}_l} + \frac{1}{N^{est}_l}\right)
     \frac{\mathbb{E}\left[\Delta\hat S^2_1 \right]}{p} = \left(\frac{1}{20} + \frac{1}{20}\right)\frac{9.7}{0.5}
  = 1.94, 
\end{align}
where we use that the contamination affects the treated group only, so $\Delta\hat S^2_0 \coloneqq \hat S^{2\dagger}_0(l;\mathcal{S}^{tr})
- \hat S^2_0(l;\mathcal{S}^{tr}) = 0$. Further, we use in \eqref{eq:app_step4a} that 
\begin{align}\label{eq:app_step3b}
  \mathbb{E}\left[\Delta\hat S^2_1\right]
  &= \mathbb{E}\bigl[ \hat S^{2\dagger}_1(l;\mathcal{S}^{tr}) - \hat S^2_1(l;\mathcal{S}^{tr})\bigr] \notag\\
  &= \mathbb{E}\left[ \frac{(c-{\bar \varepsilon_{1, l, (-1)}})^2 - (uc-{\bar \varepsilon_{1, l, (-1)}})^2}{N^{tr}_{1,l}} \right] \notag\\
  &= \frac{\bigl(c^2 + \mathbb{E}\left[{\bar \varepsilon_{1, l, (-1)}}^2\right]\bigr)
          - \bigl(\Var(\varepsilon) + \mathbb{E}\left[{\bar \varepsilon_{1, l, (-1)}}^2\right]\bigr)}{N^{tr}_{1,l}}
    \notag\\
  &= \frac{c^2 - \Var(\varepsilon)}{N^{tr}_{1,l}}
   = \frac{100 - 3}{10}
   = 9.7,
\end{align}
where ${\bar \varepsilon_{1, l, (-1)}}$ is the mean of the $N^{tr}_{1,l}-1$ treated errors $\varepsilon_2,\dots,\varepsilon_{N^{tr}_{1,l}}$ with $\mathbb{E}\left[{\bar \varepsilon_{1, l, (-1)}}\right]=0$.

Combining the results for $\mathbb{E}\bigl[\Delta R\bigr]$ in \eqref{eq:app_step2} and $\mathbb{E}\left[ \Delta\widehat{\text{pen}} \right]$ in \eqref{eq:app_step5a_expected} gives
\begin{equation*}
    \mathbb{E}\bigl[\Delta\hat q(l)\bigr]
  = \mathbb{E}\bigl[\Delta R\bigr] - \mathbb{E}\bigl[\Delta\widehat{\text{pen}}\bigr]
  = 4.97 - 1.94
  = 3.03.
\end{equation*}
Conditional on the occurrence of the extreme draw, the reward is raised by roughly $5$ and the penalty by roughly $2$, so the variance correction of \eqref{eq:MSE_CT_H} absorbs less than half of the distortion, and the leaf still contributes roughly $3$ more to the criterion than it should when correctly adapting to the noise contamination.

\subsection{Median-based Honest Splitting Criterion}
\label{appendix:Honest_Splitting_MSD}

The derivation is based on  Appendix~\ref{appendix:Honest_Splitting_CT}, replacing
the difference-in-means by the HL estimator. The samples
$\mathcal{S}^{tr},\mathcal{S}^{est},\mathcal{S}^{te}$, the treatment
probability $p$, and the population quantities $\tau(x;\Pi),\,p_l$ are
unchanged. Throughout, the $\mathbb{E}$-operator without subscript denotes joint expectation
over $(\mathcal{S}^{est},\mathcal{S}^{te})$ and we add a subscript only where the
distinction matters.

Recall that for a fixed tree $\Pi$ and estimation sample $\mathcal{S}^{est}$, the leaf-level
treatment effect is estimated by the HL statistic from \eqref{eq:HL-leafwise}:
\begin{equation}\label{eq:HL_est}
\hat\tau_{\mathrm{HL},l}\equiv\hat\tau_{\mathrm{HL}}(X_i;\mathcal{S}^{est},\Pi)
    = \operatorname{med}\!\bigl\{Y_j^{1}-Y_m^{0}: j\in\mathcal{S}^{est}_1,\;
      m\in\mathcal{S}^{est}_0,\; X_j,X_m\in l(x;\Pi)\bigr\}, 
\end{equation}
where we abbreviate $\hat\tau_{\mathrm{HL},l}\equiv\hat\tau_{\mathrm{HL}}(X_i;\mathcal{S}^{est},\Pi)$
for $i\in\mathcal{S}^{te}$ with $X_i\in l$. By
Theorems~\ref{thm:theorem_consistency} and~\ref{thm:theorem_median_unbiased},
$\hat\tau_{\mathrm{HL},l}$ is consistent for $\Delta(x;\Pi)$, median-unbiased,
and asymptotically normal (Appendix~\ref{proof:theorem_consistency}). Furthermore, when the sufficient conditions of Lemma~\ref{lemma:Delta_equals_tau_leaf} or Proposition~\ref{prop:Delta_equals_tau_equals_tauMed} hold, the shift parameter coincides with the leafwise treatment effect. Exact mean-unbiasedness fails in finite samples. However, under the location-shift model of Assumption~\ref{ass:HL_shift}(i) each within-leaf pairwise treated$-$control difference equals $\tau(l;\Pi)$ plus a noise term that is symmetric about zero, since the two arms are drawn i.i.d. from a common within-leaf baseline. This is the same symmetry that yields the median-unbiasedness of Theorem~\ref{thm:theorem_median_unbiased}. Under the median convention of Definition~\ref{median_distribution} it makes $\hat\tau_{\mathrm{HL},l}$ mean-unbiased when the two arms are equal in size, and leaves only a higher-order bias otherwise, so that
\begin{equation}\label{eq:HL_approx_unbiased}
    \mathbb{E}_{\mathcal{S}^{est}}[\hat\tau_{\mathrm{HL},l}]
    =\tau(l;\Pi)+O(1/N_l^{est})\qquad\text{as }N_l^{est}\to\infty.
\end{equation}

The MSE is defined exactly as in \eqref{eq:mse_tau} with
$\hat\tau_{\mathrm{HL}, l}$ replacing $\hat\tau_l$:
\begin{align}\label{eq:mse_tau_med_appendix}
\begin{split}
    \mathrm{MSE}_{\tau,\mathrm{med}}(\mathcal{S}^{te},\mathcal{S}^{est},\Pi)
    &=\frac{1}{N^{te}}\sum_{i\in\mathcal{S}^{te}}
      \bigl[(\tau_i-\hat\tau_{\mathrm{HL},l})^2-\tau_i^2\bigr]\\
    &=\sum_{l\in\Pi}\frac{N_l^{te}}{N^{te}}
      \!\left[\hat\tau_{\mathrm{HL},l}^{\,2}
      -2\,\hat\tau_{\mathrm{HL},l}\cdot\frac{1}{N_l^{te}}
        \sum_{\substack{i\in\mathcal{S}^{te}\\X_i\in l}}\!\tau_i\right].
\end{split}
\end{align}
The algebra mirrors \eqref{eq:expand}--\eqref{eq:grouped}. Apply $\mathbb{E}_{(\mathcal{S}^{te},\mathcal{S}^{est})}[\cdot]$ to
\eqref{eq:mse_tau_med_appendix}, and by linearity, the expectation acts separately on the two
summands inside each leaf. Within each summand, identify which factors depend on
either $\mathcal{S}^{te}$ or $\mathcal{S}^{est}$. The HL estimator
$\hat\tau_{\mathrm{HL},l}$ depends only on $\mathcal{S}^{est}$ while the leaf size
$N_l^{te}$ and the leaf-restricted sum
$\sum_{i\in\mathcal{S}^{te},X_i\in l}\tau_i$ depend only on $\mathcal{S}^{te}$.
By the independence of $\mathcal{S}^{te}$ and $\mathcal{S}^{est}$, the joint
expectation of each summand factorizes into an $\mathcal{S}^{te}$-part and an
$\mathcal{S}^{est}$-part:
\begin{align}\label{eq:factorise_med}
\begin{split}
    \mathbb{E}\!\left[\frac{N_l^{te}}{N^{te}}\,\hat\tau_{\mathrm{HL},l}^{\,2}\right]
    &\;=\;\underbrace{\mathbb{E}_{\mathcal{S}^{te}}\!\left[\frac{N_l^{te}}{N^{te}}\right]}_{\textstyle \to\,p_l}
      \;\cdot\;\mathbb{E}_{\mathcal{S}^{est}}\!\bigl[\hat\tau_{\mathrm{HL},l}^{\,2}\bigr],\\[2pt]
    \mathbb{E}\!\left[\frac{N_l^{te}}{N^{te}}\,\hat\tau_{\mathrm{HL},l}\cdot
      \frac{1}{N_l^{te}}\!\!\sum_{\substack{i\in\mathcal{S}^{te}\\X_i\in l}}\!\!\tau_i\right]
    &\;=\;\mathbb{E}_{\mathcal{S}^{est}}[\hat\tau_{\mathrm{HL},l}]
      \;\cdot\;\underbrace{\mathbb{E}_{\mathcal{S}^{te}}\!\left[\frac{1}{N^{te}}
        \!\!\sum_{\substack{i\in\mathcal{S}^{te}\\X_i\in l}}\!\!\tau_i\right]}_{\textstyle =\,p_l\,\tau(l;\Pi)}.
\end{split}
\end{align}
The $\mathcal{S}^{te}$-part in the first line is the empirical leaf
probability with limit $p_l$. The $\mathcal{S}^{te}$-part in the second line
uses the simplification
$\tfrac{N_l^{te}}{N^{te}}\cdot\tfrac{1}{N_l^{te}}=\tfrac{1}{N^{te}}$ and is the
empirical leaf-restricted sample mean of the true $\tau_i$ on
$\mathcal{S}^{te}$, which by random sampling equals $p_l\,\tau(l;\Pi)$. The
$\mathcal{S}^{est}$-part in each line carries a moment of $\hat\tau_{\mathrm{HL},l}$.

Summing \eqref{eq:factorise_med} over $l\in\Pi$ and pulling out the common
$p_l$ gives the EMSE as the expectation of \eqref{eq:mse_tau_med_appendix}:
\begin{equation}\label{eq:emse_med}
    \mathrm{EMSE}_{\tau,\mathrm{med}}(\Pi)
    \;=\;\sum_{l\in\Pi}p_l\!\left[\,\mathbb{E}[\hat\tau_{\mathrm{HL},l}^{\,2}]
      \,-\,2\,\mathbb{E}[\hat\tau_{\mathrm{HL},l}]\,\tau(l;\Pi)\,\right].
\end{equation}
This is the key population criterion. To turn it into a splitting rule computed on \(\mathcal S^{tr}\), we need training-sample approximations for two quantities:
\[\mathbb{E}[\hat\tau_{\mathrm{HL},l}]
\qquad\text{and}\qquad
\mathbb{E}[\hat\tau_{\mathrm{HL},l}^{\,2}].
\]
Both expectations are $\mathcal{S}^{est}$-only:
$\hat\tau_{\mathrm{HL},l}$ and $\hat\tau_{\mathrm{HL},l}^{\,2}$ are
$\mathcal{S}^{est}$-measurable, and the $\mathcal{S}^{te}$-side has been
absorbed into the $p_l$ and $p_l\,\tau(l;\Pi)$ in \eqref{eq:factorise_med}. We
therefore drop the $\mathcal{S}^{est}$ subscript on $\mathbb{E}[\cdot]$ from
here on.
Maximizing $-\mathrm{EMSE}_{\tau,\mathrm{med}}(\Pi)$ on $\mathcal{S}^{tr}$
requires training-sample plug-ins for the two estimation-sample moments of $\hat\tau_{\mathrm{HL},l}$ that appear in \eqref{eq:emse_med}: the first moment $\mathbb{E}[\hat\tau_{\mathrm{HL},l}]$ in the cross-term, and the second moment $\mathbb{E}[\hat\tau_{\mathrm{HL},l}^{\,2}]$ in the leading term. These two moments behave very differently.

The cross-term equals $\mathbb{E}[\hat\tau_{\mathrm{HL},l}]\cdot\tau(l;\Pi)$,
a product of two population objects, each consistently estimable on $\mathcal{S}^{tr}$. The natural plug-in for the first factor $\mathbb{E}[\hat\tau_{\mathrm{HL},l}]$ is the training-sample HL estimator
$\hat\tau_{\mathrm{HL}}(l;\mathcal{S}^{tr})$, which by
\eqref{eq:HL_approx_unbiased} (applied with $\mathcal{S}^{tr}$ in place of
$\mathcal{S}^{est}$) is consistent for $\tau(l;\Pi)$ and matches the structural
form of the population statistic. For the second factor $\tau(l;\Pi)$ two
canonical estimators are available on $\mathcal{S}^{tr}$:
$\hat\tau_{\mathrm{HL}}(l;\mathcal{S}^{tr})$, which is consistent by
Theorem~\ref{thm:theorem_consistency}; or the difference-in-means
$\hat\tau(l;\mathcal{S}^{tr})$, which is exactly unbiased under unconfoundedness. Both
are valid plug-ins, and the resulting criterion is asymptotically the same to
leading order.

The second-moment $\mathbb{E}[\hat\tau_{\mathrm{HL},l}^{\,2}]$ is qualitatively harder because no estimator of $\hat\tau_{\mathrm{HL},l}^{\,2}$
on $\mathcal{S}^{tr}$ is mean-unbiased for $\mathbb{E}[\hat\tau_{\mathrm{HL},l}^{\,2}]$
without further structure. We discuss two strategies in the following: explicit density-based variance estimation and a cross-product plug-in strategy. Note that reformulations concerning the EMSE are carried out under the leafwise location-shift model of Assumption~\ref{ass:HL_shift}(i), so that $\Delta(x;\Pi)=\tau(l;\Pi)$ and $\hat\tau_{\mathrm{HL},l}$ targets the leaf CATE.

\subsubsection{Explicit density-based penalty}

The variance--mean identity, combined with
\eqref{eq:HL_approx_unbiased},
\begin{align}\label{eq:second_moment_id}
\begin{split}
    \mathbb{E}\!\left[\hat\tau_{\mathrm{HL},l}^{\,2}\right]
    &=\bigl(\mathbb{E}[\hat\tau_{\mathrm{HL},l}]\bigr)^2
      +\Var(\hat\tau_{\mathrm{HL},l})\\
    &=\tau(l;\Pi)^2+\Var(\hat\tau_{\mathrm{HL},l})+O\left(1/N_l^{est}\right).
\end{split}
\end{align}
Thus, the second moment equals the squared leafwise CATE plus the sampling variance of the HL estimator, up to a higher-order term of order $O(1/N_l^{est})$ that collects the bias contributions of \eqref{eq:HL_approx_unbiased}. 
Substituting \eqref{eq:second_moment_id} into \eqref{eq:emse_med} gives
\begin{align}\label{eq:emse_med_var}
    \mathrm{EMSE}_{\tau,\mathrm{med}}(\Pi)
    &=\sum_{l\in\Pi}p_l\!\left[\Var(\hat\tau_{\mathrm{HL},l})
       -\tau(l;\Pi)^2\right]+O\left(1/N_l^{est}\right).
\end{align}
This is the structural analogue of the mean-based \eqref{eq:emse_final}, now with $\Var(\hat\tau_{\mathrm{HL},l})$ instead of $\Var(\hat\tau_{l})$.
The latter has a closed form from the classical theory of the Hodges--Lehmann
estimator and the Wilcoxon rank-sum statistic
\citep{hodges_estimates_1963, lehmann_nonparametrics_1975}. Under the leafwise
location-shift model of Assumption~\ref{ass:HL_shift}, the treated and control
outcomes share a common within-leaf density up to the shift, so
$\int f_{1,l}^2 = \int f_{0,l}^2 =: \int f_l^2$, and the consistency argument of
Appendix~\ref{proof:theorem_consistency} yields, under random treatment assignment,
\begin{align}\label{eq:var_HL}
    \Var(\hat\tau_{\mathrm{HL},l})
    &\approx\frac{1}{N_l^{est}}
      \underbrace{\frac{1}{12\,p\,(1-p)\,\bigl(\int f_l^2(u)\,du\bigr)^2}}_{\equiv\,V_l},
\end{align}
where the control group proportion $\lambda$ of Assumption~\ref{ass:HL_shift}(iii) equals $\lambda=1-p$, so $\lambda(1-\lambda)=p(1-p)$. The single density functional $\int f_l^2$, rather than the pointwise densities $f_d(m_d(l))$ at the two arm medians, is the signature of the rank-based estimator: the Hodges--Lehmann estimator is the median of pairwise differences, not the difference of medians, and its efficiency is governed by $\int f_l^2$,
the same functional that determines the asymptotic relative efficiency of the
Wilcoxon test \citep{hodges_estimates_1963}.

To build the sample criterion on $\mathcal{S}^{tr}$ we replace
$p_l\!\to\!N_l^{tr}/N^{tr}$, estimate $V_l$ by $\widehat V_l$ from a
nonparametric estimate of $\int f_l^2$, and use the plug-in
$\hat\tau_{\mathrm{HL}}(l;\mathcal{S}^{tr})^2$ for $\tau(l;\Pi)^2$. The plug-in
is upward-biased,
\begin{equation}\label{eq:plugin_bias}
    \mathbb{E}\!\left[\hat\tau_{\mathrm{HL}}(l;\mathcal{S}^{tr})^2\right]
    =\tau(l;\Pi)^2+\Var_{\mathcal{S}^{tr}}(\hat\tau_{\mathrm{HL},l}) + O\left(1/N_l^{tr}\right),
\end{equation}
and subtracting both the training-sample bias and the estimation-sample
variance from \eqref{eq:var_HL} yields the median-based honest criterion
\begin{align}\label{eq:Q_MSD_H}
    \hat Q^{\,\mathrm{MSD,H}}_\tau(\Pi;\mathcal{S}^{tr})
    &=\sum_{l\in\Pi}\frac{N_l^{tr}}{N^{tr}}
      \!\left[\hat\tau_{\mathrm{HL}}(l;\mathcal{S}^{tr})^2
      -\!\left(\frac{1}{N_l^{tr}}+\frac{1}{N_l^{est}}\right)\!\widehat V_l\right],
\end{align}
to be maximized. This criterion is the closest analogue of the honest mean-based splitting rule. The first term rewards HL-based treatment-effect heterogeneity. The second term penalizes leaves in which the HL estimate would be noisy, accounting for both
training-sample plug-in bias and estimation-sample variance.
The difficulty is computational. Estimating \(\widehat V_l\) requires estimating \(\int f_l^2(u)\,du\) inside every candidate leaf and at every candidate split. In small leaves this density estimate can be noisy, and in a tree-growing loop it is expensive to recompute repeatedly. For this reason, the implemented rule
uses a density-free approximation in the next subsection. 

\subsubsection{Cross-product plug-in}

We bypass the density estimation by substituting plug-ins directly into
\eqref{eq:emse_med} rather than first invoking
\eqref{eq:second_moment_id}. The three substitutions are
\begin{align}\label{eq:plugins_code}
\begin{split}
    \mathbb{E}[\hat\tau_{\mathrm{HL},l}]&\;\longmapsto\;\hat\tau_{\mathrm{HL}}(l;\mathcal{S}^{tr}),\\
    \tau(l;\Pi)&\;\longmapsto\;\hat\tau(l;\mathcal{S}^{tr}),\\
    \mathbb{E}[\hat\tau_{\mathrm{HL},l}^{\,2}]
    &\;=\;\tau(l;\Pi)^2+O(1/N_l^{est})\;\longmapsto\;\hat\tau(l;\mathcal{S}^{tr})^2.
\end{split}
\end{align}
The first two lines are direct first-moment plug-ins:
\eqref{eq:HL_approx_unbiased} motivates
$\hat\tau_{\mathrm{HL}}(l;\mathcal{S}^{tr})$ as a consistent estimate of the
HL center, and the difference-in-means $\hat\tau(l;\mathcal{S}^{tr})$ is the canonical unbiased estimator of $\tau(l;\Pi)$ under unconfoundedness. 
The third line uses the leading-order reduction \eqref{eq:second_moment_id} to replace $\mathbb{E}[\hat\tau_{\mathrm{HL},l}^{\,2}]$ by $\tau(l;\Pi)^2$, discarding the estimation-sample variance $\Var(\hat\tau_{\mathrm{HL},l})=O(1/N_l^{est})$, and then substitutes the difference-in-means estimator for $\tau(l;\Pi)$.
Two same-order effects are thereby incurred: the discarded estimation-sample variance, of order $O(1/N_l^{est})$, and the upward bias $O(1/N_l^{tr})$ of the training-sample plug-in $\hat\tau(l;\mathcal{S}^{tr})^2$ for $\tau(l;\Pi)^2$.

Assembling \eqref{eq:plugins_code} in \eqref{eq:emse_med} and negating to
the maximization objective gives the criterion as implemented:
\begin{align}\label{eq:Q_MSD_code}
\begin{split}
    \hat Q^{\,\mathrm{MSD}}_\tau\left(\Pi;\mathcal{S}^{tr}\right)
    &=\sum_{l\in\Pi}\frac{N_l^{tr}}{N^{tr}}
       \!\left[2\,\hat\tau_{\mathrm{HL}}\left(l;\mathcal{S}^{tr}\right)\,
        \hat\tau\left(l;\mathcal{S}^{tr}\right)-\hat\tau\left(l;\mathcal{S}^{tr}\right)^2\right]\\
    &=\sum_{l\in\Pi}\frac{N_l^{tr}}{N^{tr}}
       \!\left[\hat\tau_{\mathrm{HL}}\left(l;\mathcal{S}^{tr}\right)^2\,-\Bigl(
       \hat\tau_{\mathrm{HL}}\left(l;\mathcal{S}^{tr}\right)-\hat\tau\left(l;\mathcal{S}^{tr}\right)\Bigr)^2\right]
\end{split}
\end{align}
The bracketed expression in the second line of \eqref{eq:Q_MSD_code} exposes the structure: the criterion rewards effect heterogeneity by the squared HL estimator and
penalizes the squared difference between the HL and difference-in-means estimators. The disagreement term plays the role of the density-based penalty $\widehat V_l$ in \eqref{eq:Q_MSD_H}, replacing nonparametric variance estimation with a self-consistency check between two leaf-level estimators of $\tau(l;\Pi)$.

In summary, \eqref{eq:Q_MSD_H} and \eqref{eq:Q_MSD_code} target the same population object in \eqref{eq:emse_med}:
\begin{itemize}
    \item Splitting criterion \eqref{eq:Q_MSD_H} carries an explicit density-based variance penalty that absorbs both the training-sample plug-in bias and the estimation-sample variance, matching the per-leaf EMSE up to the Hodges--Lehmann bias of \eqref{eq:HL_approx_unbiased}. It is operationally costly because $\widehat V_l$ requires kernel density estimation at the leaf level in every candidate split.
    \item Splitting criterion \eqref{eq:Q_MSD_code} replaces the estimation-sample second moment of $\hat\tau_{\mathrm{HL},l}$ by the squared training-sample difference-in-means, sidestepping density estimation entirely. Its bracket estimates $\tau(l;\Pi)^2$ up to a training-sample residual of order $O(1/N_l^{tr})$ and additionally discards the estimation-sample variance $\Var(\hat\tau_{\mathrm{HL},l})=O(1/N_l^{est})$.
\end{itemize}

\subsection{Baseline robust splitting rules: MAD and LMS}
\label{appendix:MAD_LMS}
The previous section showed that the MSD criterion is a
structural analogue of the honest mean-based MSE rule of \citet{athey_recursive_2016} and that it is derived from an EMSE decomposition with a heterogeneity reward and a variance penalty.
In this section, we will show that the MAD and LMS rules do not follow the same decomposition. They are nevertheless useful comparisons because they correspond to two different robustness ideas. LMS robustifies fit-based splitting, whereas MAD isolates the discrepancy between mean-based and Hodges--Lehmann leaf estimates. Both splitting rules align with the four-rule taxonomy of \citet{athey_recursive_2016}.

We first consider LMS. LMS is the median-of-squared-residuals analog of the fit-based rule in
\citet{athey_recursive_2016}. The adaptive fit-criterion \eqref{eq:MSE_fit_A} chooses splits by within-leaf fit of $Y_i$ given $(D_i,X_i)$, where the leaf-level outcome model in
\citet{ZeileisModelBasedRP} is an intercept plus treatment indicator (in two-arm settings, equivalent to the treatment-stratified leaf mean $\hat\mu_d$ of \citealp{athey_recursive_2016}). The LMS criterion replaces the within-leaf mean of squared residuals by the within-leaf median of squared residuals, and then averages these leafwise median losses using the leaf samples \(N_l^{tr}/N^{tr}\). Recall that the
criterion is
\begin{equation}
\widehat Q_\mu^{\mathrm{LMS}}
(\mathcal S^{tr},\Pi)
:=
\sum_{l\in\Pi}
\frac{N_l^{tr}}{N^{tr}}
\underset{i\in\mathcal S_l^{tr}}{\operatorname{med}}
\left[
\left(
Y_i-
\hat\mu(D_i,X_i;\mathcal S^{tr},\Pi)
\right)^2
\right].
\end{equation}
Candidate splits are chosen by minimizing $\widehat Q_\mu^{\mathrm{LMS}}$. Since the criterion uses medians rather than means, large residuals receive less influence than under the corresponding mean-squared fit criterion. This leads to structural consequences. 
A fit-based criterion rewards covariates that improve prediction of the outcome \(Y_i\), regardless of whether those covariates generate treatment-effect heterogeneity. Thus, a covariate that shifts baseline outcomes but leaves \(\tau(X_i)\) unchanged can still be selected by LMS. For this reason, LMS is a
robust outcome-fit criterion rather than a criterion directly targeted at CATE heterogeneity.
Also, LMS does not have a direct honest analogue in the sense of \citet{athey_recursive_2016}. The honest MSE and fit-based criteria rely on per-observation squared-error decompositions. Even in the leafwise form in \eqref{eq:LMS_split}, the median of squared residuals does not decompose this way: $\mathrm{med}_i[(Y_i -
\hat\mu_i)^2]$ is not a sum of per-observation risks, so there is no variance-penalty correction analogous to
\eqref{eq:honest_median}. The rule can still be used inside an
honest forest, because split selection and leaf estimation remain separated, but the criterion itself is not an honest EMSE criterion.
Furthermore, two implementation details in \eqref{eq:LMS_split} depart from \eqref{eq:MSE_fit_A}. 
First, the leaf model $\hat\mu(D_i,X_i;\mathcal{S}^{tr},\Pi)$ from \citet{ZeileisModelBasedRP} is an
intercept-plus-indicator regression rather than the unconstrained
treatment-stratified mean; in two-arm settings the two coincide. Second, the $-Y_i^2$
offset that appears in \eqref{eq:MSE_fit_A} is dropped in \eqref{eq:LMS_split}, because the median is not linear. Subtracting \(Y_i^2\) inside the median would not produce a partition-invariant offset and could therefore change the selected
split.

The MAD rule has a different role. Unlike LMS, it is directly related to the implemented MSD criterion. The implemented MSD rule in \eqref{eq:Q_MSD_code_main} can be written as
\begin{equation}
    \hat\tau_{\mathrm{HL}}(l;\mathcal{S}^{tr})^2
    -\bigl(\hat\tau_{\mathrm{HL}}(l;\mathcal{S}^{tr})-\hat\tau(l;\mathcal{S}^{tr})\bigr)^2. 
\end{equation}
The first term is a heterogeneity reward.  The second term penalizes the $L^2$
discrepancy between the two leaf-level estimators. MAD keeps only the corresponding \(L^1\) discrepancy. The motivation for this rule comes from the infeasible criterion \(\mathrm{MAD}_\tau\) in \eqref{eq:infeasible_MAD}, which is the natural \(L^1\) analogue of the infeasible MSE in \eqref{eq:MSE_treatment_effect}. Instead of squared deviations from the CATE, it considers absolute deviations from the CMTE \(\tau_{\operatorname{med}}(X_i)\).
However, the feasible \eqref{eq:MAD_split} is not an
estimator of \eqref{eq:infeasible_MAD}. Beyond the shared missing-counterfactual problem, $\tau_{\operatorname{med}}(X_i)$ has an identification problem \citep{Addanki2024LimitsOA} that the CATE $\tau(X_i)$ does not. In addition, the two objects inside the absolute value in \eqref{eq:MAD_split} are both leaf-level CATE estimators, not residuals of $\tau_i$ around a conditional median.

This feasible construction has two consequences.
First, MAD is also a heuristic rule rather than an honest EMSE criterion. The criterion in \eqref{eq:MAD_split} has no target-plus-penalty structure analogous to \eqref{eq:honest_median}. Second, under the location-shift conditions of Appendix~\ref{appendix:HL_estimator}, $\hat\tau(l;\mathcal{S}^{tr})$ and $\hat\tau_{\mathrm{HL}}(l;\mathcal{S}^{tr})$ both target the same leafwise CATE, so agreement between them is not itself evidence of treatment-effect heterogeneity. When the location-shift conditions fail, discrepancy between the two estimators may instead reflect skewness, heavy tails, or other distributional features within the leaf. In that case, MAD tends to partition according to estimator discrepancy or distributional shape rather than directly according to variation in \(\tau(X_i)\).
Note that none of the four rules in \citet[\S 3.1]{athey_recursive_2016} compares two
leaf-level CATE estimators against each other. 

\subsection{AIPW for ATE Estimation with Causal Forests}
\label{appendix:AIPW_ATE_CF}

We follow \cite{chernozhukov_doubledebiased_2018} and consider the standard AIPW estimator for the ATE,
\begin{align}
\tau_{\text{AIPW}} &= \frac{1}{N}\sum_{i=1}^N \left(\underbrace{\mu_1(X_i) - \mu_0(X_i)}_{= \tau(X_i)} + \frac{D_i(Y_i - \mu_1(X_i))}{p(X_i)} - \frac{(1-D_i)(Y_i - \mu_0(X_i))}{1-p(X_i)} \right) \\
&= \frac{1}{N}\sum_{i=1}^N \left( \tau(X_i) + \underbrace{\frac{D_i - p(X_i)}{p(X_i)(1-p(X_i))}}_{\text{debiasing weight}} \cdot  \underbrace{(Y_i - \mu_{D_i}(X_i))}_{\text{outcome residual}} \right),
\end{align}
which can be re-written in terms of CATE estimates that are corrected by debiasing-weighted outcome residuals. The initial CATE estimates of causal forests might be noisy due to slow convergence rates and simple averaging is inefficient. The debiasing weights provide a correction for estimation error based on actual treatment status and the propensity score while the outcome residual measures the prediction error in the outcome model. Combined with cross-fitting, $\hat{\tau}_{\text{AIPW}}$ remains asymptotically normal and efficient although the underlying $\hat\tau(x)$ estimates might be imprecise \citep{chernozhukov_doubledebiased_2018, athey_policy_2021}.
We estimate the cross-fitted AIPW estimator \citep{robins_estimation_1994, chernozhukov_doubledebiased_2018} by
\begin{equation}
\label{ATE_AIPW}
\begin{split}
\hat{\tau}_{\text{AIPW}} = \frac{1}{N} \sum_{i=1}^{N} &\Bigg(
\hat{\tau}^{(-k(i))}(X_i) + \frac{D_i - \hat{p}^{(-k(i))}(X_i)}{\hat{p}^{(-k(i))}(X_i) (1 - \hat{p}^{(-k(i))}(X_i))} 
\cdot ( Y_i - \hat{\mu}_{D_i}^{(-k(i))}(X_i) ) \Bigg),
\end{split}
\end{equation}
where the data is partitioned into disjoint folds indexed by $k$ and $$\hat{\tau}^{(-k)}(X_i) = \hat{\mu}_1^{(-k)}(X_i) - \hat{\mu}_0^{(-k)}(X_i),$$ by construction. We construct $\hat{\mu}^{(-k)}_{d}(x)$ by training on the complement of fold $k$, and define $k(i)$ as the fold assignment function that maps observation $i$ to its fold. 
That formulation in \eqref{ATE_AIPW} allows us to plug-in estimates from a causal forest into a doubly robust average treatment effect estimator which results in a semiparametrically efficient average treatment effect estimate and precise variance estimates \citep{chernozhukov_doubledebiased_2018, athey_policy_2021}. 

\subsection{Confidence Intervals via Bootstrap of Little Bags}
\label{appendix:BLB} 

Suppose $\hat{\tau}(x)$ is an honest forest-based estimator constructed from $B$ causal trees to target the CATE in Equation \eqref{eq:CATE}. Our goal is to quantify the uncertainty of $\hat{\tau}(x)$ and construct pointwise confidence intervals. We follow the Bootstrap of Little Bags (BLB) construction of \citet{athey_generalized_2019} and adapt it to the tree-level treatment-effect predictions produced by the original causal forest algorithm \citep{athey_recursive_2016, wager_estimation_2018}. The theoretical results of \citet{athey_generalized_2019} establish asymptotic validity for their corresponding generalized random forest estimator, whereas our implementation applies the same resampling principle to completed causal forest predictions based on the framework in \citet{athey_recursive_2016} and \cite{wager_estimation_2018}.

For the generalized random forest estimator of \citet{athey_generalized_2019}, the asymptotic variance in their equation~(15) satisfies 
\begin{align}\label{eq:Var_BLB}
\Var(\hat\tau(x)) = \kappa(x)^{-2} H_N(x),
\end{align}
where $H_N(x)$ is the variance of the score forest and $\kappa(x)$ is a curvature term that reflects local identification strength. 
Because our BLB is applied to completed treatment effect predictions, our $\widehat{H}_N^{\text{BLB}}(x)$ in \eqref{eq:Var_BLB_estim} already estimates the sampling variance of $\hat\tau(x)$ on the treatment effect scale, and the identity in \eqref{eq:Var_BLB} with $\kappa(x)^{-2}$ does not apply to it. We nonetheless retain the multiplier $\kappa(x)^{-2}$, with $\kappa(x)=\mathbb{E}[(D_i-p(X_i))^2\mid X_i=x]$, as a deliberately conservative scaling that widens the intervals and improves finite-sample coverage in the designs of Section~\ref{sec:simulations}. 
We estimate $\widehat{\kappa}(x)$ via an honest regression forest of $(D_i-\hat{p}(X_i))^2$ on covariates. Under a known constant propensity $p(X_i)$, $\kappa(x)$ is constant and the factor reduces to a fixed scalar. Table \ref{tab:blb_scaling_ci} compares estimated confidence intervals with and without this scaling factor and reports undercoverage when omitting $\widehat{\kappa}(x)^{-2}$ in \eqref{eq:Var_BLB_estim}.

Following \citet{sexton_standard_2009} and Section 4.1 of \citet{athey_generalized_2019}, we approximate the infeasible half-sampling variance estimator $H_N^{\text{HS}}(x)$ of equation (18) in \citet{athey_generalized_2019} via a BLB. Fix an integer $\ell \ge 2$ for the size of each bag (a group of trees) and define the number of groups $G = \lfloor B / \ell \rfloor$.
For each group $g = 1, \dots, G$, draw a half-sample
\begin{align*}
    \mathcal{S}_{\text{half},g} \subset  \{1,\dots,N\}, \qquad |\mathcal{S}_{\text{half},g}| = \lfloor N/2 \rfloor,
\end{align*}
uniformly without replacement. Conditional on $\mathcal{S}_{\text{half},g}$, draw $\ell$ subsamples $I_{g,b} \subseteq \mathcal{S}_{\text{half},g}$ of size
\begin{align*}
    |I_{g,b}| = \lfloor |\mathcal{S}_{\text{half},g}| \cdot {\rho} \rfloor,
    \qquad \rho \in (0,1),
\end{align*}
and grow $\ell$ honest causal trees using only observations in $I_{g,b}$. Let $\hat{\tau}_{g,b}(x)$ denote the CATE prediction of tree $b$ in group $g$ at covariate value $x$. Trees within a group share the same half-sample and are thus dependent, while trees from different groups are approximately independent. Define the group-level mean predictions as an average over trees $b \in \{1, ..., \ell \}$ in group $g$ by
\begin{align}
    \bar{\tau}_g(x) = \frac{1}{\ell} \sum_{b=1}^{\ell} \hat{\tau}_{g,b}(x),
\qquad g = 1,\dots,G,
\end{align}
and the overall mean over all groups $G$ by
\begin{align}
    \bar{\tau}(x) = \frac{1}{G} \sum_{g=1}^{G} \bar{\tau}_g(x).
\end{align}
We estimate the forest variance $H_N(x)$ using the ANOVA decomposition based on within- and between-group variation of tree predictions by
\begin{align*}
\widehat{V}_{\text{between}}(x)
&=
\frac{1}{G-1}
\sum_{g=1}^{G}
\left( \bar{\tau}_g(x) - \bar{\tau}(x) \right)^2, \\
\widehat{V}_{\text{within}}(x)
&=
\frac{1}{G}
\sum_{g=1}^{G}
\left[
\frac{1}{\ell}
\sum_{b=1}^{\ell}
\left( \hat{\tau}_{g,b}(x) - \bar{\tau}_g(x) \right)^2
\right].
\end{align*}

The between-group term centers the bag means on the estimated grand mean $\bar\tau(x)$, so we divide by $G-1$ rather than $G$ to remove the degree of freedom lost in centering, which renders $\widehat V_{\text{between}}$ unbiased for the between-group variance and $\widehat H_N^{\text{BLB,raw}}$ in \eqref{eq:H_BLB_raw} unbiased for the half-sampling variance.
Because $\widehat V_{\text{within}}$ divides each bag's squared deviations by $\ell$ rather than $\ell-1$, it understates the tree-to-tree variance within a bag by the factor $(\ell-1)/\ell$, and the coefficient $1/(\ell-1)$ corrects for this so that the quantity subtracted from $\widehat V_{\text{between}}$ equals the residual within-bag noise carried by a bag mean. This matches the decomposition of \citet{athey_generalized_2019}.
The classical BLB estimator reads
\begin{align} \label{eq:H_BLB_raw}
\widehat{H}_N^{\text{BLB,raw}}(x)
=
\widehat{V}_{\text{between}}(x)
-
\frac{1}{\ell - 1}\widehat{V}_{\text{within}}(x),
\end{align}
which may be negative in finite samples. To ensure non-negativity, we adopt a Bayesian ANOVA adjustment corresponding to an improper uniform prior on $[0,\infty)$ for $H_N(x)$ \citep{gelman2013bayesian, athey_generalized_2019} using
\begin{align}\label{eq:H_BLB_N}
\widehat H_N^{\text{BLB}} = \Big(\widehat V_{\text{between}} - \frac{1}{\ell-1}\,\widehat V_{\text{within}}\Big)_+,
\end{align}
where $\left(.\right)_+ = \max\{.,0\}$. As the number of trees $B \to \infty$, this estimator converges to the infeasible half-sampling variance. Combining the above components, we get
\begin{align}\label{eq:Var_BLB_estim}
    \widehat{\Var}_N(\hat\tau(x)) = \widehat{\kappa}(x)^{-2} \, \widehat{H}_N^{\text{BLB}}(x), \qquad \widehat{\mathrm{SE}}(\hat\tau(x)) = \sqrt{\widehat{\Var}_N(\hat\tau(x))}.
\end{align}

The asymptotic validity results of \citet{athey_generalized_2019} apply to the generalized random forest estimator under the conditions of Theorem 5 and the consistency result for the half-sampling estimator in Theorem 6 of \citet{athey_generalized_2019}. Since our implementation applies the BLB construction directly to the tree-level treatment effect predictions produced by the original causal forest algorithm \citep{athey_recursive_2016, wager_estimation_2018}, the resulting variance estimates should be viewed as an adaptation of the BLB procedure rather than as a direct application of the generalized random forest inference framework. Accordingly, the asymptotic guarantees of \cite{athey_generalized_2019} are not established for our implementation. The finite-sample performance of \eqref{eq:Var_BLB_estim} is evaluated empirically in Appendix~\ref{app:add_sim_res}. We further note that the empirical conclusion of Section~\ref{sec:empappl}, the absence of treatment effect heterogeneity distinguishable from a constant effect in Figures \ref{fig:plot_cate_progresa} and \ref{fig:plot_cate_actg175}, is conservative with respect to the scaling factor. The factor $\widehat{\kappa}(x)^{-2}\ge 1$ only widens the intervals, which can only make heterogeneity harder to distinguish compared to an unscaled version of \eqref{eq:Var_BLB_estim}.

\begin{table}[htbp]
\centering
\caption{Effect of the BLB scaling on confidence interval performance.}
\label{tab:blb_scaling_ci}
\scriptsize
\begin{tabular}{clcccccccccc}
\toprule
& & \multicolumn{5}{c}{BLB without scaling} & \multicolumn{5}{c}{BLB with scaling by  $\widehat{\kappa}(x)^{-2}$} \\
\cmidrule(lr){3-7} \cmidrule(lr){8-12}
Scenario & method & Bias & Cov. & Width & Miss$_{\text{lo}}$ & Miss$_{\text{up}}$ & Bias & Cov. & Width & Miss$_{\text{lo}}$ & Miss$_{\text{up}}$ \\
\midrule
\multirow{4}{*}{S1} & CF (LMS) & -0.002 & 0.190 & 1.15 & 0.420 & 0.390 & -0.002 & 0.658 & 2.26 & 0.116 & 0.226\\
& CF (MAD) & -0.003 & 0.194 & 1.15 & 0.416 & 0.390 & -0.002 & 0.674 & 2.26 & 0.114 & 0.212\\
& CF (MSD) & 0.004 & 0.419 & 1.27 & 0.234 & 0.348 & 0.002 & 0.954 & 2.49 & 0.031 & 0.015\\
& CF (MSE) & -0.002 & 0.265 & 1.19 & 0.356 & 0.379 & -0.001 & 0.823 & 2.34 & 0.089 & 0.088\\
\addlinespace
\multirow{4}{*}{S2} & CF (LMS) & 0.010 & 0.305 & 1.43 & 0.326 & 0.369 & 0.008 & 0.893 & 2.81 & 0.086 & 0.022\\
& CF (MAD) & 0.010 & 0.323 & 1.43 & 0.311 & 0.366 & 0.010 & 0.903 & 2.83 & 0.079 & 0.019\\
& CF (MSD) & 0.016 & 0.554 & 1.55 & 0.182 & 0.264 & 0.014 & 0.980 & 3.07 & 0.013 & 0.006\\
& CF (MSE) & 0.009 & 0.384 & 1.46 & 0.261 & 0.355 & 0.010 & 0.947 & 2.89 & 0.045 & 0.008\\
\addlinespace
\multirow{4}{*}{S3} & CF (LMS) & 0.001 & 0.325 & 1.44 & 0.496 & 0.178 & 0.001 & 0.798 & 2.84 & 0.154 & 0.047\\
& CF (MAD) & 0.001 & 0.353 & 1.44 & 0.496 & 0.152 & 0.002 & 0.805 & 2.83 & 0.147 & 0.047\\
& CF (MSD) & 0.009 & 0.516 & 1.54 & 0.431 & 0.053 & 0.006 & 0.882 & 3.03 & 0.067 & 0.051\\
& CF (MSE) & 0.002 & 0.413 & 1.46 & 0.486 & 0.101 & 0.003 & 0.823 & 2.89 & 0.130 & 0.047\\
\addlinespace
\multirow{4}{*}{S4} & CF (LMS) & 0.010 & 0.555 & 1.93 & 0.175 & 0.270 & 0.012 & 0.983 & 3.80 & 0.013 & 0.003\\
& CF (MAD) & 0.007 & 0.569 & 1.94 & 0.166 & 0.265 & 0.013 & 0.992 & 3.81 & 0.007 & 0.001\\
& CF (MSD) & 0.014 & 0.709 & 2.04 & 0.119 & 0.172 & 0.014 & 0.988 & 4.01 & 0.005 & 0.007\\
& CF (MSE) & 0.012 & 0.601 & 1.94 & 0.155 & 0.244 & 0.011 & 0.996 & 3.83 & 0.003 & 0.001\\
\bottomrule
\end{tabular}
\begin{notes}
 Scenarios according to Table \ref{tab:scenarios} with $K=10$ and $N=1000$ for $MC=40$ simulation runs. For each scenario and splitting method, results are shown without and with the scaling factor in \eqref{eq:Var_BLB_estim}. The average bias of CATE estimates across simulations for all methods and scenarios is close to zero and differs between the two blocks (without and with scaling) only by Monte Carlo error. Coverage, width and the lower and upper miss rates summarize the bootstrapped 95\% confidence intervals of CATE estimates as averages across simulation runs. The lower miss rate (Miss$_{\text{lo}}$) measures how often the true CATE lies below the lower bound of the estimated confidence interval and the upper miss rate (Miss$_{\text{up}}$) how often the true CATE lies above the upper bound of the estimated confidence interval.
\end{notes}
\end{table}

\section{Monte Carlo Simulation}
\label{appendix:MC_Sim}

All simulations are conducted in R using RStudio as the
development environment \citep{r_2025, rstudio_2025}. The code for reproducibility is provided at
\url{https://github.com/LennMass/medianCausalForest}.

\subsection{Implementation Details}
\label{appendix:implementation}

In this subsection, we further specify the forest-construction choices behind the results in
Sections~\ref{sec:simulations} and~\ref{sec:empappl}, so that the implemented
estimator can be matched to the criteria of Section~\ref{sec:MedianSplitting}
and to the conditions of \citet{wager_estimation_2018}. We implement the MSD, MAD, and LMS splitting
criteria by extending the \texttt{causalTree} framework
\citep{athey_recursive_2016, susanathey_susanatheycausaltree_2026} within the \texttt{htetree} package
\citep{brand_uncovering_2021, xu_htetree_2026}. Specifically, we modify the internal split evaluation routines
to accommodate the robust loss functions described in
Section~\ref{sec:MedianSplitting}. 
For variance estimation of the CATE, we
implement a BLB procedure following the approach
used in the \cite{athey_generalized_2019} that is briefly summarized in Appendix \ref{appendix:BLB}. The corresponding confidence intervals are constructed from these variance estimates. For average treatment effect estimation, we implemented an AIPW-style estimator described in Appendix \ref{appendix:AIPW_ATE_CF}.

\subsubsection{Split search and median convention}\label{app:split-search}

At each node, candidate splits are evaluated by a single left-to-right sweep over the sorted values of each candidate covariate, following the CART procedure of \citet{cart84}. For the median-based rules the leaf statistic is the Hodges--Lehmann estimate of Equation~\eqref{eq:HL-leafwise}, recomputed on both
child leaves at each candidate cut point. The estimate is formed by enumerating all $N_{1,l}N_{0,l}$ within-leaf pairwise differences and selecting their median. The implementation rebuilds this set at each evaluation
rather than updating it incrementally, at a per-split cost of $O(N_{1,l}N_{0,l})$. This cost is the practical reason the density penalty of
Equation~\eqref{eq:Q_MSD_H} is replaced by the surrogate of Equation~\eqref{eq:Q_MSD_code}, which avoids any additional density estimation
inside the sweep.

The split objective evaluated in the code is the cross-product surrogate of Equation~\eqref{eq:Q_MSD_code}, summed over child leaves. The explicit Hodges--Lehmann variance penalty of the honest criterion Equation~\eqref{eq:Q_MSD_H} is not computed. Consistent with Section~\ref{sec:MSD}, the simulations therefore evaluate the leading-order
surrogate, not the density-penalty criterion. The penalty weight $\alpha$ multiplies the entire surviving objective and so does not affect the selected split; it is retained only for compatibility with the honest form.

The within-leaf median is computed by the Quickselect algorithm of \citet{hoare_algorithm_1961}, as described in \citet{numerical_press_1992}, which returns a single central order statistic. 
For an even number of pairwise differences this is the selection median $\theta=\min\{\theta_1,\theta_2\}$ of Definition~\ref{median_distribution} rather than the averaged median $\theta=(\theta_1+\theta_2)/2$ used in the theory. 
The two coincide unless the two central order statistics straddle the midpoint, and their difference is of order $O_p(1/N_l)$ within a leaf under the continuous within-leaf densities of Assumption~\ref{ass:HL_shift}. 
The median-unbiasedness result of Theorem~\ref{thm:theorem_median_unbiased} is thus exact for the averaged median and holds for the implemented estimator up to this vanishing discrepancy.

\subsubsection{Correspondence to tree-growing conditions}

Table~\ref{tab:wa_conditions} maps each condition of Theorems~3.1 and~4.1 of \citet{wager_estimation_2018} to the mechanism that enforces it. The conditions are structural properties of the tree-growing protocol, shared by all four rules, since the rule enters only at split selection. Thereby, we leave honesty, regularity, symmetry, randomization, and the leaf estimator unchanged. These conditions govern the asymptotic regime in which the minimum leaf size shrinks in measure as $N$ grows. The fixed minimum leaf size in our runs departs from that regime and is the source of the finite-sample coverage sensitivity reported at $N=2000$ in Appendix~\ref{app:add_sim_res}, which is a matter of tuning rather than of protocol validity.

Covariates are treated as continuous, so exact ties occur with probability zero and are broken by observation order, keeping the split rule a symmetric function of the data in the sense of \citet{wager_estimation_2018}. Each candidate split
must leave at least two treated and at least two control observations in each child, enforced separately per arm. Moreover, each leaf node needs at least five observations in total to be considered for splitting. This guarantees that $\hat\mu_1$, $\hat\mu_0$, and hence the Hodges--Lehmann estimate are defined in every leaf, and together with the edge constraint on leaf counts supplies the $\alpha$-regularity condition of \citet{wager_estimation_2018}.

At each node, a random subset of the covariates is drawn as split candidates, supplying the random-split condition of \citet{wager_estimation_2018} that forces
leaves to shrink in every coordinate. Each tree is grown on an independent subsample drawn without replacement and uses the double-sample
construction of \citet{athey_recursive_2016} and \cite{wager_estimation_2018}. The
split-selection sample $\mathcal{S}^{tr}$ and the leaf-estimation sample $\mathcal{S}^{est}$ are disjoint, and leaf-level effects are estimated by the difference in means on $\mathcal{S}^{est}$ regardless of the splitting rule. This is the property that makes the forest honest and that underpins the leaf-level unbiasedness used throughout Section~\ref{sec:MedianSplitting}.

\begin{table}[htbp]
\centering
\caption{Forest conditions \citep{wager_estimation_2018} and the
mechanism enforcing each.}
\label{tab:wa_conditions}
\begin{tabular}{ll}
\toprule
Condition & Enforced by \\
\midrule
Honesty & Disjoint $\mathcal{S}^{tr},\mathcal{S}^{est}$ and leaves by difference in means \\
$\alpha$-regularity & Per-arm minimum node size in each child \\
Random-split & Random covariate subset at each node \\
Symmetry & Order-based tie-breaking \\
Subsampling rate & Subsampling without replacement \\
\bottomrule
\end{tabular}
\end{table}

\subsubsection{ATE estimators}
\label{appendix:ATE_estimators}

We compute the competing ATE estimators following
\citet{ghosh_robustness_2026} and using the accompanying replication code in \url{https://github.com/ghoshadi/RRE}. 
The difference-in-means (Diff.Mean) estimator is the difference between treatment and control group sample means. 
The OLS-with-interaction (Lin.OLS) estimator follows \citet{lin_agnostic_2013}: we regress the outcome on the treatment indicator, the mean-centered numeric covariates, and their interactions with treatment, enter factor covariates as main effects, and report the treatment coefficient. 
The regression-adjusted (R.Adj) estimator of \citet{ghosh_robustness_2026} is based on \citet{Rosenbaum1993} and solves for the shift that balances the Wilcoxon rank-sum of the covariate-adjusted residuals across treatment and control groups, obtained by root-finding. The efficient-influence-function (EIF) and weighted-average-of-quantiles (WAQ) estimators target the quantile treatment effect of \citet{Athey2023}. 
The two causal forest, CF (MSE) and CF (MSD) estimators aggregate the forest CATE estimates by the AIPW construction of Appendix~\ref{appendix:AIPW_ATE_CF}.

\subsubsection{Performance Criteria for Simulation Study}
\label{appendix:mc_performance_criteria}

A sample consists of $N \in \{1000, 2000\}$ observations in total. For each sample, we use $4$-fold cross-fitting to estimate $\tau(X_i)$ for each $i$ out-of-sample. That is, we use $75 \%$ of the sample (three out of four folds) for model training. For each tree in a causal forest, half of the data is used for tree building, and the other half is used for leaf prediction to ensure honesty. Then, we estimate $\tau(X_i)$ for the $25 \%$ of test observations in the $k$th fold and repeat this process for each fold to get $\hat\tau(X_i)$ for each $i \in \{1, \ldots, N\}$.
We generate $MC$ independent samples of observations such that $mc \in \{1, ...., MC\}$. 
Precision of the CATE predictions is evaluated in terms of RMSE and absolute bias ($\abs{\text{Bias}}$) taken as an average over the distinct samples by
\begin{align}
\label{avg_RMSE}
\operatorname{\overline{RMSE}} &= \frac{1}{MC} \sum_{mc=1}^{MC} \left[ \sqrt{\frac{1}{N}\sum_{i = 1}^{N} \left( \tau(X_{i}^{mc})-\hat\tau(X_{i}^{mc}) \right)^2} \right], 
\end{align}
\begin{align}
\label{avg_bias}
\operatorname{\overline{|Bias|}} &= \frac{1}{MC} \sum_{mc=1}^{MC} \left[ \frac{1}{N}\sum_{i = 1}^{N} \left| \tau(X_{i}^{mc})-\hat\tau(X_{i}^{mc})\right| \right].
\end{align}

Moreover, we estimate the variance of the treatment effects using the BLB described in Appendix~\ref{appendix:BLB} to evaluate coverage properties of the considered methods across simulation settings. We assess the confidence interval coverage for the 95\% confidence interval and the width of the 95\% confidence interval by averaging across datasets by 
\begin{align}
    \label{coverage_def}
    \text{Coverage}_{CI_{95}\left( \hat\tau(X_{i}^{mc}) \right)} &= \frac{1}{MC} \sum_{mc=1}^{MC} \left[ \frac{1}{N}\sum_{i = 1}^{N}  \mathds{1}\left( \tau\left(X_{i}^{mc}\right) \in {CI_{95}\left( \hat\tau(X_{i}^{mc}) \right)} \right) \right]
\end{align}
\begin{align}
    \label{width_def}
    \text{Width}_{CI_{95}\left( \hat\tau(X_{i}^{mc}) \right)} &= \frac{1}{MC} \sum_{mc=1}^{MC} \left[ \frac{1}{N}\sum_{i = 1}^{N}   \operatorname*{length}\left( {CI_{95}\left( \hat\tau(X_{i}^{mc}) \right)} \right) \right],    
\end{align}
where ${CI_{95}\left( \hat\tau(X_{i}^{mc}) \right)}$ is the corresponding estimated 95\% confidence interval for the estimated CATE, $\hat\tau(X_{i}^{mc})$, in the specific $mc$-dataset.

\subsection{Additional Simulation Results}
\label{app:add_sim_res}

Tables~\ref{tab:cate_samplesize}--\ref{tab:ate_covariatesize} report
sensitivity checks that vary the sample size ($N \in \{1000, 2000\}$) and
the number of covariates ($K \in \{5, 20\}$) for both CATE and ATE
estimation, following the main data-generating process in Section \ref{sec:simulations} and Table \ref{tab:scenarios}. Taken together, these robustness checks confirm that the relative
performance differences documented in the main text are stable across the
sample sizes and covariate dimensions considered here.

Table~\ref{tab:cate_samplesize} shows results for CATE estimation while increasing the sample size from $N = 1000$ to $N = 2000$.
Doubling the sample size does not improve accuracy by much across scenarios. We note here that the $N=2000$ columns use 40 replications against 100 replications at $N=1000$, so their coverage figures carry larger Monte Carlo error. RMSE and absolute bias stay flat across the two sample sizes while the confidence intervals narrow, so coverage falls for all four criteria in S1 to S3. In S4 the large variance contributed by the skewed $U_i$ keeps all methods overcovering, so the narrowing is not yet enough to pull coverage down to nominal.
The suspected cause is the fixed minimum leaf node size parameter within the causal forest. Leaves do not refine as $N$ grows, point estimates stop improving, and the variance estimates keep shrinking, resulting in overconfident intervals. 
Because this affects every rule the same way, the ordering is preserved. CF (MSD) keeps the lowest RMSE and absolute bias in each scenario and loses the least coverage, while MSE undercovers in S1 and S3 at $N=2000$, and MAD and LMS undercover in S1 to S3, while in S4 all rules overcover. Recovering nominal coverage at larger $N$ would require the minimum leaf size to grow with the sample, which we do not impose here. We therefore read the $N=2000$ column in Table~\ref{tab:cate_samplesize} as a stability check on the relative ordering and a deliberate illustration of leaf-size sensitivity, not as evidence on large-sample accuracy.
This degradation is a statement about the asymptotic regime, not about protocol validity. The conditions of Theorems~3.1 and~4.1 of \citet{wager_estimation_2018}, discussed for our forest in Appendix~\ref{appendix:implementation}, are structural properties that hold at every sample size. Their asymptotic conclusion, however, requires the minimum leaf size to shrink in measure as $N$ grows, so that leaves contain diverging counts while covering vanishing regions. Holding the minimum leaf size fixed, as we do, keeps the protocol honest and regular but places the $N=2000$ run outside that regime, which is the suspected reason why coverage degrades there. The forest remains a valid honest forest, but the rate condition is not met.

For ATE estimation in Table~\ref{tab:ate_samplesize},
CF (MSE) and CF (MSD) produce similar point estimates and interval
widths, confirming that the splitting rule has little influence once
leaf-level estimates are averaged to the population level. Among the non-forest estimators, the difference-in-means and linear OLS estimators perform well in scenarios~S1 and S2, where all regularity conditions hold.
In scenarios~S3 and S4, however, the regression-adjusted estimator (R.Adj), the efficient influence function estimator (EIF), and the weighted-average-of-quantiles estimator (WAQ) exhibit severe bias and coverage collapse, because each targets a rank- or quantile-based contrast that departs from the mean ATE once the responders of S3 and the skewed effects of S4 break the symmetry under which the targets coincide. This is estimand mismatch rather than model misspecification.

Tables~\ref{tab:cate_covariatesize} and \ref{tab:ate_covariatesize} increase the covariate dimension from $K = 5$ to $K = 20$ by adding 15 noise variables, while keeping the sample size fixed at $N=1000$. For CATE estimation in Table~\ref{tab:cate_covariatesize}, all
forest-based methods show only a modest increase in RMSE when moving to the higher-dimensional setting, indicating that the subsampling and split-selection mechanisms of the causal forest provide a degree of built-in regularization against irrelevant covariates. CF (MSD) is the most affected in absolute terms, yet it continues to achieve the lowest RMSE and the best coverage across all scenarios. 
The ranking among methods is otherwise preserved. For ATE estimation (Table~\ref{tab:ate_covariatesize}), the same pattern as in the sample-size check emerges: CF (MSE) and CF (MSD) are again nearly indistinguishable, the simple estimators remain reliable in the well-specified scenarios, and the semiparametric estimators break down in scenarios~S3 and S4.

Further, we compare the runtime of the four splitting criteria in the causal forest in Table \ref{tab:runtime}.
The mean-based MSE criterion and LMS each complete a forest fit in under a minute, whereas MSD and MAD, both anchored on the Hodges--Lehmann estimator, take roughly four and six times as long. Because LMS is as robust as the two slow criteria yet stays close to MSE in cost, the overhead traces to the Hodges--Lehmann computation and not to robust splitting as such.  The gap follows directly from the split search of
Appendix~\ref{app:split-search} where we describe that the anchor is rebuilt from all within-leaf pairwise differences at each candidate cut point. In contrast, the mean-based criterion updates running sums in constant time and LMS forms no pairs. Further, we note that our implementation of LMS, MAD, and MSD is not tuned for speed. Computing the same estimates through modified algorithms that do not materialize the pairwise set might lower the computational cost of MSD and MAD without changing any result \citep{monahan_algorithm_1984}.

\begin{table}[htbp]
\tiny
\centering
\caption{Sample size comparison with $N \in \{1000, 2000\}$ for CATE estimation with causal forests (CF) using different splitting criteria.}
\label{tab:cate_samplesize}
\begin{tabular}{clcccccccc}
\toprule
& & \multicolumn{4}{c}{$N = 1000$} & \multicolumn{4}{c}{$N = 2000$} \\
\cmidrule(lr){3-6} \cmidrule(lr){7-10}
Sc. & method & RMSE & $\abs{\text{Bias}}$ & Cov. & Width & RMSE & $\abs{\text{Bias}}$ & Cov. & Width \\
\midrule
\multirow{4}{*}{S1}& CF (MSE) & 0.89 (0.002) & 0.82 (0.002) & 0.88 (0.006) & 2.46 (0.006) &  0.89 (0.003) & 0.82 (0.003) & 0.71 (0.012) & 2.13 (0.008)\\
                   & CF (LMS) & 1.00 (0.001) & 0.91 (0.001) & 0.68 (0.008) & 2.28 (0.004) &  0.99 (0.001) & 0.91 (0.001) & 0.55 (0.003) & 1.95 (0.007) \\
                   & CF (MAD) & 0.98 (0.001) & 0.90 (0.001) & 0.71 (0.009) & 2.30 (0.005) & 0.98 (0.001) & 0.90 (0.001) & 0.55 (0.003) & 1.97 (0.007) \\
                   & CF (MSD) & 0.76 (0.003) & 0.70 (0.003) & 0.99 (0.001) & 2.77 (0.006) & 0.75 (0.004) & 0.70 (0.004) & 0.95 (0.003) & 2.47 (0.007)\\
\addlinespace
\multirow{4}{*}{S2} & CF (MSE) & 0.91 (0.003) & 0.83 (0.002) & 0.96 (0.003) & 3.02 (0.017) & 0.91 (0.003) & 0.83 (0.003) & 0.92 (0.005) & 2.67 (0.023)\\
                    & CF (LMS) & 1.00 (0.002) & 0.91 (0.002) & 0.90 (0.004) & 2.86 (0.016)& 1.00 (0.001) & 0.91 (0.001) & 0.78 (0.015) & 2.47 (0.025)\\
                    & CF (MAD) & 0.98 (0.002) & 0.90 (0.002) & 0.92 (0.003) & 2.88 (0.017)& 0.98 (0.001) & 0.90 (0.002) & 0.82 (0.013) & 2.52 (0.025)\\
                    & CF (MSD) & 0.79 (0.004) & 0.73 (0.004) & 1.00 (0.001) & 3.37 (0.016) & 0.78 (0.005) & 0.72 (0.004) & 0.98 (0.001) & 3.02 (0.025)\\
\addlinespace
\multirow{4}{*}{S3} & CF (MSE) & 2.48 (0.013) & 1.34 (0.007) & 0.84 (0.003) & 3.08 (0.009)& 2.44 (0.015) & 1.32 (0.008) & 0.78 (0.006) & 2.66 (0.011)\\
                    & CF (LMS) & 2.62 (0.015) & 1.47 (0.008) & 0.79 (0.005) & 2.89 (0.008) & 2.57 (0.017) & 1.45 (0.009) & 0.64 (0.012) & 2.41 (0.010)\\
                    & CF (MAD) & 2.59 (0.015) & 1.44 (0.008) & 0.81 (0.004) & 2.92 (0.007)& 2.54 (0.016) & 1.42 (0.009) & 0.68 (0.011) & 2.47 (0.010)\\
                    & CF (MSD) & 2.33 (0.012) & 1.12 (0.008) & 0.93 (0.002) & 3.48 (0.010)& 2.30 (0.014) & 1.10 (0.008) & 0.89 (0.004) & 3.03 (0.016)\\
\addlinespace
\multirow{4}{*}{S4} & CF (MSE) & 0.97 (0.005) & 0.87 (0.004) & 0.99 (0.001) & 3.95 (0.027)& 0.94 (0.004) & 0.86 (0.003) & 0.99 (0.003) & 3.44 (0.034)\\
                    & CF (LMS) & 1.03 (0.004) & 0.92 (0.003) & 0.98 (0.003) & 3.83 (0.028)& 1.01 (0.003) & 0.91 (0.002) & 0.94 (0.005) & 3.29 (0.032)\\
                    & CF (MAD) & 1.00 (0.004) & 0.90 (0.003) & 0.99 (0.002) & 3.91 (0.029)& 0.98 (0.004) & 0.89 (0.002) & 0.97 (0.005) & 3.38 (0.031)\\
                    & CF (MSD) & 0.88 (0.007) & 0.78 (0.005) & 1.00 (0.001) & 4.42 (0.030)& 0.84 (0.007) & 0.76 (0.006) & 0.99 (0.001) & 3.94 (0.035)\\
\bottomrule
\end{tabular}
\begin{notes}
Scenarios according to Table \ref{tab:scenarios}. The left side ($N=1000$) corresponds to the values in Figure \ref{fig:CATE_precision_s1s4}. Note that we used 100 $MC$ runs for the main simulation setup ($N=1000$) and 40 $MC$ runs for the larger sample size ($N=2000$).
\end{notes}
\end{table}

\begin{table}[htbp]
\tiny
\centering
\caption{Sample size comparison with $N \in \{1000, 2000\}$ for ATE estimation.}
\label{tab:ate_samplesize}
\vspace{0.8em}
\begin{tabular}{clcccccccc}
\toprule
& & \multicolumn{4}{c}{$N = 1000$} & \multicolumn{4}{c}{$N = 2000$} \\
\cmidrule(lr){3-6} \cmidrule(lr){7-10}
Sc. & method & RMSE & $\abs{\text{Bias}}$ & Cov. & Width & RMSE & $\abs{\text{Bias}}$ & Cov. & Width \\
\midrule
\multirow{7}{*}{S1} & CF (MSE) & 0.08 (0.005) & 0.06 (0.005) & 0.96 (0.020) & 0.31 (0.001) & 0.05 (0.006) & 0.04 (0.005) & 0.95 (0.035) & 0.21 (0.001)\\
                    & CF (MSD) & 0.08 (0.005) & 0.06 (0.005) & 0.96 (0.020) & 0.30 (0.001)  & 0.05 (0.006) & 0.04 (0.005) & 0.95 (0.035) & 0.21 (0.001)\\
                    & Diff.Mean & 0.07 (0.005) & 0.06 (0.004) & 0.97 (0.017) & 0.30 (0.001) & 0.05 (0.007) & 0.04 (0.005) & 0.95 (0.035) & 0.21 (0.001)\\
                    & Lin.OLS & 0.07 (0.004) & 0.06 (0.004) & 0.94 (0.024) & 0.26 (0.001) & 0.05 (0.007) & 0.04 (0.005) & 0.95 (0.035) & 0.19 (0.000)\\
                    & R.Adj & 0.07 (0.004) & 0.06 (0.004) & 0.97 (0.017) & 0.30 (0.002) & 0.05 (0.007) & 0.04 (0.005) & 0.95 (0.035) & 0.21 (0.002)\\
                    & EIF & 0.09 (0.006) & 0.07 (0.005) & 0.88 (0.033) & 0.25 (0.001) & 0.07 (0.010) & 0.05 (0.007) & 0.85 (0.057) & 0.18 (0.001)\\
                    & WAQ & 0.09 (0.005) & 0.07 (0.005) & 0.86 (0.035) & 0.25 (0.001) & 0.07 (0.009) & 0.05 (0.007) & 0.90 (0.048) & 0.18 (0.001)\\
                    
\addlinespace
\multirow{7}{*}{S2} & CF (MSE) & 0.11 (0.008) & 0.08 (0.007) & 0.98 (0.014) & 0.48 (0.006)& 0.08 (0.008) & 0.06 (0.007) & 1.00 (0.000) & 0.34 (0.006)\\
& CF (MSD) & 0.11 (0.008) & 0.08 (0.007) & 0.98 (0.014) & 0.48 (0.006)& 0.08 (0.008) & 0.06 (0.007) & 1.00 (0.000) & 0.34 (0.006)\\
& Diff.Mean & 0.11 (0.008) & 0.08 (0.007) & 0.95 (0.022) & 0.46 (0.006)& 0.08 (0.008) & 0.06 (0.008) & 0.98 (0.025) & 0.33 (0.007)\\
& Lin.OLS & 0.10 (0.007) & 0.08 (0.006) & 0.96 (0.020) & 0.43 (0.007)& 0.08 (0.008) & 0.06 (0.008) & 0.92 (0.042) & 0.32 (0.007)\\
& R.Adj & 0.08 (0.006) & 0.06 (0.005) & 0.97 (0.017) & 0.36 (0.003)& 0.06 (0.007) & 0.05 (0.006) & 0.98 (0.025) & 0.25 (0.002)\\
& EIF & 0.11 (0.008) & 0.08 (0.007) & 0.85 (0.036) & 0.30 (0.002)& 0.06 (0.008) & 0.05 (0.007) & 0.88 (0.053) & 0.21 (0.001)\\
& WAQ & 0.10 (0.007) & 0.08 (0.006) & 0.86 (0.035) & 0.30 (0.002)& 0.06 (0.008) & 0.05 (0.007) & 0.90 (0.048) & 0.21 (0.001)\\
                    
\addlinespace
\multirow{7}{*}{S3} & CF (MSE) & 0.09 (0.006) & 0.07 (0.005) & 1.00 (0.000) & 0.49 (0.002) & 0.06 (0.009) & 0.05 (0.007) & 0.98 (0.025) & 0.33 (0.002)\\
& CF (MSD) & 0.08 (0.006) & 0.07 (0.005) & 1.00 (0.000) & 0.49 (0.002) & 0.06 (0.009) & 0.05 (0.007) & 0.98 (0.025) & 0.33 (0.002)\\
& Diff.Mean & 0.10 (0.006) & 0.09 (0.006) & 1.00 (0.000) & 0.53 (0.003) & 0.07 (0.008) & 0.06 (0.007) & 0.98 (0.025) & 0.36 (0.003)\\
& Lin.OLS & 0.09 (0.005) & 0.07 (0.005) & 1.00 (0.000) & 0.42 (0.002)& 0.06 (0.006) & 0.05 (0.006) & 0.98 (0.025) & 0.29 (0.002)\\
& R.Adj & 0.36 (0.009) & 0.35 (0.009) & 0.05 (0.022) & 0.36 (0.004)& 0.33 (0.010) & 0.32 (0.010) & 0.00 (0.000) & 0.25 (0.003)\\
& EIF & 0.45 (0.012) & 0.44 (0.012) & 0.00 (0.000) & 0.25 (0.001)& 0.40 (0.013) & 0.39 (0.013) & 0.00 (0.000) & 0.18 (0.001)\\
& WAQ & 0.30 (0.017) & 0.26 (0.016) & 0.27 (0.045) & 0.25 (0.001)& 0.26 (0.025) & 0.21 (0.025) & 0.35 (0.076) & 0.18 (0.001)\\
\addlinespace
\multirow{7}{*}{S4} & CF (MSE) & 0.25 (0.016) & 0.20 (0.015) & 0.93 (0.026) & 0.87 (0.014) & 0.17 (0.018) & 0.14 (0.017) & 0.90 (0.048) & 0.61 (0.012)\\
& CF (MSD) & 0.25 (0.016) & 0.20 (0.015) & 0.93 (0.026) & 0.87 (0.014) & 0.17 (0.018) & 0.14 (0.017) & 0.90 (0.048) & 0.61 (0.012)\\
& Diff.Mean & 0.24 (0.017) & 0.18 (0.015) & 0.93 (0.026) & 0.81 (0.012) & 0.15 (0.014) & 0.13 (0.014) & 0.90 (0.048) & 0.57 (0.011)\\
& Lin.OLS & 0.24 (0.017) & 0.19 (0.015) & 0.90 (0.030) & 0.79 (0.012)& 0.15 (0.015) & 0.12 (0.015) & 0.88 (0.053) & 0.56 (0.011)\\
& R.Adj & 0.90 (0.012) & 0.89 (0.013) & 0.00 (0.000) & 0.52 (0.005)& 0.91 (0.015) & 0.90 (0.015) & 0.00 (0.000) & 0.36 (0.005)\\
& EIF & 0.88 (0.025) & 0.84 (0.027) & 0.01 (0.010) & 0.25 (0.001)& 0.83 (0.030) & 0.81 (0.030) & 0.00 (0.000) & 0.18 (0.001)\\
& WAQ & 0.58 (0.026) & 0.51 (0.028) & 0.11 (0.031) & 0.25 (0.001)& 0.50 (0.034) & 0.45 (0.035) & 0.08 (0.042) & 0.18 (0.001)\\
\bottomrule
\end{tabular}
\begin{notes}
Scenarios according to Table \ref{tab:scenarios}. The left side ($N=1000$) corresponds to the values in Figure \ref{fig:CATE_precision_s1s4}. We used 100 $MC$ runs for the main simulation setup ($N=1000$) and 40 $MC$ runs for the larger sample size ($N=2000$). Moreover, we do not report ATE estimates based on CF (MAD) and CF (LMS) for brevity and as they show similar results as CF (MSE) and CF (MSD) due to the discussion in Appendix \ref{appendix:AIPW_ATE_CF}. 
\end{notes}
\end{table}

\begin{table}[htbp]
\tiny
\centering
\caption{Covariate size comparison with $K \in \{5, 20\}$ for CATE estimation with causal forests (CF) using different splitting criteria.}
\label{tab:cate_covariatesize}

\begin{tabular}{clcccccccc}
\toprule
& & \multicolumn{4}{c}{$K = 5$} & \multicolumn{4}{c}{$K = 20$} \\
\cmidrule(lr){3-6} \cmidrule(lr){7-10} 
Scenario & method & RMSE & $\abs{\text{Bias}}$ & Cov. & Width & RMSE & $\abs{\text{Bias}}$ & Cov. & Width \\
\midrule
\multirow{4}{*}{S1}& CF (MSE) & 0.90 (0.004) & 0.82 (0.004) & 0.88 (0.009) & 2.46 (0.009) & 0.94 (0.002) & 0.86 (0.003) & 0.81 (0.012) & 2.38 (0.007)\\

 & CF (LMS) & 1.00 (0.002) & 0.91 (0.002) & 0.68 (0.013) & 2.30 (0.007) & 1.00 (0.002) & 0.91 (0.002) & 0.70 (0.015) & 2.29 (0.008)\\

 & CF (MAD) & 0.98 (0.002) & 0.89 (0.002) & 0.71 (0.015) & 2.31 (0.007) & 0.99 (0.002) & 0.90 (0.002) & 0.72 (0.016) & 2.30 (0.006)\\
& CF (MSD) & 0.75 (0.004) & 0.69 (0.004) & 0.99 (0.002) & 2.77 (0.010) & 0.87 (0.003) & 0.80 (0.003) & 0.93 (0.003) & 2.63 (0.007)\\
\addlinespace
\multirow{4}{*}{S2} & CF (MSE) & 0.90 (0.003) & 0.83 (0.003) & 0.97 (0.004) & 3.01 (0.021) & 0.96 (0.003) & 0.87 (0.003) & 0.94 (0.005) & 2.98 (0.026)\\

 & CF (LMS) & 1.00 (0.002) & 0.91 (0.003) & 0.90 (0.007) & 2.85 (0.023) & 1.01 (0.002) & 0.91 (0.002) & 0.90 (0.008) & 2.87 (0.024)\\

 & CF (MAD) & 0.98 (0.002) & 0.89 (0.003) & 0.92 (0.005) & 2.87 (0.020) & 0.99 (0.002) & 0.90 (0.002) & 0.91 (0.007) & 2.92 (0.027)\\

& CF (MSD) & 0.77 (0.006) & 0.71 (0.006) & 1.00 (0.001) & 3.37 (0.023) & 0.90 (0.004) & 0.82 (0.003) & 0.98 (0.003) & 3.27 (0.026)\\

\addlinespace
\multirow{4}{*}{S3} & CF (MSE) & 2.48 (0.024) & 1.33 (0.014) & 0.84 (0.005) & 3.08 (0.016) & 2.50 (0.023) & 1.38 (0.012) & 0.83 (0.004) & 3.00 (0.015)\\

 & CF (LMS) & 2.62 (0.026) & 1.47 (0.014) & 0.79 (0.007) & 2.86 (0.016) & 2.57 (0.024) & 1.45 (0.013) & 0.80 (0.006) & 2.89 (0.013)\\

 & CF (MAD) & 2.57 (0.025) & 1.43 (0.014) & 0.81 (0.006) & 2.90 (0.016) & 2.56 (0.024) & 1.43 (0.013) & 0.81 (0.004) & 2.93 (0.013)\\

& CF (MSD) & 2.31 (0.022) & 1.14 (0.014) & 0.93 (0.003) & 3.42 (0.018) & 2.44 (0.022) & 1.27 (0.011) & 0.87 (0.003) & 3.29 (0.014)\\

\addlinespace
\multirow{4}{*}{S4} & CF (MSE) & 0.92 (0.005) & 0.83 (0.005) & 1.00 (0.000) & 3.98 (0.032) & 0.99 (0.005) & 0.89 (0.004) & 0.99 (0.002) & 3.94 (0.039)\\

 & CF (LMS) & 1.02 (0.003) & 0.92 (0.003) & 0.99 (0.003) & 3.80 (0.035) & 1.02 (0.005) & 0.92 (0.003) & 0.98 (0.003) & 3.84 (0.040)\\

 & CF (MAD) & 0.98 (0.004) & 0.88 (0.004) & 1.00 (0.002) & 3.86 (0.033) & 1.01 (0.005) & 0.91 (0.003) & 0.99 (0.003) & 3.89 (0.039)\\

& CF (MSD) & 0.82 (0.010) & 0.73 (0.009) & 1.00 (0.000) & 4.44 (0.037) & 0.94 (0.007) & 0.85 (0.005) & 1.00 (0.001) & 4.30 (0.042)\\

\bottomrule
\end{tabular}
\begin{notes}
Scenarios according to Table \ref{tab:scenarios}. We used $MC=40$ runs for each scenario and covariate size. 
\end{notes}
\end{table}

\begin{table}[htbp]
\tiny
\centering
\caption{Covariate size comparison with $K \in \{5, 20\}$ for ATE estimation.}
\vspace{0.8em}
\label{tab:ate_covariatesize}
\begin{tabular}{clcccccccc}
\toprule
& & \multicolumn{4}{c}{$K = 5$} & \multicolumn{4}{c}{$K = 20$} \\
\cmidrule(lr){3-6} \cmidrule(lr){7-10}
Scenario & method & RMSE & $\abs{\text{Bias}}$ & Cov. & Width & RMSE & $\abs{\text{Bias}}$ & Cov. & Width \\
\midrule
\multirow{7}{*}{S1} & CF (MSE) & 0.08 (0.008) & 0.07 (0.007) & 0.98 (0.025) & 0.32 (0.002) & 0.07 (0.007) & 0.06 (0.007) & 0.95 (0.035) & 0.30 (0.001)\\
& CF (MSD) & 0.08 (0.008) & 0.07 (0.007) & 0.98 (0.025) & 0.32 (0.002)& 0.07 (0.007) & 0.06 (0.007) & 0.95 (0.035) & 0.30 (0.001)\\
& Diff.Mean & 0.07 (0.006) & 0.06 (0.006) & 0.98 (0.025) & 0.30 (0.001)& 0.07 (0.008) & 0.06 (0.007) & 0.92 (0.042) & 0.30 (0.001)\\
& Lin.OLS & 0.07 (0.006) & 0.06 (0.006) & 0.92 (0.042) & 0.26 (0.001)& 0.07 (0.006) & 0.06 (0.006) & 0.98 (0.025) & 0.26 (0.001)\\
& R.Adj & 0.08 (0.007) & 0.06 (0.007) & 0.92 (0.042) & 0.29 (0.003)& 0.07 (0.006) & 0.06 (0.007) & 0.98 (0.025) & 0.30 (0.004)\\
& EIF & 0.08 (0.008) & 0.07 (0.008) & 0.90 (0.048) & 0.25 (0.002)& 0.09 (0.009) & 0.08 (0.008) & 0.82 (0.061) & 0.25 (0.002)\\
& WAQ & 0.07 (0.007) & 0.06 (0.007) & 0.90 (0.048) & 0.25 (0.002)& 0.09 (0.008) & 0.07 (0.008) & 0.85 (0.057) & 0.25 (0.002)\\
\addlinespace
\multirow{7}{*}{S2} & CF (MSE) & 0.13 (0.013) & 0.10 (0.012) & 0.90 (0.048) & 0.47 (0.009) & 0.13 (0.011) & 0.10 (0.011) & 1.00 (0.000) & 0.51 (0.008)\\
& CF (MSD) & 0.13 (0.013) & 0.10 (0.012) & 0.90 (0.048) & 0.47 (0.009) & 0.12 (0.010) & 0.10 (0.011) & 1.00 (0.000) & 0.51 (0.008)\\
& Diff.Mean & 0.12 (0.013) & 0.10 (0.012) & 0.90 (0.048) & 0.46 (0.008) & 0.12 (0.010) & 0.10 (0.010) & 0.98 (0.025) & 0.46 (0.006)\\
& Lin.OLS & 0.13 (0.013) & 0.10 (0.012) & 0.88 (0.053) & 0.44 (0.008) & 0.13 (0.012) & 0.10 (0.012) & 0.95 (0.035) & 0.43 (0.007)\\
& R.Adj & 0.09 (0.008) & 0.07 (0.007) & 0.98 (0.025) & 0.38 (0.004) & 0.09 (0.012) & 0.07 (0.009) & 0.92 (0.042) & 0.35 (0.005)\\
& EIF & 0.10 (0.011) & 0.08 (0.010) & 0.90 (0.048) & 0.30 (0.002) & 0.08 (0.008) & 0.07 (0.008) & 0.95 (0.035) & 0.30 (0.003)\\
& WAQ & 0.10 (0.009) & 0.08 (0.009) & 0.92 (0.042) & 0.30 (0.002) & 0.08 (0.007) & 0.07 (0.007) & 0.95 (0.035) & 0.30 (0.003)\\
\addlinespace
\multirow{7}{*}{S3} & CF (MSE) & 0.08 (0.008) & 0.07 (0.007) & 1.00 (0.000) & 0.48 (0.004)& 0.11 (0.010) & 0.09 (0.009) & 1.00 (0.000) & 0.53 (0.006)\\
& CF (MSD) & 0.08 (0.008) & 0.07 (0.007) & 1.00 (0.000) & 0.48 (0.004)& 0.11 (0.010) & 0.09 (0.009) & 1.00 (0.000) & 0.52 (0.006)\\
& Diff.Mean & 0.09 (0.010) & 0.07 (0.008) & 1.00 (0.000) & 0.51 (0.005)& 0.10 (0.011) & 0.08 (0.009) & 0.98 (0.025) & 0.52 (0.006)\\
& Lin.OLS & 0.09 (0.009) & 0.08 (0.008) & 0.98 (0.025) & 0.40 (0.004)& 0.08 (0.008) & 0.06 (0.007) & 1.00 (0.000) & 0.42 (0.003)\\
& R.Adj & 0.35 (0.017) & 0.33 (0.017) & 0.08 (0.042) & 0.36 (0.006)& 0.35 (0.012) & 0.34 (0.013) & 0.02 (0.025) & 0.36 (0.005)\\
& EIF & 0.43 (0.018) & 0.42 (0.019) & 0.00 (0.000) & 0.25 (0.001)& 0.47 (0.018) & 0.45 (0.018) & 0.00 (0.000) & 0.25 (0.002)\\
& WAQ & 0.31 (0.029) & 0.26 (0.028) & 0.33 (0.075) & 0.25 (0.001)& 0.35 (0.034) & 0.29 (0.031) & 0.20 (0.064) & 0.25 (0.002)\\
\addlinespace
\multirow{7}{*}{S4} & CF (MSE) & 0.21 (0.021) & 0.16 (0.020) & 0.90 (0.048) & 0.86 (0.019) & 0.18 (0.019) & 0.14 (0.018) & 1.00 (0.000) & 0.91 (0.021)\\
& CF (MSD) & 0.21 (0.021) & 0.16 (0.020) & 0.90 (0.048) & 0.86 (0.019) & 0.18 (0.019) & 0.14 (0.018) & 1.00 (0.000) & 0.91 (0.021)\\
& Diff.Mean & 0.20 (0.020) & 0.15 (0.020) & 0.92 (0.042) & 0.80 (0.016) & 0.15 (0.014) & 0.12 (0.014) & 0.98 (0.025) & 0.80 (0.013)\\
& Lin.OLS & 0.20 (0.020) & 0.16 (0.019) & 0.92 (0.042) & 0.80 (0.018) & 0.15 (0.014) & 0.13 (0.014) & 0.98 (0.025) & 0.78 (0.014)\\
& R.Adj & 0.87 (0.020) & 0.86 (0.020) & 0.00 (0.000) & 0.57 (0.009) & 0.92 (0.016) & 0.91 (0.017) & 0.00 (0.000) & 0.50 (0.005)\\
& EIF & 0.88 (0.037) & 0.85 (0.038) & 0.00 (0.000) & 0.25 (0.002) & 0.88 (0.037) & 0.85 (0.038) & 0.00 (0.000) & 0.25 (0.001)\\
& WAQ & 0.58 (0.042) & 0.52 (0.042) & 0.05 (0.035) & 0.25 (0.002) & 0.57 (0.038) & 0.51 (0.040) & 0.08 (0.042) & 0.25 (0.001)\\
\bottomrule
\end{tabular}
\begin{notes}
Scenarios according to Table \ref{tab:scenarios}. Note that we used $MC=40$ runs for each scenario and covariate size. Moreover, we do not report ATE estimates based on CF (MAD) and CF (LMS) for brevity and as they show similar results as CF (MSE) and CF (MSD) due to the discussion in Appendix \ref{appendix:AIPW_ATE_CF}.
\end{notes}
\end{table}

\begin{table}[htbp]
\centering
\caption{Computational cost of the four splitting criteria under the S1 design.}
\vspace{0.8em}
\label{tab:runtime}
\begin{tabular}{l c r r}
\toprule
Method & HL anchor & Median time (in seconds) & Median time relative to MSE \\
\midrule
CF (MSE) & no  &  49.3 & 1.00 \\
CF (LMS) & no  &  56.7 & 1.15 \\
CF (MSD) & yes & 203.5 & 4.13 \\
CF (MAD) & yes & 296.9 & 6.03 \\
\bottomrule
\end{tabular}
\begin{notes}
Entries are the median wall-clock time for a cross-validated causal forest fit over 10 replications, pooled across $N \in \{1000, 2000\}$ and the default and high-dimensional covariate settings $K \in \{10, 20\}$. Doubling either $N$ or $K$ changed any criterion's time by at most a factor of $1.09$. The insensitivity to $N$ reflects the design rather than a property of the criteria, as we fix each tree's subsample size at 200. HL marks the two criteria anchored on the Hodges-Lehmann estimator. 
\end{notes}
\end{table}

\clearpage
\newpage

\section{Additional Application Results}
\label{append:emp_appl}

The Progresa sample in Section \ref{sec:progresa} covers 417 electoral precincts, of which 138 are assigned to the late (control) group and 279 to the early (treatment) group \citep{de_la_o_conditional_2013}. We present descriptive statistics to the dataset in Figure \ref{fig:outcome_distribution_progresa} and Table \ref{tab:desc_stats_progresa}.
The outcome, the PRI vote share in 2000, is right-skewed in the pooled sample (skewness $1.48$), and its dispersion is markedly larger in the treatment group (standard deviation $19.73$ against $14.82$). 
The covariates that count votes and population are themselves right-skewed, with means well above their medians. 
Total population is the most extreme case, with a median close to $1{,}200$ in both groups but a treatment group standard deviation above $8{,}000$, which points to a small number of very large precincts. The average poverty level, in contrast, is balanced across arms and has low variance (means $4.59$ and $4.58$). 
The skew and heavy tails in both the outcome and several covariates are exactly the features that a median-based splitting criterion is meant to accommodate. The same data have also been analyzed for the rank-based estimation procedure of \cite{ghosh_robustness_2026}.

The ACTG 175 sample considered in Section \ref{sec:ACTG_175} is larger than the Progresa sample and contains $1{,}342$ patients, $321$ on zidovudine alone (control group) and $1{,}021$ on alternative antiretroviral therapies (treatment group) \citep{hammer_trial_1996}. Figure \ref{fig:outcome_distribution_actg175} illustrates the distribution of the outcome variable, while Table \ref{tab:desc_stats_actg175} shows descriptive statistics of the sample. 
The outcome, the CD4 count at 96 weeks, is right-skewed but less so than in Progresa (skewness $0.52$), with the mean above the median in both groups and a higher central level under alternative antiretroviral therapies (median of $330$ against $283$). 
The continuous and binary pre-treatment covariates are well balanced across arms, as expected under randomization, with means, medians, and proportions almost identical between the two groups. 
The clearest departure from symmetry among the covariates is prior antiretroviral use in days, whose mean is more than twice its median and whose standard deviation exceeds its mean in both groups, so it carries a long right tail. Baseline CD8 count shows a milder right skew. The application thus pairs a larger sample with a more moderate skew and complements the more strongly skewed Progresa data.

Figures~\ref{fig:plot_cate_progresa} and~\ref{fig:plot_cate_actg175} report the pointwise confidence intervals behind the CATE distributions of Section~\ref{sec:empappl}. 
In both applications the point estimates vary around the ATE, but the intervals are wide and almost every one contains the ATE. 
Neither application therefore supports treatment effect heterogeneity that is statistically distinguishable from a constant effect. The MSD intervals run somewhat wider than the MSE intervals, most clearly for Progresa, which reproduces the conservative width premium seen in the simulations. The agreement between the two splitting rules at the level of individual estimates and intervals matches their agreement on the ATE.
%
%
\begin{figure}[htbp]
  \centering
    \caption{Distribution of the outcome variable \texttt{pri2000s}, the PRI vote share in 2000, for the Progresa application in Section \ref{sec:progresa}.}
  \includegraphics[width=\textwidth]{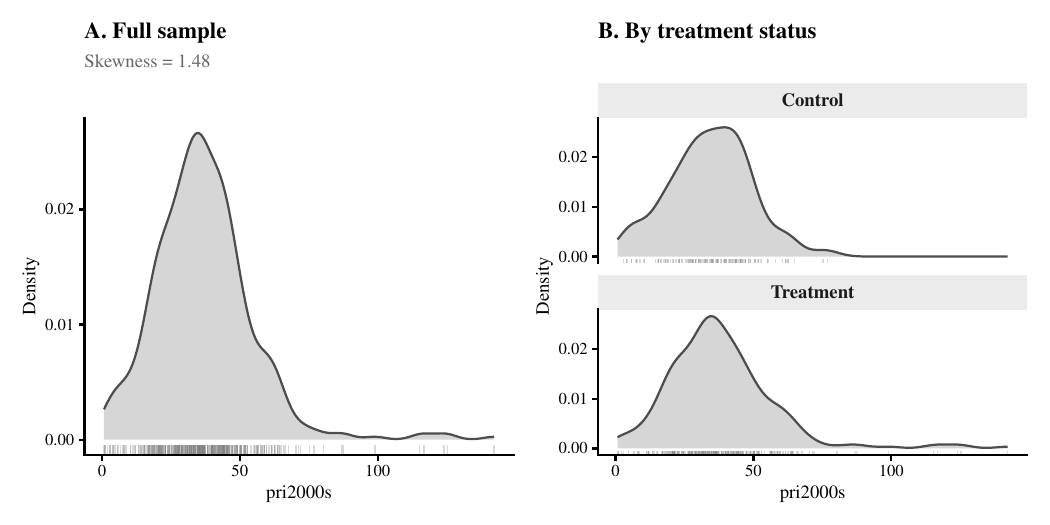}
  \label{fig:outcome_distribution_progresa}
  \begin{notes}
  Panel A (left) shows the full-sample density with a rug plot. Panel B (right) shows the density separately for the treatment and control groups. Both panels indicate pronounced right skewness.
  \end{notes}
\end{figure}
\begin{table}[htbp]
\centering
\caption{\label{tab:desc_stats_progresa}Descriptive statistics by treatment status for the Progresa application in Section \ref{sec:progresa}.}
\centering
\begin{tabular}[htbp]{lrrrrrr}
\toprule
\multicolumn{1}{c}{ } & \multicolumn{3}{c}{Control (late)} & \multicolumn{3}{c}{Treatment (early)} \\
\cmidrule(l{3pt}r{3pt}){2-4} \cmidrule(l{3pt}r{3pt}){5-7}
Variable & Mean & Median & SD & Mean & Median & SD\\
\midrule
\addlinespace[0.3em]
\multicolumn{7}{l}{\textbf{Outcome}}\\
\hspace{1em}PRI vote share, 2000 (\%) & 34.49 & 35.07 & 14.82 & 38.11 & 35.76 & 19.73\\
\addlinespace[0.3em]
\multicolumn{7}{l}{\textbf{Pre-treatment covariates}}\\
\hspace{1em}Avg. poverty level & 4.59 & 4.75 & 0.48 & 4.58 & 4.73 & 0.49\\
\hspace{1em}Total population (1994) & 1900.91 & 1306.00 & 2809.93 & 2104.89 & 1112.00 & 8270.35\\
\hspace{1em}Voter turnout (1994) & 411.97 & 388.00 & 225.50 & 374.56 & 356.00 & 218.49\\
\hspace{1em}PRI votes (1994) & 256.76 & 230.00 & 148.65 & 236.70 & 208.00 & 148.02\\
\hspace{1em}PAN votes (1994) & 38.33 & 19.50 & 48.92 & 32.06 & 15.00 & 40.70\\
\hspace{1em}PRD votes (1994) & 67.36 & 40.50 & 77.66 & 60.72 & 30.00 & 78.08\\
\hspace{1em}No. of villages & 6.28 & 5.00 & 4.01 & 5.90 & 5.00 & 3.72\\
\bottomrule
\end{tabular}
\begin{notes}
The table reports the mean, median, and standard deviation of the outcome and pre-treatment covariates by treatment status. Sample sizes are $N = 417$ (total), $N_0 = 138$ (control), $N_1 = 279$ (treatment).
\end{notes}
\end{table}
\begin{figure}[htbp]
  \centering
    \caption{Precinct-level CATE estimates for the Progresa data with pointwise 95\% confidence intervals.}
  \includegraphics[width=\textwidth]{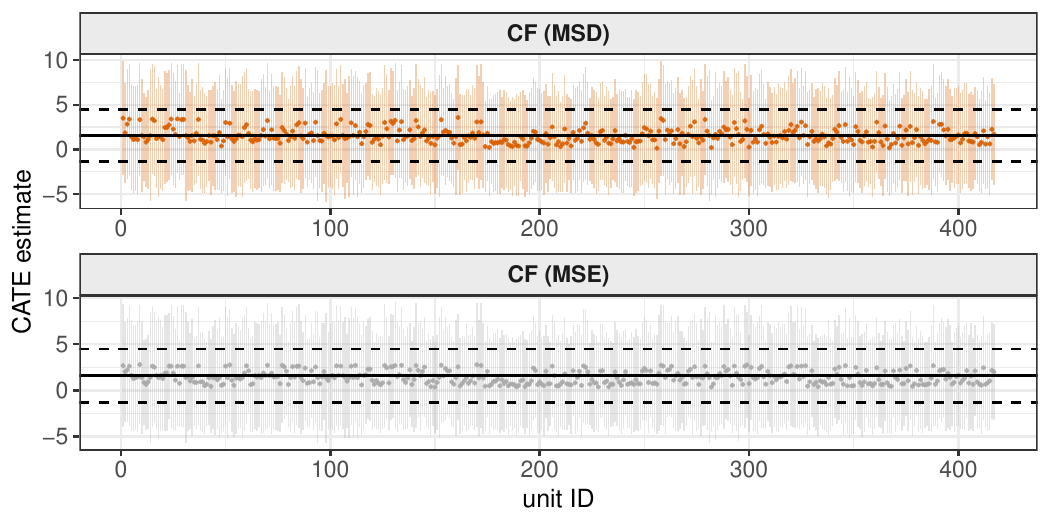}
  \label{fig:plot_cate_progresa}
  \begin{notes}
      Causal forests with MSD (orange, top panel) and MSE (grey, bottom panel) splitting. Precincts are ordered by identifier. The solid horizontal line is the ATE estimate and the dashed lines are its 95\% confidence interval. The pointwise intervals are wide and almost all contain the ATE.
  \end{notes}
\end{figure}
%
%
\begin{figure}[htbp]
  \centering
    \caption{Distribution of the outcome variable \texttt{cd496}, the CD4 count at 96 weeks, for the ACTG 175 application in Section \ref{sec:ACTG_175}.}
  \includegraphics[width=\textwidth]{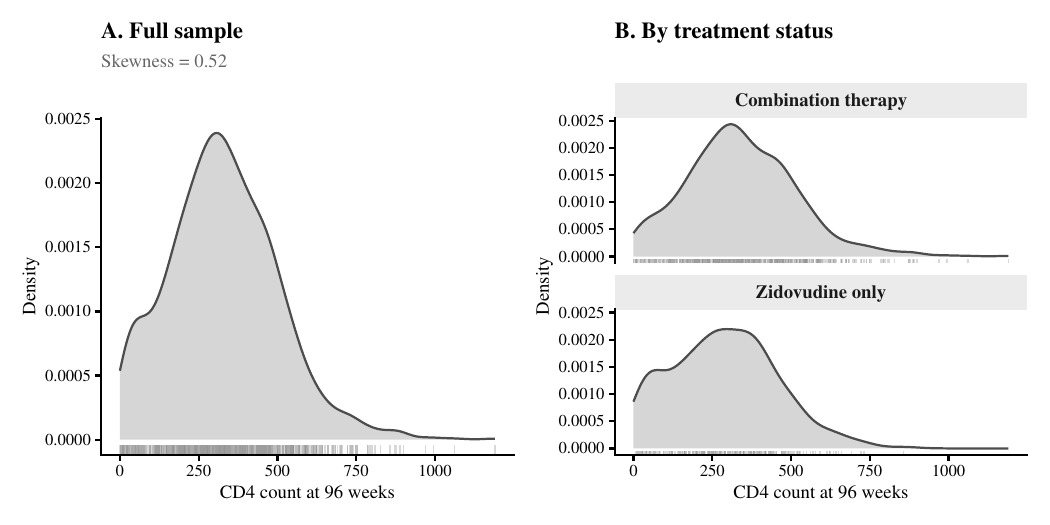}
  \label{fig:outcome_distribution_actg175}
  \begin{notes}
      Panel~A (left) shows the full-sample density with a rug plot. Panel~B (right) shows the density separately for the treatment and control groups. Both panels indicate pronounced right skewness.
  \end{notes}
\end{figure}
\begin{table}[htbp]
\centering
\caption{Descriptive statistics by treatment status for the ACTG 175 application in Section \ref{sec:ACTG_175}.}
\label{tab:desc_stats_actg175}
\begin{tabular}[htbp]{lrrrrrr}
\toprule
\multicolumn{1}{c}{ } & \multicolumn{3}{c}{Control (Zidovudine only)} & \multicolumn{3}{c}{Treatment (Alternative therapies)} \\
\cmidrule(l{3pt}r{3pt}){2-4} \cmidrule(l{3pt}r{3pt}){5-7}
Variable & Mean & Median & SD & Mean & Median & SD\\
\midrule
\addlinespace[0.3em]
\multicolumn{7}{l}{\textbf{Outcome}}\\
\hspace{1em}CD4 count at 96 weeks & 287.62 & 283.0 & 166.38 & 341.45 & 330.0 & 175.29\\
\addlinespace[0.3em]
\multicolumn{7}{l}{\textbf{Pre-treatment covariates}}\\
\hspace{1em}Baseline CD4 count & 365.12 & 354.0 & 115.51 & 351.80 & 340.0 & 117.22\\
\hspace{1em}Baseline CD8 count & 1009.38 & 890.0 & 503.59 & 991.58 & 909.0 & 474.53\\
\hspace{1em}Age (years) & 35.09 & 34.0 & 8.69 & 35.42 & 34.0 & 8.65\\
\hspace{1em}Weight (kg) & 75.14 & 73.6 & 11.93 & 74.67 & 74.1 & 12.77\\
\hspace{1em}Karnofsky score & 95.79 & 100.0 & 5.82 & 95.90 & 100.0 & 5.47\\
\hspace{1em}Prior antiretroviral use (days) & 392.37 & 166.0 & 475.29 & 385.80 & 161.0 & 461.80\\
\hspace{1em}Non-white & 0.26 & 0.0 & 0.44 & 0.26 & 0.0 & 0.44\\
\hspace{1em}Male & 0.82 & 1.0 & 0.39 & 0.84 & 1.0 & 0.36\\
\hspace{1em}Hemophilia & 0.07 & 0.0 & 0.26 & 0.08 & 0.0 & 0.27\\
\hspace{1em}Homosexual activity & 0.65 & 1.0 & 0.48 & 0.70 & 1.0 & 0.46\\
\hspace{1em}History of drug use & 0.09 & 0.0 & 0.29 & 0.12 & 0.0 & 0.32\\
\hspace{1em}Symptomatic & 0.16 & 0.0 & 0.36 & 0.18 & 0.0 & 0.39\\
\hspace{1em}Prior zidovudine use ($>$30 days) & 0.54 & 1.0 & 0.50 & 0.55 & 1.0 & 0.50\\
\bottomrule
\end{tabular}
\begin{notes} The table reports descriptive statistics (mean, median, and standard deviation) for the outcome and pre-treatment covariates by treatment status. Sample sizes are $N = 1342$ (total), $N_0 = 321$ (zidovudine only), and $N_1 = 1021$ (alternative antiretroviral therapies).
\end{notes}
\end{table}
\begin{figure}[htbp]
  \centering
    \caption{CATE estimates for the ACTG 175 data.}
  \includegraphics[width=\textwidth]{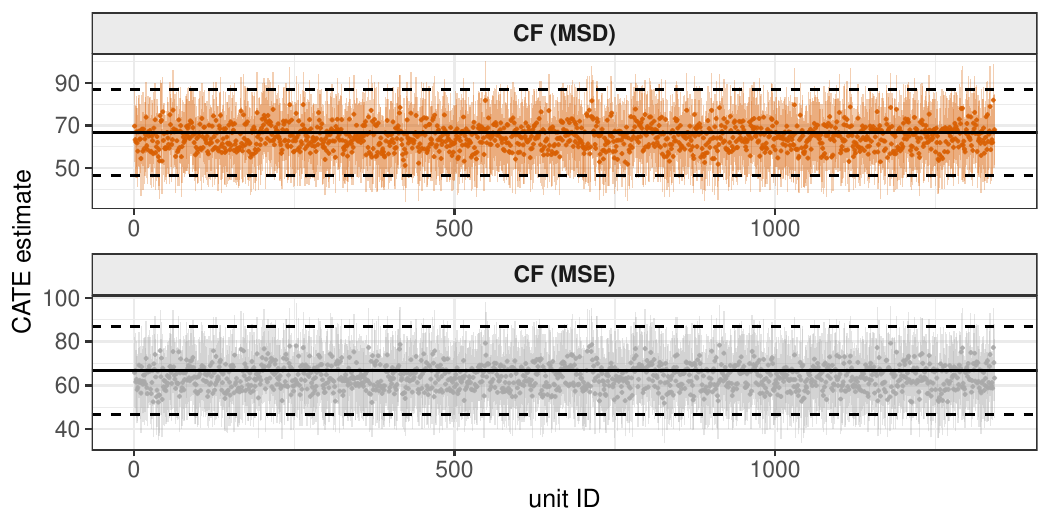} 
  \begin{notes}
  Individual-level CATE estimates for the ACTG 175 data with pointwise 95\% confidence intervals. Causal forests with MSD (orange, top panel) and MSE (grey, bottom panel) splitting criteria. Patients are ordered by identifier. The solid horizontal line is the ATE estimate and the dashed lines are its 95\% confidence interval. The pointwise intervals are wide and all contain the ATE.
  \end{notes}
  \label{fig:plot_cate_actg175}
\end{figure}

\clearpage

\section{Proofs}
\label{appendix:proofs}

This section presents the proofs of the results proposed in the main paper in chronological order.  We begin with Lemma \ref{lem:lemma_bounds}, which provides bounds for the distribution of $\hat{\tau}_{\text{HL}}(X_i)$. We then provide the proof of Theorem \ref{thm:theorem_consistency}, which extends the asymptotic consistency result of \citet{hoyland_HL} to the potential outcome setting. The subsequent proof of Theorem \ref{thm:theorem_median_unbiased} establish the median-unbiasedness under symmetry of the leafwise Wilcoxon rank-sum statistic.

\begin{proof} \label{proof:proof_lemma_4}
\noindent\textbf{of Lemma \ref{lem:lemma_bounds}.}
This proof adapts \citet{hodges_estimates_1963} to the leafwise conditional framework. Recall from
Definition~\ref{def:HL_general} that
\[
\Delta^* = \sup\{ \Delta : t(Y^0,Y^1 - \Delta) > \xi \},
\qquad
\Delta_{*} = \inf\{ \Delta : t(Y^0,Y^1 - \Delta) < \xi \},
\]
and that, by condition~(i) of Definition~\ref{def:HL_general}, $t(Y^0,Y^1+a)$ is non-decreasing in
$a$, equivalently $t(Y^0,Y^1-\Delta)$ is non-increasing in $\Delta$. Conditioning on the fixed leaf
$l(x;\Pi)$, the treated and control outcomes are i.i.d.\ draws from the continuous leafwise
conditionals $F_{1,l}$ and $F_{0,l}$, so the joint distribution of $(Y^0,Y^1)$ is continuous and
$\Delta^*$ and $\Delta_*$ have continuous distributions. In particular the boundary events have
probability zero, so strict and non-strict inequalities in $\Delta^*,\Delta_*$ may be interchanged
freely below.

By monotonicity, the definitions of $\Delta^*$ and $\Delta_*$ give, for any $a\in\mathbb R$, the
event equivalences
\[
\{\Delta_* < a\} = \{\, t(Y^0,Y^1-a) < \xi \,\},
\qquad
\{\Delta^* \ge a\} = \{\, t(Y^0,Y^1-a) > \xi \,\}.
\]
Taking complements in the second equivalence yields
$\{\Delta^* < a\} = \{\, t(Y^0,Y^1-a) \le \xi \,\}$, and hence
\[
Pr\bigl(\Delta_{*} < a\bigr) = Pr\bigl(t(Y^0,Y^1 - a) < \xi\bigr),
\qquad
Pr\bigl(\Delta^* < a\bigr) = Pr\bigl(t(Y^0,Y^1 - a) \le \xi\bigr).
\]
The strict inequality in the first identity and the non-strict inequality in the second are exactly
the two sides produced by the complement step. Finally, by Definition~\ref{def:HL_general},
\[
\hat\tau_{\mathrm{HL}}(x;\Pi)=\frac{\Delta^*+\Delta_*}{2}.
\]
Since \(t(Y^0,Y^1-\Delta)\) is non-increasing in \(\Delta\), we have
\(\Delta^*\leq \Delta_*\), and hence
\[
\Delta^*
\leq
\hat\tau_{\mathrm{HL}}(x;\Pi)
\leq
\Delta_*.
\]
Therefore,
\[
\{\Delta_*<a\}
\subseteq
\{\hat\tau_{\mathrm{HL}}(x;\Pi)<a\}
\subseteq
\{\Delta^*<a\}.
\]
Combining this with the probability identities above gives
\[
Pr\bigl(t(Y^0,Y^1-a)<\xi\bigr)
\leq
Pr\bigl(\hat\tau_{\mathrm{HL}}(x;\Pi)<a\bigr)
\leq
Pr\bigl(t(Y^0,Y^1-a)\leq\xi\bigr),
\]
which is the claimed bound.
\end{proof}

\begin{proof}
\label{proof:theorem_consistency}
\noindent\textbf{of Theorem \ref{thm:theorem_consistency}.}
This proof builds on Theorem 2.3 of \citet{hoyland_HL} and adapts it to the leafwise setting.
Under Assumption~\ref{ass:HL_shift}, we show that the asymptotic distribution of
$\hat\tau_{\mathrm{HL}}(x;\Pi)$ is centered around $\Delta(x;\Pi)$. In particular, we need to
show that for every fixed $\Delta(x;\Pi)$,
\begin{equation} \label{proof:goal_3_1}
\lim_{N_l \to \infty}
Pr\left( \sqrt{N_l} \left(\hat{\tau}_{\mathrm{HL}}(x;\Pi) - \Delta(x;\Pi) \right) \leq a \right)
=
\Phi\left( \frac{a B}{A} \right),
\end{equation}
with $a \in \mathbb{R}$, where $\Phi$ is the cumulative distribution function of the standard
normal distribution and $A$ and $B$ are given below. We distinguish between the leafwise
population objects and the corresponding observed samples. Let $l=l(x;\Pi)$ denote the leaf
containing $x$. By Assumption~\ref{ass:HL_shift}, the leafwise population conditional
distributions are
\[
F_{1,l}(z)=Pr\bigl(Y_i(1)\le z \mid X_i\in l\bigr),
\qquad
F_{0,l}(z)=Pr\bigl(Y_i(0)\le z \mid X_i\in l\bigr),
\]
and these satisfy the location-shift relation
\[
F_{1,l}(z)=F_{0,l}(z-\Delta(x;\Pi))
\qquad\text{for all } z\in\mathbb R.
\]
The corresponding observed leafwise samples are
\[
Y^1_l=\{Y_j \mid j\in\mathcal S_1,\; X_j\in l\},
\qquad
Y^0_l=\{Y_m \mid m\in\mathcal S_0,\; X_m\in l\},
\]
with sample sizes $N_{1,l}$ and $N_{0,l}$ and total leafwise sample size $N_l=N_{1,l}+N_{0,l}$.

By the shift-equivariance of the Hodges--Lehmann estimator, adding a constant to every treated
outcome shifts $\hat\tau_{\mathrm{HL}}(x;\Pi)$ by that same constant, so
$\hat\tau_{\mathrm{HL}}(x;\Pi)-\Delta(x;\Pi)$ has the same distribution as the estimator computed
on two samples that both follow $F_{0,l}$. It therefore suffices to derive the limiting law under
the null shift $\Delta(x;\Pi)=0$, at which $F_{1,l}=F_{0,l}$. We accordingly write
\begin{equation}\label{def:G}
\tilde F(z)\coloneqq F_{1,l}\bigl(z+\Delta(x;\Pi)\bigr) = F_{0,l}(z),
\end{equation}
for the common null distribution, and let $\tilde f=f_{0,l}$ denote its density. The argument $a$
of \eqref{proof:goal_3_1} enters only through the local sequence
\[
\Delta_N(x;\Pi)= -\frac{a}{\sqrt{N_l}},
\qquad\text{so that}\qquad
\Delta_N(x;\Pi)\to 0
\quad\text{as}\quad N_l \to \infty,
\]
which is introduced by the rank-statistic inversion of Lemma~\ref{lem:lemma_bounds} below. Recall
that $N_l \to \infty$ and $N_{0,l}/N_l\to\lambda\in(0,1)$.

We first prove that, for every $u\in\mathbb{R}$,
\begin{equation}\label{proof:help_3_1}
\lim_{N_l\to\infty}
Pr_N\left(
\sqrt{N_l}(t_N-\xi_N)\le u
\right)
=
\Phi\!\left(\frac{u+aB}{A}\right),
\end{equation}
where $Pr_N$ denotes probability under the local shift $\Delta_N(x;\Pi)$, and $\mathbb E_N$ denotes the expectation with respect to $Pr_N$. Furthermore, we define 
\[
t_N=t_N(Y^0_l,Y^1_l)\coloneqq \frac{\mathcal W_l}{N_{1,l}N_{0,l}},
\]
where $\mathcal W_l$ is the Wilcoxon statistic from Definition~\ref{defn: Mann-Whitney} computed on the
leafwise samples $Y^0_l$ and $Y^1_l$, and $\xi_N=\tfrac12$ is the symmetry point under the null.
Under $Pr_N$ the control sample follows $F_{0,l}$ and the treated sample follows
$F_{0,l}(\boldsymbol{\cdot} -\Delta_N(x;\Pi))$. Let \(\tilde Y_l^0\) and \(\tilde Y_l^1\) be independent generic leafwise
observations with $\tilde Y_l^0\sim F_{0,l}$ and 
$\tilde Y_l^1\sim F_{1,l,N}$. Then
\[
\mathbb E[\varphi(\tilde Y_l^0,\tilde Y_l^1)]
=
Pr_N(\tilde Y_l^0<\tilde Y_l^1)
=
\mathbb E_N[t_N].\] We decompose
\begin{align*}
\sqrt{N_l}\left(t_N(Y^0_l,Y^1_l)-\frac12\right)
&=Q_N+R_N,
\\
Q_N
&:=\sqrt{N_l}\left(t_N(Y^0_l,Y^1_l)-\mathbb E[\varphi(\tilde Y^0_l,\tilde Y^1_l)]\right),
\\
R_N
&:=\sqrt{N_l}\left(\mathbb E[\varphi(\tilde Y^0_l,\tilde Y^1_l)]-\frac12\right),
\end{align*}
where $Q_N$ is the centered (mean-zero) stochastic part and $R_N$ is the deterministic drift
induced by the local shift. We analyze $R_N$ and $Q_N$ separately, starting with $R_N$. Since
\[
\mathbb E[\varphi(\tilde Y^0_l,\tilde Y^1_l)]=Pr_N(\tilde Y^0_l<\tilde Y^1_l)
=\int F_{0,l}(u)\,f_{0,l}\bigl(u-\Delta_N(x;\Pi)\bigr)\,du,
\]
and since the change of variable $v=u-\Delta_N(x;\Pi)$ gives the exact null identity
\[
\int F_{0,l}\bigl(u-\Delta_N(x;\Pi)\bigr)\,f_{0,l}\bigl(u-\Delta_N(x;\Pi)\bigr)\,du
=\int F_{0,l}(v)\,f_{0,l}(v)\,dv=\frac12,
\]
we may write
\begin{align}
R_N
&=
\sqrt{N_l}
\left(
\int F_{0,l}(u)\,f_{0,l}\bigl(u-\Delta_N(x;\Pi)\bigr)\,du
-
\int F_{0,l}\bigl(u-\Delta_N(x;\Pi)\bigr)\,f_{0,l}\bigl(u-\Delta_N(x;\Pi)\bigr)\,du
\right).
\label{eq:RN_new}
\end{align}
Replacing $f_{0,l}(u-\Delta_N(x;\Pi))$ by its leading-order limit $f_{0,l}(u)=\tilde f(u)$, which
changes \eqref{eq:RN_new} only at order $o(1)$ because $\Delta_N(x;\Pi)\to 0$, we obtain
\begin{align*}
R_N
&=
\sqrt{N_l}
\left(
\int F_{0,l}(u)\,\tilde f(u)\,du
-
\int F_{0,l}\bigl(u-\Delta_N(x;\Pi)\bigr)\,\tilde f(u)\,du
\right)+o(1)
\\
&=
-a
\int
\frac{F_{0,l}\bigl(u+a/\sqrt{N_l}\bigr)-F_{0,l}(u)}{a/\sqrt{N_l}}
\,\tilde f(u)\,du+o(1)
\;\longrightarrow\;
-a\underbrace{\int f_{0,l}(u)\,\tilde f(u)\,du}_{=:B}
=
-aB,
\end{align*}
with $B=\int f_{0,l}(u)\tilde f(u)\,du=\int f_{0,l}^2(u)\,du>0$.

For the second part, we show that $\lim_{N_l\to\infty}Pr_N(Q_N\le u)=\Phi(u/A)$, where
\[
A^2
=
\frac{1}{\lambda}\int \tilde F^2(u)\,f_{0,l}(u)\,du
+
\frac{1}{1-\lambda}\int F_{0,l}^2(u)\,\tilde f(u)\,du
-
\frac{1}{4\lambda(1-\lambda)},
\]
with $\lambda$ as in condition~(iii) of Assumption~\ref{ass:HL_shift}. The two integrals are the
Hájek projections of the two-sample U-statistic onto the control and treated arms, weighted by the
inverse arm proportions $1/\lambda$ and $1/(1-\lambda)$. Conditioning on the fixed leaf, the
leafwise treated and control outcomes are i.i.d.\ from continuous laws that coincide under the
null, so $t_N(Y^0_l,Y^1_l)$ is a two-sample (generalized) U-statistic
\citep{hollander_nonparametric_2014}, and the central limit argument of \citet{hoyland_HL} carries
over verbatim to the leafwise samples. Since $\Delta_N(x;\Pi)\to 0$ is a contiguous local sequence,
the limiting variance equals its null value, where $\tilde F=F_{0,l}$ and $\tilde f=f_{0,l}$ give
$\int \tilde F^2 f_{0,l}\,du=\int F_{0,l}^2\tilde f\,du=\int F_{0,l}^2\,dF_{0,l}=\tfrac13$, so that
\[
A^2
=\left(\frac{1}{\lambda}+\frac{1}{1-\lambda}\right)\frac13-\frac{1}{4\lambda(1-\lambda)}
=\frac{1}{12\,\lambda(1-\lambda)},
\]
the classical Wilcoxon constant. Applying Slutsky's theorem to $Q_N+R_N$ gives
\[
\lim_{N_l \to \infty}
Pr_N\left(
\sqrt{N_l}(t_N-\xi_N)\le u
\right)
=
\Phi\!\left(\frac{u+aB}{A}\right),
\]
which is \eqref{proof:help_3_1}.

Using Lemma~\ref{lem:lemma_bounds}, we now establish \eqref{proof:goal_3_1}. By shift-equivariance
the probability is evaluated under the null, and by the monotonicity of the rank-statistic
inversion,
\begin{align*}
&\lim_{N_l \to \infty}
Pr\left(
\sqrt{N_l}\bigl(\hat\tau_{\mathrm{HL}}(x;\Pi)-\Delta(x;\Pi)\bigr)\le a
\right)
\\
={}&
\lim_{N_l \to \infty}
Pr\left(
\hat\tau_{\mathrm{HL}}(x;\Pi)\le \frac{a}{\sqrt{N_l}}
\;\middle|\;
\Delta(x;\Pi)=0
\right)
&&\text{(shift-equivariance)}
\\
={}&
\lim_{N_l \to \infty}
Pr\left(
t_N\!\left(Y^0_l,\,Y^1_l-\frac{a}{\sqrt{N_l}}\right)\le \xi_N
\;\middle|\;
\Delta(x;\Pi)=0
\right)
&&\text{(Lemma~\ref{lem:lemma_bounds})}
\\
={}&
\lim_{N_l\to\infty}
Pr_N(t_N-\xi_N\le 0)
\;=\;
\Phi\!\left(\frac{aB}{A}\right),
\end{align*}
where the last equality applies \eqref{proof:help_3_1} with $u=0$.

Hence, under Assumption~\ref{ass:HL_shift}, the asymptotic distribution of
$\hat\tau_{\mathrm{HL}}(x;\Pi)$ is centered at the leaf-specific location shift $\Delta(x;\Pi)$,
which establishes consistency and completes the proof. In particular,
\eqref{proof:goal_3_1} states that
$\sqrt{N_l}\bigl(\hat\tau_{\mathrm{HL}}(x;\Pi)-\Delta(x;\Pi)\bigr)$ is asymptotically
$N\!\bigl(0,(A/B)^2\bigr)$ with
$(A/B)^2=\bigl[12\,\lambda(1-\lambda)\,(\int f_{0,l}^2)^2\bigr]^{-1}$, the classical
Hodges--Lehmann asymptotic variance.
\end{proof}

\begin{proof} \label{proof:theorem_median_unbiased}
\noindent\textbf{of Theorem  \ref{thm:theorem_median_unbiased}.}
Suppose first that $N_{1,l}N_{0,l}$ is odd, so that $N_{1,l}N_{0,l}=2m+1$ for some integer $m$. Then
\[
\xi_l=m+\frac12,
\]
which is not an attainable value of the statistic. Hence
\[
Pr\bigl(t(Y^0_l,Y^1_l)=\xi_l \mid \Delta(x;\Pi)=0\bigr)=0.
\]
By Lemma~\ref{lem:median_unbiased}, it follows that
\[
Pr\!\left(\hat\tau_{\mathrm{HL}}(x;\Pi)\le \Delta(x;\Pi)\right)=\frac12,
\]
so the estimator is exactly median unbiased in this case. Now suppose that $N_{1,l}N_{0,l}$ is even, say $N_{1,l}N_{0,l}=2m.$ Then
\[
\xi_l=m
\]
is an attainable value of the statistic. In this case the HL estimator is given by the average of
the two middle ordered pairwise differences,
\[
\hat\tau_{\mathrm{HL}}(x;\Pi)=\frac{U^{(m)}+U^{(m+1)}}{2},
\]
and therefore exact equality need not hold. However, Lemma~\ref{lem:median_unbiased} implies that
the probability of over- and underestimation remains close to $\frac12$. Hence the estimator is
approximately median unbiased when $N_{1,l}N_{0,l}$ is even.
\end{proof}

\begin{proof}\label{proof:median_unbiased}
\noindent\textbf{of Lemma \ref{lem:median_unbiased}.}
For the shift parameter $\Delta$, the bound is established in \citet{hodges_estimates_1963}.
Conditioning on the fixed leaf $l(x;\Pi)$,  randomized treatment assignment implies that observed outcomes are i.i.d.\ draws from the continuous leafwise conditionals $F_{1,l}$ and $F_{0,l}$, which coincide under the null shift. Therefore, the leafwise
Wilcoxon statistic is distribution-free and symmetric about $\xi_l$. The same argument therefore applies to the leafwise samples.
\end{proof}

\begin{proof}\label{proof:Delta_equals_tau_leaf} \noindent\textbf{of Lemma \ref{lemma:Delta_equals_tau_leaf}.}
Under the stated condition, the conditional distribution of $Y_i(1)$ given $X_i \in l(x;\Pi)$ is
the conditional distribution of $Y_i(0)$ shifted by the constant $\tau(l;\Pi)$. Hence
\[
F_{1,l}(z)=F_{0,l}(z-\tau(l;\Pi))
\qquad \text{for all } z\in\mathbb R.
\]
By the definition of the leafwise location shift, it follows that
\[
\Delta(x;\Pi)=\tau(l;\Pi).\]
Moreover, the constant-effect condition implies that

\[Y_i(1)-Y_i(0)=\tau(l;\Pi) \qquad \text{almost surely given } X_i\in l(x;\Pi).\]

Therefore the leafwise distribution of the individual treatment effect is degenerate at $\tau(l;\Pi)$, and hence

\[\tau_{\operatorname{med}}(x;\Pi)= \operatorname{med}\bigl(Y_i(1)-Y_i(0)\mid X_i\in l(x;\Pi)\bigr)
=\tau(l;\Pi).
\]
Hence,

\[\Delta(x;\Pi)=\tau(l;\Pi)=\tau_{\operatorname{med}}(x;\Pi),\]
as claimed.
\end{proof}

\begin{proof}\label{proof:Delta_equals_tau_equals_tauMed}\noindent\textbf{of Proposition~\ref{prop:Delta_equals_tau_equals_tauMed}.}
Assume the conditional means exist. By the location-shift model in condition~(i), the marginal
conditional distributions satisfy $F_{1,l}(z)=F_{0,l}(z-\Delta(x;\Pi))$, so
\[
\mathbb E[Y_i(1)\mid X_i\in l]
=\mathbb E[Y_i(0)\mid X_i\in l]+\Delta(x;\Pi),
\]
and therefore
\[
\tau(l;\Pi)
=\mathbb E[Y_i(1)-Y_i(0)\mid X_i\in l]
=\Delta(x;\Pi).
\]
By condition~(ii), the conditional distribution of
$Y_i(1)-Y_i(0)\mid X_i\in l$ is symmetric. A distribution that is symmetric about a point and
has a finite mean is symmetric about that mean, and its median equals its mean. Hence
\[
\tau_{\operatorname{med}}(x;\Pi)
=\operatorname{med}\bigl(Y_i(1)-Y_i(0)\mid X_i\in l\bigr)
=\mathbb E[Y_i(1)-Y_i(0)\mid X_i\in l]
=\tau(l;\Pi).
\]
Combining the two equations gives
\[
\Delta(x;\Pi)=\tau(l;\Pi)=\tau_{\operatorname{med}}(x;\Pi),
\]
as claimed.
\end{proof}

\end{document}